\documentclass[superscriptaddress,nofootinbib,notitlepage,aps,pra,10pt]{revtex4-1}
\pdfoutput=1

\usepackage{silence}
\usepackage{graphicx}
\usepackage{amsthm}
\usepackage{thmtools}
\usepackage{mathtools}
\usepackage{thm-restate}
\usepackage{amsmath}
\usepackage{algorithm}
\usepackage{algorithmic}
\usepackage{amssymb}
\usepackage{bm}
\usepackage{xcolor}
\usepackage{placeins}
\usepackage[normalem]{ulem}
\usepackage[colorlinks]{hyperref}
\hypersetup{
	pdfstartview={FitH},
	pdfnewwindow=true,
	colorlinks=true,
	linkcolor=blue,
	citecolor=blue,
	filecolor=blue,
	urlcolor=blue}
\usepackage{tikz}
\usetikzlibrary{calc,backgrounds}
\usepackage{physics}
\usepackage{cancel}
\usepackage{multirow}
\usepackage{pgffor}
\usepackage{url}
\usepackage{hypcap}
\usepackage{dsfont}
\usepackage{soul}
\usepackage{inconsolata}

\usepackage{hyperref}
\hypersetup{pdfpagemode=UseNone}

\allowdisplaybreaks

\newcommand{\thm}[1]{\hyperref[thm:#1]{Theorem~\ref*{thm:#1}}}
\newcommand{\defn}[1]{\hyperref[defn:#1]{Definition~\ref*{defn:#1}}}
\newcommand{\lem}[1]{\hyperref[lem:#1]{Lemma~\ref*{lem:#1}}}
\newcommand{\prop}[1]{\hyperref[prop:#1]{Proposition~\ref*{prop:#1}}}
\newcommand{\fig}[1]{\hyperref[fig:#1]{Figure~\ref*{fig:#1}}}
\newcommand{\tab}[1]{\hyperref[tab:#1]{Table~\ref*{tab:#1}}}
\renewcommand{\sec}[1]{\hyperref[sec:#1]{Section~\ref*{sec:#1}}}
\newcommand{\app}[1]{\hyperref[app:#1]{Appendix~\ref*{app:#1}}}
\newcommand{\cor}[1]{\hyperref[cor:#1]{Corollary~\ref*{cor:#1}}}
\newcommand{\obs}[1]{\hyperref[obs:#1]{Observation~\ref*{obs:#1}}}
\newcommand{\nn}{\nonumber \\}
\newcommand{\append}[1]{\hyperref[append:#1]{Appendix~\ref*{append:#1}}}

\renewcommand{\ket}[1]{|#1\rangle}
\renewcommand{\bra}[1]{\langle #1|}

\usepackage[capitalise]{cleveref} 

\newtheorem{theorem}{Theorem}
\newtheorem{definition}[theorem]{Definition}
\newtheorem{conjecture}[theorem]{Conjecture}
\newtheorem{lemma}[theorem]{Lemma}

\newtheorem{corollary}[theorem]{Corollary}
\newtheorem{problem}[theorem]{Problem}

\input{Qcircuit}

\newcommand{\MQ}{\affiliation{
School of Mathematical and Physical Sciences,
Macquarie University, Sydney, NSW 2109, AU} }

\newcommand{\UMD}{\affiliation{
Joint Center for Quantum Information and Computer Science,\\ University of Maryland, College Park, MD 20742, USA}}          

\newcommand{\BBQP}{\affiliation{
BosonQ Psi (BQP) Corp., New York, USA}}

\newcommand{\PKU}{\affiliation{
Beijing International Center for Mathematical Research, Peking University, Beijing, China}}

\begin{document}
\title{Large time-step discretisation of adiabatic quantum dynamics}
\date{\today}
\author{Dong An} \PKU\UMD
\author{Pedro C.~S.~Costa} \BBQP\MQ
\author{Dominic W.~Berry} \MQ

\begin{abstract}
    Adiabatic quantum computing is a general framework for preparing eigenstates of Hamiltonians on quantum devices. However, its digital implementation requires an efficient Hamiltonian simulation subroutine, which may introduce extra computational overhead or complicated quantum control logic. In this work, we show that the time step sizes in time discretization can be much larger than expected, and the overall complexity is greatly reduced. Remarkably, regardless of the general convergence order of the numerical method, we can choose a uniform time step size independent of tolerated error and evolution time for sufficiently accurate simulation. Furthermore, with the boundary cancellation condition where the continuous diabatic errors are exponentially suppressed, we provide strong evidence on an exponential convergence of even first-order Trotter with uniform time step size. We apply our analysis to the example of adiabatic unstructured search and show several preferable features of the Trotterized adiabatic approach: it can match the Grover lower bound, it does not require \emph{a priori} knowledge on the number of marked states, and its performance can be asymptotically comparable with that of the quantum approximate optimization algorithm. 
\end{abstract}

\maketitle

\tableofcontents

\section{Introduction}

Preparing eigenstates of Hamiltonians is a fundamental task in quantum physics, quantum chemistry, and quantum information. Classical algorithms have a cost that scales polynomially with the dimension of the Hilbert space, which motivates the exploration of quantum algorithms for potential speedup. 
An approach to solving the eigenstate preparation problem is adiabatic quantum computing (AQC)~\cite{RevModPhys.90.015002}. 
In the framework of AQC, we consider the dynamics of a time-dependent Hamiltonian given by 
\begin{equation}\label{eqn:AQC_dynamics}
    i \partial_t \ket{\psi(t)} = H(t/T) \ket{\psi(t)}, \qquad 0 \leq t \leq T, 
\end{equation}
where $H(s) = (1-f(s)) H_0 + f(s) H_1$ is an interpolating Hamiltonian that smoothly evolves from an initial Hamiltonian $H_0$ to a target Hamiltonian $H_1$, and $f(s)$ is called the scheduling function such that $f(0) = 0$ and $f(1) = 1$. 
Typically, $H_0$ is designed to be simple with easily preparable eigenstates. 
Then, under the gap condition that the desired eigenpath of $H(s)$ is separated from the rest of the spectrum, the dynamics of~\cref{eqn:AQC_dynamics} will approximately drive the system from the eigenstate of $H_0$ to that of $H_1$, as $T \rightarrow \infty$.

While AQC is a promising technique for solving general eigenstate preparation problems, its implementation on digital quantum devices requires an efficient subroutine for simulating~\cref{eqn:AQC_dynamics}, which may introduce non-optimal asymptotic scalings. 
Commonly used methods include the first-order exponential integrator~\cite{HochbruckOstermann2010,SchifferTuraCirac2022} and the first-order product formula~\cite{farhi2000quantum,vanDamMoscaVazirani2001,ChildsSuTranEtAl2019,AnFangLin2021}. 
Both methods, despite their simplicity, have only first-order convergence, and thus make the number of discrete evolution steps super-linear in the evolution time $T$ (for first-order methods, it is specifically quadratic in $T$). 
This can be expensive since $T$ is usually large to minimize diabatic excitation in AQC dynamics. 
Alternatively, high-order time-dependent Hamiltonian simulation algorithms, such as the high-order Trotter-Suzuki formula~\cite{suzuki1993general,ChildsSuTranEtAl2019} and the truncated Dyson series method~\cite{KieferovSchererBerry2019,LowWiebe2019}, can reduce the number of evolution steps to almost linear in $T$. 
However, these methods may still introduce extra computational complication. 
The number of exponentially many local operators in the product formulae may grow rapidly as the convergence order increases, and the truncated Dyson series method requires complicated quantum control logic~\cite{SuBerryWiebeEtAl2021,RajputRoggeroWiebe2022}. 

In this work, we develop general methodology of analyzing numerical discretization methods for digitally simulating near adiabatic dynamics and obtain greatly improved error bounds. 
According to our analysis, the time step sizes in discretizing near-adiabatic dynamics can be much larger than expected, and the overall complexity is greatly reduced. 
Specifically, we consider the first-order exponential integrator and the product formula for time discretization. 
We show that, as long as the time step size is smaller than a uniform threshold independent of tolerated error $\epsilon$ and evolution time $T$, time discretization errors will be non-dominant and the digital simulation can achieve sufficient accuracy. 
Furthermore, under an additional assumption called the boundary cancellation condition, we show that in practice even a first-order product formula with uniform time step size can achieve exponential convergence, which achieves an exponential speedup in precision over previous results.

Our analysis is based on the discrete adiabatic theorem, which are discrete analog of the continuous adiabatic theorems and guarantee that sequentially applying a set of slowly varying gapped unitary walk operators can approximately drive the system from the eigenstate of the initial walk operator to the connected eigenstate of the final walk operator. 
In our analysis, we treat the numerical integrators as the set of discrete walk operators and show that the discretized dynamics indeed follows an adiabatic path with the same initial and final eigenstates as the continuous one but different in the middle. 
This not only allows us to obtain aforementioned improved error bounds, but also shows that the discretized evolution can potentially solve eigenstate preparation problems which are beyond the reach of the original continuous formulation. 
For example, if the Hamiltonian $H(s)$ becomes gapless, then the continuous AQC based on~\cref{eqn:AQC_dynamics} cannot follow the desired eigenpath. 
However, we find that, even for a gapless Hamiltonian, it is possible that the first-order Trotterized operator with large time step size satisfies the gap condition, so the corresponding discretized evolution may still succeed in preparing the desired eigenstates. 

As an application, we use our analysis to design a discrete adiabatic algorithm for the unstructured search problem with multiple marked items. 
Our algorithm has several preferable features: it can match the Grover lower bound, it does not require \emph{a priori} knowledge on the number of marked states, and it can achieve a logarithmic scaling in the inverse precision. 
Additionally, we contribute to the active discussions on the relation between AQC and quantum approximate optimization algorithm (QAOA)~\cite{FarhiGoldstoneGutmann2014}, by showing that our discrete-adiabatic-based algorithm gives a set of QAOA angles which can achieve the optimal query complexity $\mathcal{O}(\sqrt{N})$ for finding the single marked state among $N$ states. 

There have been a few recent efforts on improving the analysis of discretizing AQC. 
For instance, \cite{Yi2021} demonstrates the robustness of digital adiabatic simulation using first-order Trotter formula with large time step sizes by applying continuous adiabatic theorems to an effective Hamiltonian. 
Additionally, a more recent work~\cite{kocia2022digital} proves a similar result by relating the first-order Trotter error estimate to its second-order counterpart. 
Our results improve over prior work in the following ways: 
\begin{enumerate}
    \item Our analysis is applicable to the first-order exponential integrator and product formulae of any order. 
    As a comparison, previous works~\cite{Yi2021,kocia2022digital} only focus on the first-order Trotter formula, and their analysis heavily relies on its specific structure. 
    \item Under the boundary cancellation condition, our analysis can give an exponentially better complexity. 
    Specifically, our analysis explains why first-order methods with large time step sizes are very likely to have an exponential time convergence in this case, while previous works~\cite{Yi2021,kocia2022digital} still only show a linear convergence as it is unclear how their approaches can be naturally combined with the additional boundary cancellation condition. 
    \item Our analysis uncovers the intrinsic property of time discretization of the AQC as following another (discrete) adiabatic path. 
    This is not only conceptually different from the standard numerical analysis which views time discretization only as introducing numerical errors and driving the dynamics away from the adiabatic path, but also the key for revealing new findings. 
    For example, with the discrete adiabatic point of view, we are able to show that Trotterized adiabatic evolution may still succeed even for a gapless Hamiltonian, but all the previous works~\cite{Yi2021,kocia2022digital} only apply to gapped Hamiltonians. 
\end{enumerate}

The rest of this paper is organized as follows. 
We first summarize our main results in~\cref{sec:summary}. 
In~\cref{sec:analysis}, we show our improved complexity estimates based on the discrete adiabatic theorem in the general case, and~\cref{sec:analysis_boundary_cancellation} is devoted to the case with the boundary cancellation condition. 
\cref{sec:gap} discusses the possibility of using the first-order product formula with large step for gapless Hamiltonians. 
In~\cref{sec:Grover}, we present our results for robust adiabatic Grover search and compare it with QAOA. 
\cref{sec:conclusion} summarizes this work and discusses potential further work.

\section{Summary of the results}\label{sec:summary}

\subsection{Setup and numerical methods}

We first formally restate the state preparation problem. 
\begin{problem}[State preparation]\label{prob:state_prep}
    Let $H_1$ be a Hamiltonian and $\ket{\phi}$ be a target eigenstate of $H_1$. 
    For $\varepsilon > 0$, the goal of the state preparation problem is to prepare a quantum state $\ket{\widetilde{\phi}}$ such that $\|\ket{\widetilde{\phi}}\bra{\widetilde{\phi}} - \ket{\phi}\bra{\phi}\| \leq \varepsilon$. 
\end{problem}
In the framework of AQC, given another simple Hamiltonian $H_0$ with known eigenstate, we consider the dynamics of~\cref{eqn:AQC_dynamics}
with interpolating Hamiltonian
\begin{equation}\label{eq:schedule}
    H(s) = (1-f(s)) H_0 + f(s) H_1 \, ,
\end{equation}
and the scheduling function $f(s)$. 
Under the gap condition, starting from the eigenstate of $H_0$, the solution of~\cref{eqn:AQC_dynamics} at the final time will be a good approximation of the corresponding eigenstate of $H_1$.

The AQC dynamics of~\cref{eqn:AQC_dynamics} needs further time discretization to be implemented on a digital quantum device. 
The most commonly used strategy for numerical propagation is to divide the time interval into $T_{d}$ equidistant sub-intervals and approximate the exact local evolution operator by a local unitary numerical propagator. 
Specifically, a $p$-th order local unitary evolution operator $U_{\text{num}}$ is constructed such that for any $0 \leq j \leq T_{d}-1$, it approximates the time-ordering operator from $jh$ to $(j+1)h$ as follows
\begin{equation}
    \mathcal{T}e^{-i \int_{jh}^{(j+1)h} H(t/T) dt} = U_{\text{num}}((j+1)h,jh) + \mathcal{O}(h^{p+1}), 
\end{equation}
and the global exact evolution operator can be approximated as 
\begin{equation}\label{eqn:AQC_numerical}
    \mathcal{T}e^{-i \int_0^T H(t/T) dt} \approx \prod_{j=0}^{T_{d}-1} U_{\text{num}}((j+1)h,jh)
\end{equation}
where $T_{d}$ is the number of discrete time steps, and $h = T/T_{d}$ is the step size. 

In practice, there are two common classes of numerical methods, namely the first-order exponential integrator and the product formulae. 
They differ in the level of time discretisation. 
If we only approximate the time-ordering operator by a time-independent Hamiltonian simulation problem, we obtain the first-order exponential integrator, which is given by 
\begin{equation}\label{eqn:exponential_integrator_1st}
    U_{\exp}(t+h,t) =  \exp\left(-ihH(t/T)\right).
\end{equation}
In this method, each propagation step only involves one matrix exponential but the exponent is still an interpolation of the Hamiltonians $H_0$ and $H_1$. 
If we further split the evolution of $H(t/T)$ into exponentials of $H_0$ and $H_1$ separately, we obtain the product formulae. 
The first- and second-order product formulae are given as
\begin{equation}\label{eqn:1st_trotter}
    U_{\text{pf},1}(t+h,t) =  \exp\left(-i h f(t/T) H_1\right)\exp\left(-i h (1-f(t/T)) H_0\right), 
\end{equation}
and 
\begin{equation}\label{eqn:2nd_trotter}
     U_{\text{pf},2}(t+h,t) =  \exp\left(-i h/2 (1-f((t+h/2)/T)) H_0\right)\exp\left(-i h f((t+h/2)/T) H_1\right)\exp\left(-i h/2 (1-f((t+h/2))/T) H_0\right). 
\end{equation}
Arbitrary $p$-th order formulae can be constructed via several approaches, such as Trotter-Suzuki recursion~\cite{suzuki1993general,WiebeBerryHoyerEtAl2010}, Yoshida method~\cite{Yoshida1990} and their improved versions~\cite{MoralesCostaBurgarthEtAl2022,BlanesCasasMurua2024}. 
The resulting $p$-th order formula can be written in the general form as
\begin{equation}\label{eqn:high_order_trotter_general}
    U_{\text{pf},p}(t+h,t) = \prod_{k=0}^{K_p} \exp\left(  -i\beta_{p,k} h f((t+\delta_{p,k}h)/T) H_1\right)  \exp\left( -i\alpha_{p,k}h (1-f((t+\gamma_{p,k}h)/T)) H_0 \right)
\end{equation}
for some real parameters $\alpha_{p,k},\beta_{p,k},\gamma_{p,k},\delta_{p,k}$. 
We will also consider a simplified higher-order product formula 
\begin{equation}\label{eqn:high_order_trotter_simplified}
    U_{\text{spf},p}(t+h,t) = \prod_{k=0}^{K_p} \exp\left(  -i\beta_{p,k} h f(t/T) H_1\right)  \exp\left( -i\alpha_{p,k}h (1-f(t/T)) H_0 \right). 
\end{equation}
Equation \eqref{eqn:high_order_trotter_simplified} is obtained by replacing all the discrete time points~\cref{eqn:high_order_trotter_general} with the fixed initial time. 
If the function $f$ is a constant function (which is not the case for the adiabatic dynamics), then~\cref{eqn:high_order_trotter_simplified} represents a $p$-th order time-independent product formula.

\subsection{Main results}

\subsubsection{Improved complexity}

In this work, we propose a general framework for analyzing algorithms to simulate near adiabatic dynamics and obtain greatly improved error bounds in~\cref{sec:analysis}. 
Our results and a comparison with previous results are summarized in~\cref{tab:main_result_2}. 

\begin{table}[t]
    \centering
    \begin{tabular}{cc|cc}\hline
        Method & Error bound & Time step size & Scaling \\\hline 
        First-order exponential integrator & Standard & $\mathcal{O}(\epsilon \alpha^{-1})$ & $\mathcal{O}(\alpha^{3}\Delta_*^{-3}\epsilon^{-2})$ \\
        First-order product formula & Standard & $\mathcal{O}(\epsilon \alpha^{-2}T^{-1})$ & $\mathcal{O}(\alpha^{6}\Delta_*^{-6}\epsilon^{-3})$ \\
        $p$-th order product formula & Standard & $\mathcal{O}(\epsilon^{1/p}\alpha^{-1-1/p}T^{-1/p})$ & $\mathcal{O}(\alpha^{3+3/p}\Delta_*^{-3-3/p}\epsilon^{-1-2/p})$ \\ \hline
        First-order exponential integrator & \cref{cor:exponential_linear} & $\alpha^{-1}$ & $\mathcal{O}(\alpha^3\Delta_*^{-3}\epsilon^{-1})$ \\
        First-order product formula & \cref{cor:1st_2nd_product_formula} & $\mathcal{O}(\min\{\alpha^{-1}, \Delta_*^{1/2}\widetilde{\alpha}_2^{-1/2}\})$ & $\mathcal{O}(\alpha^3\Delta_*^{-3}\epsilon^{-1} \max\{1,\widetilde{\alpha}_2^{1/2}\Delta_*^{-1/2}\alpha^{-1}\})$ \\
        $p$-th order simplified product formula & \cref{cor:Trotter_linear} & $\mathcal{O}(\min\{\alpha^{-1},\Delta_*^{1/p}\widetilde{\alpha}_p^{-1/p}\})$ & $\mathcal{O}(\alpha^3\Delta_*^{-3}\epsilon^{-1} \max\{ 1,\widetilde{\alpha}_p^{1/p}\Delta_*^{-1/p} \alpha^{-1} \})$ \\\hline
    \end{tabular}
    \caption{Summary of the complexity estimates on discretizing~\cref{eqn:AQC_dynamics}. 
    Query complexity is measured by the number of time propagation steps required in the evolution. 
    We assume the gap condition only for the time-dependent Hamiltonian $H(t)$, and that both $|f'(s)|$ and $|f''(s)|$ are uniformly bounded. 
    Here $\alpha = \|H_0\|+\|H_1\|$, $\widetilde{\alpha}_p = \sum_{\gamma_0,\cdots,\gamma_p\in\left\{0,1\right\}} \|[H_{\gamma_p}, \cdots, [H_{\gamma_1},H_{\gamma_0}]]\|$, $T$ is the evolution time, $\epsilon$ is the tolerated error, and $\Delta_*$ represents the minimum gap of $H(s)$. 
   }
    \label{tab:main_result_2}
\end{table}

We first state our results for the first-order exponential integrator. 
In~\cref{sec:complexity_exponential_integrator}, we demonstrate that this first-order method can utilize a time step size as large as $\alpha^{-1}$ without introducing significant discretization errors, where $\alpha = \|H_0\|+ \|H_1\|$. 
In comparison, the standard error bound typically requires a time step size of $\mathcal{O}(\epsilon\alpha^{-1})$, where $\epsilon$ is the tolerated level of errors (see~\cref{app:standard_analysis}). 
Therefore, our analysis greatly increases the allowed time step size for sufficiently accurate simulation by fully excluding the error dependence in the time step size. 
Since the overall evolution time $T$ should be $\mathcal{O}(\alpha^2\Delta_*^{-3}\epsilon^{-1})$ to mitigate diabatic errors~\cite{jansen2007bounds}, where $\Delta_*$ is the minimum gap of $H(s)$, our result reduces the number of the total propagation steps from $\mathcal{O}(\alpha^3 \Delta_*^{-3} \epsilon^{-2})$ to $ \mathcal{O}(\alpha^3 \Delta_*^{-3} \epsilon^{-1})$.

The main theoretical tool in our analysis is the discrete adiabatic theorem, which can be viewed as a discrete analog of the continuous adiabatic theorem. 
The theorem was first established in~\cite{DKS98}, and an improved version with explicit gap dependence was introduced in the recent work~\cite{CostaAnYuvalEtAl2022}. 
In the discrete adiabatic theorem, we are provided with a set of slowly varying unitary operators, along with a promise of the gap condition of these unitary operators. 
Sequentially applying these operators allows us to steer the system from the eigenstate of the initial unitary to the eigenstate of the final unitary, with a diabatic error bounded by $\mathcal{O}(1/T_{d})$, where $T_d$ is the number of the unitary operators. 
In our analysis, we treat the numerical integrators as a set of discrete walk operators. 
This allows us to directly apply the discrete adiabatic theorems to bound the error between the actual state and the exact eigenstates, making our approach tighter and more straightforward compared to standard analyses where errors are typically bounded by the sum of continuous diabatic errors and time discretization errors. 
The gap condition of the numerical integrator $e^{-ihH(t/T)}$ can be guaranteed from the gap condition of $H(s)$, as long as the step size $h \leq \alpha^{-1}$.  
The changes between adjacent numerical integrators are determined by the scheduling function $f(t/T)$, which is already sufficiently slow with large evolution time $T$ and suppresses the discrete adiabatic errors. 
This is why further reduction of the time step size is unnecessary.

Our analysis based on the discrete adiabatic theorem also applies to the high-order product formula. 
However, unlike the first-order exponential integrator case, now the gap condition of the Trotter numerical integrators can be very different from the gap condition of the Hamiltonian $H(s)$. 
For example, even for the first-order Trotter method $e^{-ihf(s)H_1}e^{-ih(1-f(s))H_0}$, it is not clear how its gap is related to the gap condition of the Hamiltonian $H(s)$ when the step size $h$ is large. 
It becomes even more complicated when we consider higher-order Trotter in~\cref{eqn:high_order_trotter_general}, because the exponentials are the time evolution of Hamiltonian $H(s)$ evaluated at different times. 
Therefore, in~\cref{sec:complexity_Trotter}, we instead consider the simplified higher-order product formula in~\cref{eqn:high_order_trotter_simplified} with all the Hamiltonians evaluated at the same time. 
Our strategy is to reduce the step size $h$ such that the numerical operator $U_{\text{spf},p}$ is close to the walk operator $e^{-ihH(t/T)}$ by the distance at most $c\Delta_*$ for a small constant $c$, then the gap of $H(s)$ would imply the gap of $U_{\text{spf},p}$. 
Such a step size can be determined using the time-independent Trotter error bound in~\cite{ChildsSuTranEtAl2019}. 
An important observation is that the sole purpose of reducing the time step size is to ensure the gap condition, so the time step size depends only on the norms and gaps of the Hamiltonians, but is independent of the overall evolution time $T$ and the tolerated error $\epsilon$. 

Combining this with the discrete adiabatic theorem, we show in~\cref{sec:complexity_Trotter} that, in a $p$-th simplified product formula, a sufficiently accurate simulation may be obtained with time step size $h = \mathcal{O}(\min\{\alpha^{-1},\Delta_*^{1/p}\widetilde{\alpha}_p^{-1/p}\})$, where $\widetilde{\alpha}_p = \sum_{\gamma_0,\cdots,\gamma_p\in\left\{0,1\right\}} \|[H_{\gamma_p}, \cdots, [H_{\gamma_1},H_{\gamma_0}]]\|$ is the scaling of nested commutators of depth $p$. 
Again, our analysis shows that the time step size can be independent of $T$ and $\epsilon$, while the best previous analysis~\cite{WiebeBerryHoyerEtAl2010} gives the estimate of $h = \mathcal{O}(\epsilon^{1/p}\alpha^{-1-1/p}T^{-1/p})$ which still depends on $\epsilon$ and $T$. 
As a result, our results still improve the standard complexity estimate by suggesting a larger time step size and smaller number of propagation steps for accurate simulation. 
Another interesting remark is that our analysis implies the advantage of higher-order methods from a completely different perspective -- the advantage is because higher-order methods can better match the gap condition with larger time step sizes, instead of better $\epsilon$ and $T$ dependence as suggested by the standard numerical analysis.

\subsubsection{Improved complexity with boundary cancellation condition}

We have shown that the overall complexity of simulating near adiabatic dynamics using the first-order exponential integrator or product formula is $\mathcal{O}(1/\epsilon)$, omitting dependence on other parameters. 
This linear dependence is due to the scaling of the adiabatic error in general case, namely $\mathcal{O}(1/T)$ in the continuous adiabatic theorem and $\mathcal{O}(1/T_{d})$ in the discrete adiabatic theorem. 
In~\cref{sec:analysis_boundary_cancellation}, we show how the overall complexity can be further improved by taking advantage of high-order adiabatic theorems, under an additional assumption which is known as the boundary cancellation condition. 

The boundary cancellation condition assumes the scheduling function $f(s)$ has vanishing derivatives of any order at the boundary $s = 0, 1$. 
Then, in the continuous case, the diabatic error can be suppressed by the high-order adiabatic theorems to a super-polynomial scaling $\mathcal{O}(1/T^{k})$ for any positive integer $k$~\cite{jansen2007bounds}, or even an exponential convergence $\mathcal{O}(e^{-cT})$~\cite{Nenciu1993,GeMolnarCirac2016}. 
In the discrete case, the work~\cite{DKS98} also claims a high-order discrete adiabatic theorem, which has the error scaling $\mathcal{O}(1/T^k)$ for any positive integer $k$. 
Again, in the time discretization of~\cref{eqn:AQC_dynamics}, we view the numerical integrators as a set of discrete walk operators. 
Then the high-order discrete adiabatic theorem directly guarantees that, if the Hamiltonian has constant spectral norm and gap, then it suffices to use an $\mathcal{O}(1)$ time step size and the overall number of steps is $\mathcal{O}(1/\epsilon^{o(1)})$. 
Such a high-order convergence holds for any numerical methods, even including first-order ones like the first-order exponential integrator and the first-order product formula. 

However, we remark that the improved scaling $\mathcal{O}(1/\epsilon^{o(1)})$ we find here is not mathematically rigorously proved yet. 
This is because our analysis is based on the high-order discrete adiabatic theorem claimed in~\cite{DKS98}, but we find a missing step in its proof and this missing step seemingly cannot be fixed in a simple way (see~\cref{app:DKS_missing_step} for details). 
Nevertheless, in~\cref{app:DKS_missing_step} we show numerical evidence that supports the correctness of the high-order discrete adiabatic theorem, so our improved complexity would hold in practice.

\subsubsection{Advantage of Trotterized evolution over continuous AQC}

So far, we have given improved complexity estimates of discretized AQC for gapped Hamiltonian $H(s)$. 
Notice that in order to apply the discrete adiabatic theorem, what we need is the gap condition of the numerical integrators rather than the Hamiltonian $H(s)$. 
In the first-order product formula with large time step size, these two types of gap conditions can be very different. 
Our strategy so far is to moderately reduce the time step size $h$ such that the Trotter operator $e^{-ihf(s)H_1}e^{-ih(1-f(s))H_0}$ is sufficiently close to the walk operator $e^{-ihH(s)}$, whose spectral gap is closely related to that of the Hamiltonian $H(s)$.

It seems that the difference of the gap conditions between the interpolating Hamiltonian $H(s)$ and the unitary numerical operators only introduces extra computational overheads, as reduced time step size is required to match the gap conditions. 
However, in~\cref{sec:gap}, we discuss an interesting phenomenon that such difference of the gap conditions can actually be advantageous when the interpolating Hamiltonian has an exponentially small spectral gap or becomes gapless. 
In such cases, Trotterized evolution with large time step size may significantly outperform the continuous AQC and even becomes the only way for adiabatic state preparation. 
More specifically, we consider the case where the gap of $H(s)$ is exponentially small or just becomes $0$. 
Then, continuous AQC requires an exponentially large evolution time or may even fail~\cite{AvronElgart1999}. 
However, numerical discretization with a large time step size, such as the first-order Trotter method with time step size $1$ (that is, $e^{-if(s)H_1}e^{-i(1-f(s))H_0}$), may still have a large gap. 
Thus, discrete adiabatic theorems suggest that the large-step-size evolution may be able to find the desired eigenstates with much lower complexity. 
Notice that in this case, reducing the time step size further could make the situation worse, as the discrete evolution would be forced to approximate the continuous AQC that does not work. 
Therefore, for Hamiltonians with exponentially small gaps, their large-step-size discretization may surprisingly yield an exponential speedup compared to continuous AQC. 
We provide numerical evidence for this phenomenon using a toy model in~\cref{sec:gap}.

\subsubsection{Discrete adiabatic Grover search and implication on QAOA}

All of our aforementioned results have super-linear spectral gap dependence, as they hold generically and in the worst-case scenario. 
However, in continuous AQC, a linear gap dependence can be achieved by carefully designing an interpolation function that takes into account simultaneous gap information~\cite{RolandCerf2002,jansen2007bounds,RevModPhys.90.015002,an2019quantum}. 
We observe that the linear gap dependence can be preserved with the same interpolating function when discretizing the dynamics using large time step size. 

To illustrate this, in~\cref{sec:Grover}, we consider the example of adiabatic Grover search~\cite{RolandCerf2002}, where we aim to find a uniform superposition of $M$ target states in an $N$-dimensional Hilbert space. 
We consider two possible interpolation functions: one generalized from~\cite{RolandCerf2002,DalzellYoderChuang2017} that changes speed proportional to the gap size, and the second, a smooth connection of two identical glue functions, used in~\cite{an2019quantum} to construct adiabatic quantum linear system solvers.  
The latter simultaneously slows down when gap closes and satisfies the boundary cancellation condition. 
We show that the AQC dynamics with either interpolation function can be discretized using a time step size $1$ without incurring extra discretization errors, and the overall number of steps using the second interpolation function can be $\widetilde{\mathcal{O}}(\sqrt{N/M})$ in dimension, almost recovering the quadratic Grover speedup, and $\mathcal{O}(1/\epsilon^{o(1)})$ in precision assuming the correctness of high-order discrete adiabatic theorems. 
Our second scheduling function provides a super-polynomial improvement in precision compared to previous work on adiabatic quantum search~\cite{DalzellYoderChuang2017}. 
A remarkable feature of our algorithm is that it is robust against the estimations of $M$ in two aspects. 
First, constructing the interpolation functions does not require an estimate of $M$ (and for the second interpolation function, not even information about $N$). 
Additionally, increasing the number of propagation steps improves accuracy, similar to fixed-point iteration~\cite{YoderLowChuang2014,DalzellYoderChuang2017}, instead of ``overcooking'' the problem as in the original Grover algorithm~\cite{Grover1996}.

Another powerful approach for eigenstate preparation is the Quantum Approximate Optimization Algorithm (QAOA)~\cite{FarhiGoldstoneGutmann2014}. 
In QAOA, an ansatz is considered that alternatively evolves an input state with the target Hamiltonian $H_1$ and a mixing Hamiltonian $H_0$. 
The goal is to optimize the evolution times, also called QAOA angles, such that the final output state approximates the desired eigenstate. 
QAOA is closely related to AQC and especially discretized AQC, since for gapped Hamiltonian the Trotterization of AQC directly gives a candidate of the QAOA angles, yet QAOA is believed to have the potential to outperform AQC by resulting in a shorter overall evolution time and having better gap dependence. 
Even for the simple example of the Grover search problem, it was previously believed that the performance and analysis of QAOA should involve a very different process from estimating the gap dependence in the AQC approach~\cite{JiangRieffelWang2017}. 
However, our improved AQC results demonstrate the existence of a Trotterized AQC with time step size $1$, which can readily solve the Grover search problem with a cost of $\mathcal{O}(\sqrt{N})$ in the case of a single marked state. 
Therefore, the performance of discretized AQC can match that of QAOA for this task.

\section{Improved time step size and complexity estimate}\label{sec:analysis}

\subsection{Key idea}

We first explain the key components in our improved analysis. 
The first observation is that, unlike the standard analysis which uses the triangle inequality and bounds both the continuous adiabatic error and the time discretization error, we directly view the numerical propagation operator $U_{\text{num}}((j+1)h, jh)$ as the discrete adiabatic walk operator. 
Then, as long as $U_{\text{num}}$ satisfies the assumptions in the discrete adiabatic theorems and the eigenstate of the last numerical propagation operator coincides with our target eigenstate, the overall approximation error can be directly bounded using the discrete adiabatic theorems. 
This is a more straightforward and tighter way to bound the overall error. 
The second observation is that numerical operators with relatively large time step size suffice to fulfill the assumption of slow changes in the discrete adiabatic theorems. 
This is because the Hamiltonian $H(t/T)$ varies on the rescaled time $s = t/T$, which already controls the change of the walk operators at a low level and thus smaller time step sizes are not necessary. 

A further result from our analysis is a better time discretization error bound for the first-order exponential integrator and the product formula. 
We may use the triangle inequality in a different way to bound the numerical discretization error by the sum of the continuous and discrete adiabatic errors, which are both well bounded even with large time step sizes.

\subsection{A review of the discrete adiabatic theorem}\label{sec:DAT_linear}

Discrete adiabatic evolution is a discrete analog of continuous adiabatic evolution. 
In continuous adiabatic evolution (see~\cref{app:continuous_adiabatic} for a detailed discussion), we are given a time-dependent Schr\"odinger equation. 
If the time-dependent Hamiltonian varies slowly and there is a gap in its spectrum, then propagating the Schr\"odinger equation approximately drives the state from an eigenstate of the initial Hamiltonian to an eigenstate of the final Hamiltonian. 
In the discrete setting, we are given a set of unitary walk operators and sequentially apply those on an input quantum state. 
If two adjacent walk operators are close and satisfy the gap condition, then such discrete evolution will approximately drive the quantum state from the eigenstate of the first walk operator to the eigenstate of the final walk operator. 

The discrete adiabatic theorem bounds the distance between the output of the discrete evolution and the ideal eigenspace of the final walk operator. 
It was first proved in~\cite{DKS98}, and a recent work~\cite{CostaAnYuvalEtAl2022} improved the linearly convergent discrete adiabatic theorem to an explicit form with clear gap dependence. 
Here, we briefly review the result in~\cite{CostaAnYuvalEtAl2022}.

Let $W(s)$ be a set of parameterized unitary walk operators for $s \in [0,1]$. 
For an input quantum state $\ket{\psi}$ and a positive integer $T_{d}$, consider the discrete evolution $U(n/T_{d})$ defined as 
\begin{equation}\label{eqn:discrete_evolution}
    U(n/T_{d})\ket{\psi} = \prod_{j=0}^{n-1}W(j/T_{d}) \ket{\psi}, \quad 1 \leq n \leq T_{d}. 
\end{equation}
Suppose that the discrete evolution satisfies the following conditions. 
\begin{enumerate}
    \item For $k = 1,2$, there exist real-valued functions $c_k(s)$ such that 
    \begin{equation}\label{eqn:def_ck}
        \|D^{(k)}W(s)\| \leq \frac{c_k(s)}{T_{d}^k}
    \end{equation}
    for $0 \leq s \leq 1-k/T_{d}$, where $D^{(k)}W(s)$ represents the $k$-th order finite difference of $W(s)$ with step size $1/T_{d}$. \label{cond:1}
    \item The eigenvalues of $W(s)$ can be separated into two groups $\sigma_P(s)$ and $\sigma_Q(s)$ such that for any $k = 0,1,2$ and $0 \leq s \leq 1-k/T_{d}$, the angular distance between the arcs $\sigma_P^{(k)}(s)$ and $\sigma_Q^{(k)}(s)$ is bounded from below by $\Delta_k(s) > 0$, where 
    \begin{equation}\label{eq:seqgap}
        \sigma_P^{(k)}(s) \supseteq \bigcup_{l=0}^k \sigma_P(s+l/T_{d}), \qquad \sigma_Q^{(k)}(s) \supseteq \bigcup_{l=0}^k \sigma_Q(s+l/T_{d}).  
    \end{equation} \label{cond:2}
    \item The input state $\ket{\psi}$ is within the eigenspace corresponding to $\sigma_P(0)$. \label{cond:3}
\end{enumerate}

The first of these conditions means that the rate of variation of the walk operator is bounded, and is a similar condition to the bound on the first and second time derivatives in the continuous case.
The second condition is ensuring that there is a gap.
For generality it is allowed to be a gap between eigenspaces, rather than just a ground state and excited states.
There is a further subtlety in the discrete case that the gap is considered between successive steps as well; that is, the eigenspace $\sigma_P$ at one step cannot overlap with $\sigma_Q$ at the next step.
This is accounted for by considering the union of the eigenspaces at three successive steps as in Eq.~\eqref{eq:seqgap}. 
However, we remark that if $T_d$ is large enough, then this multistep gap condition can be guaranteed by a uniformly bounded spectral gap of each $W(s)$. 
Intuitively, this is because the eigenvalues of $W(s)$ do not change much from one step to another, and thus different arcs of the spectrum at different time steps will not overlap. 
We discuss this with more technical details in~\cref{app:multistep_gap}. 
The third condition just specifies $\sigma_P$ as the eigenspace that the state should remain within during the discrete adiabatic evolution.

For reference we also need to define an ideal adiabatic evolution operator $U_A(n/T_{d})$, which drives the input state to the ideal final eigenspace without any leakage. 
The operator $U_A(n/T_{d})$ is defined through a set of ideal adiabatic walk operators $W_A(s)$ such that 
\begin{equation}
    U_A(n/T_{d}) \ket{\psi} = \prod_{j=0}^{n-1} W_A(j/T_{d}) \ket{\psi}, \qquad 1 \leq n \leq T_{d}. 
\end{equation}
The ideal adiabatic walk operators $W_A(s)$ include a counter-diabatic correction of the original walk operators $W(s)$. 
For technical simplicity, we omit the explicit definition of $W_A(s)$ here and refer to~\cite{DKS98,CostaAnYuvalEtAl2022}. 
An important proposition involving $U_A(n/T_{d})$ is that for all $1 \leq n \leq T_{d}$, 
\begin{equation}\label{eqn:prop_UA}
    U_A(n/T_{d}) P(0) = P(n/T_{d}) U_A(n/T_{d})
\end{equation}
where $P(s)$ denotes the spectral projection onto the eigenspace $\sigma_P(s)$.
The left-hand side corresponds to projecting onto the desired initial eigenspace, then evolving for $n$ steps, which should give the eigenspace at step $n$ if the evolution is adiabatic.
That is indeed what the right-hand side shows, because it has the evolution for $n$ steps \emph{followed} by the projection onto the eigenspace, so acting with this operator on any state must yield a state in the correct eigenspace. 
Therefore $U_A(n/T_{d})$ exactly preserves the simultaneous eigenspace of $W(n/T_{d})$ corresponding to $\sigma_P(n/T_{d})$.  

Now we are ready to state the first discrete adiabatic theorem proved in~\cite{CostaAnYuvalEtAl2022}. 

\begin{lemma}\label{lem:DAE_linear}
    Suppose that $T_{d} \geq \sup_{s\in[0,1]} (4\hat{c}_1(s)/\check{\Delta}_2(s))$, where the hat (check) notations of a function at $s$ represent the maximum (minimum) of its function values at neighboring points $\{s-1/T_{d},s,s+1/T_{d}\}$ whenever well-defined. 
    Then for any $1 \leq n \leq T_{d}$, there exists an absolute constant $C$ such that 
    \begin{equation}
        \|U(n/T_{d})-U_A(n/T_{d})\| \leq \frac{C}{T_{d}}\left(\frac{\hat{c}_1(0)}{\check{\Delta}_2(0)^2} + \frac{\hat{c}_1(n/T_{d})}{\check{\Delta}_2(n/T_{d})^2} + \sum_{j=0}^{n-1} \frac{\hat{c}_1(j/T_{d})^2}{T_{d}\check{\Delta}_2(j/T_{d})^3} + \sum_{j=0}^{n-1} \frac{\hat{c}_2(j/T_{d})}{T_{d}\check{\Delta}_2(j/T_{d})^2}\right). 
    \end{equation}
\end{lemma}

Because $U_A(n/T_{d})$ is exact adiabatic evolution, it would map an initial state in the initial eigenspace to one in the desired final eigenspace.
The bound on the difference of operators therefore gives a bound on how much the state can leak outside the eigenspace under the true evolution $U(n/T_{d})$. 

\subsection{First-order exponential integrator}\label{sec:complexity_exponential_integrator}

\subsubsection{General case with gap assumption on the walk operator}

Now we consider walk operators that are given by Hamiltonian evolution for time $h$, $W(s) = e^{-i h H(s)}$, where $H(s)$ is varied in an adiabatic way as in Eq.~\eqref{eq:schedule}.
We can then use this form to determine $c_k(s)$ and provide a theorem for the error in terms of the properties of the Hamiltonians $H_0,H_1$.
For this theorem the gap is still given for the walk operator, though it will be implicitly dependent on $h$. 

In general cases, we can obtain a linear convergence in the evolution time. 
We first prove the following result for arbitrary time step size $h$ under the multistep gap assumption for the walk operator.  

\begin{theorem}\label{thm:exp_linear}
    Consider~\cref{prob:state_prep} using digital adiabatic simulation with first-order exponential integrator $W(s) = e^{-i h H(s)}$, where $H(s)$ is given as~\cref{eq:schedule}. 
    Suppose that both $|f'(s)|$ and $|f''(s)|$ are uniformly bounded over $[0,1]$. 
    Let $\widetilde{\Delta}(s)$ denote the multistep spectral gap of $U_{\exp}(sT+h,sT)$, and $\widetilde{\Delta}_* = \min_{s\in[0,1-h/T]} \widetilde{\Delta}(s)$. 
    Then, the overall error between the actual and the ideal evolution can be bounded by 
    \begin{equation}
        \mathcal{O}\left(\frac{h}{T}\left(\frac{h(\|H_0\|+\|H_1\|)}{\widetilde{\Delta}_*^2} + \frac{ h^2 (\|H_0\|+\|H_1\|)^2}{\widetilde{\Delta}_*^2} + \frac{h^2(\|H_0\|+\|H_1\|)^2}{\widetilde{\Delta}_*^3} \right) \right). 
    \end{equation}
\end{theorem}
\begin{proof}
    Let $T_{d} = T/h$. 
    For $s\in[0,1]$ and any $0 \leq j \leq T_{d}-1$, we have 
    \begin{equation}
        W(j/T_{d}) = e^{-i h H(j/T_{d})} = e^{-i h H(jh/T)} = U_{\exp}((j+1)h,jh). 
    \end{equation}
    We show in~\cref{lem:c1_c2_exp} that, for this walk operator $W(s)$, we can choose its corresponding $c_1(s)$ and $c_2(s)$ (defined in~\cref{eqn:def_ck}) as 
    \begin{equation}
        c_1(s) = \mathcal{O}(h(\|H_0\|+\|H_1\|)), \quad c_2(s) = \mathcal{O}(h(\|H_0\|+\|H_1\|) + h^2(\|H_0\|+\|H_1\|)^2). 
    \end{equation}
    Then we can apply~\cref{lem:DAE_linear} to obtain the desired error bound. 
\end{proof}

\subsubsection{General case without gap assumption on the walk operator}

\cref{thm:exp_linear} assumes the multistep gap condition for the first-order exponential integrator. 
While this may be established independently in some applications, a more common assumption is the (fixed-time) gap condition for the Hamiltonian $H(s)$. 
When the time step size is small enough, we may relate the multistep gap condition for $W(s)$ to the fixed-time gap condition for $H(s)$ and show the following result. 

\begin{corollary}\label{cor:exponential_linear}
    Consider~\cref{prob:state_prep} using digital adiabatic simulation with the first-order exponential integrator~\cref{eqn:exponential_integrator_1st}, $W(s) = e^{-ihH(s)}$.
    Suppose that both $|f'(s)|$ and $|f''(s)|$ are uniformly bounded over $[0,1]$. 
    Let $\Delta(s)$ denote the spectral gap of $H(s)$, and $\Delta_* = \min_{s\in[0,1]} \Delta(s)$. 
    Then, in order to bound the overall error by $\epsilon \in (0,1)$, it suffices to choose 
    \begin{equation}
        T = \mathcal{O}\left( \frac{(\|H_0\|+\|H_1\|)^2}{ \Delta_*^3 \epsilon}  \right), 
    \end{equation}
    and 
    \begin{equation}
        h = \frac{1}{\|H_0\|+\|H_1\|} , 
    \end{equation}
    and the overall number of steps becomes 
    \begin{equation}
        T_{d} = \mathcal{O}\left( \frac{(\|H_0\|+\|H_1\|)^3}{ \Delta_*^3 \epsilon}  \right). 
    \end{equation}
\end{corollary}
\begin{proof}
    Following the notations in the proof of~\cref{thm:exp_linear}, we use $T_d = T/h$ to denote the number of steps and define the walk operator $W(s) = e^{-ih H(s)}$. 
    Notice that when $h = 1/(\|H_0\|+\|H_1\|)$, we always have $\|hH(s)\| \leq 1$, so the gap of $W(s)$ is exactly the same as that of $hH(s)$, 
    which is bounded from below by $h \Delta_*$. 
    Suppose that our choice of $T$ is 
    \begin{equation}
        T = C \frac{(\|H_0\|+\|H_1\|)^2}{ \Delta_*^3 \epsilon}
    \end{equation}
    for a constant $C > 0$. 
    Then we have 
    \begin{equation}
        T_d = T/h = C \frac{(\|H_0\|+\|H_1\|)^3}{ \Delta_*^3 \epsilon} \geq C \frac{(\|H_0\|+\|H_1\|)^3}{ \Delta_*^3}. 
    \end{equation}
    Notice that $\Delta_* \leq 2\|H(s)\| \leq 2(\|H_0\|+\|H_1\|)$, so we have 
    \begin{equation}
        T_d \geq \frac{C}{4} \frac{\|H_0\|+\|H_1\|}{ \Delta_*} = \frac{C}{4 h\Delta_* }. 
    \end{equation}
    According to~\cref{lem:c1_c2_exp}, we have $c_1(s) = \mathcal{O}(h(\|H_0\|+\|H_1\|)) = \mathcal{O}(1)$, so we can choose a sufficiently large constant $C$ such that $T_d \geq \frac{2\pi}{h\Delta_*} \sup c_1(s)$. 
    Then~\cref{lem:multistep_gap} ensures that the multistep gap of $W(s)$ is bounded from below by $h\Delta_*/2$. 
    
    Plugging this back into~\cref{thm:exp_linear} yields the error bound 
    \begin{align}
        & \quad \mathcal{O}\left(\frac{h}{T}\left(\frac{h(\|H_0\|+\|H_1\|)}{ h^2 \Delta_*^2} + \frac{ h^2 (\|H_0\|+\|H_1\|)^2}{h^2 \Delta_*^2} + \frac{h^2(\|H_0\|+\|H_1\|)^2}{h^3 \Delta_*^3} \right) \right) \nonumber \\
        & \leq \mathcal{O}\left(\frac{1}{T}\left(\frac{\|H_0\|+\|H_1\|}{  \Delta_*^2} + \frac{ h (\|H_0\|+\|H_1\|)^2}{ \Delta_*^2} + \frac{(\|H_0\|+\|H_1\|)^2}{ \Delta_*^3} \right) \right) \nonumber\\
        & \leq \mathcal{O}\left( \frac{(\|H_0\|+\|H_1\|)^2}{ T \Delta_*^3}  \right). 
    \end{align}
    In order to further bound the error by $\epsilon$, it suffices to choose 
    \begin{equation}
        T = \mathcal{O}\left( \frac{(\|H_0\|+\|H_1\|)^2}{ \Delta_*^3 \epsilon}  \right). 
    \end{equation}
    Since the error bound is independent of the time step $h$, we directly choose the largest possible time step size 
    \begin{equation}
        h = \frac{1}{\|H_0\|+\|H_1\|}
    \end{equation}
    which ensures the fixed-time gap condition of $W(s)$. 
    Finally, the condition $T/h > 2\widetilde{c} (\|H_0\|+\|H_1\|)/\Delta_*$ is naturally satisfied, since $\Delta_* \leq 2\|H(s)\| \leq 2(\|H_0\|+\|H_1\|)$ and 
    \begin{equation}
        T_{d} \sim \frac{(\|H_0\|+\|H_1\|)^3}{\Delta_*^3 \epsilon} = \frac{1}{\epsilon}\left(\frac{(\|H_0\|+\|H_1\|)}{\Delta_*}\right)^2 \frac{\|H_0\|+\|H_1\|}{\Delta_*},
    \end{equation}
    from which we can also infer that $T_{d} \gtrsim \frac{\|H_0\|+\|H_1\|}{\Delta_*} $ and thus satisfies the condition for $T/h$. 
\end{proof}

\subsubsection{Implementation}

So far we have shown an improved estimate on the number of steps in the first-order exponential integrator. 
However, it is still not straightforward to implement the first-order exponential integrator $e^{-ihH(s)}$ from the input models of $H_0$ and $H_1$. 
Here we briefly describe how to implement $e^{-ihH(s)}$. 
Suppose that we are given unitaries $V_0$ and $V_1$ such that, for $j = 0,1$, $V_j$ is a $(\eta_j,n_a,0)$-block-encoding of $H_j$, \emph{i.e.}, 
\begin{equation}
    _a\bra{0} V_j \ket{0}_a = \frac{1}{\eta_j} H_j. 
\end{equation}
Here $\eta_j$ represents the block-encoding factor such that $\eta_j \geq \|H_j\|$. 
The number of the ancilla qubits for both block-encodings is the same as $n_a$ without loss of generality (if the number of the ancilla qubits are not the same, then we may define $n_a$ to be the larger one and supplement the other unitary with extra trivial ancilla qubits). 
We further assume that the block-encoding factors are asymptotically optimal in the sense that $\eta_j = \mathcal{O}(\|H_j\|)$ as well. 

For every fixed $s \in [0,1]$, we first block encode $H(s) = (1-f(s))H_0 + f(s) H_1$. 
This can be done using the linear combination of unitaries technique supplemented with an extra ancilla qubit. 
Let the rotation be 
\begin{equation}
    R(s) = \frac{1}{\sqrt{(1-f(s))^2\eta_0^2 + f(s)^2\eta_1^2}} 
    \left(  \begin{array}{cc}
        (1-f(s))\eta_0 & f(s)\eta_1 \\
        f(s)\eta_1 &  -(1-f(s))\eta_0
    \end{array}  \right), 
\end{equation}
and the select oracle be 
\begin{equation}
    \text{SEL} = \ket{0}\bra{0} \otimes V_0 + \ket{1}\bra{1} \otimes V_1. 
\end{equation}
Then, according to~\cite[Lemma 52]{GilyenSuLowEtAl2019}, the operator $(\mathrm{H}\otimes I_{n_a} \otimes I_n)\text{SEL}(R(s)\otimes I_{n_a} \otimes I_n)$ is a $(1,n_a+1,0)$-block-encoding of 
\begin{equation}
    \frac{1}{ \sqrt{2} \sqrt{(1-f(s))^2\eta_0^2 + f(s)^2\eta_1^2} } H(s). 
\end{equation}
To implement $e^{-ihH(s)}$, we may use the qubitization~\cite{Low2019hamiltonian} or the quantum singular value transformation technique~\cite{GilyenSuLowEtAl2019}. 
In particular, write 
\begin{equation}
    e^{-ihH(s)} = \exp \left( -i h \sqrt{2} \sqrt{(1-f(s))^2\eta_0^2 + f(s)^2\eta_1^2} \frac{1}{\sqrt{2} \sqrt{(1-f(s))^2\eta_0^2 + f(s)^2\eta_1^2}} H(s) \right), 
\end{equation}
so we need to evolve the rescaled Hamiltonian $\frac{1}{\sqrt{2} \sqrt{(1-f(s))^2\eta_0^2 + f(s)^2\eta_1^2}} H(s)$ to time $h \sqrt{2} \sqrt{(1-f(s))^2\eta_0^2 + f(s)^2\eta_1^2}$, which is still $\mathcal{O}(1)$. 
This is because $h \sqrt{2} \sqrt{(1-f(s))^2\eta_0^2 + f(s)^2\eta_1^2} \leq \sqrt{2} h (\eta_0+\eta_1)$, which is $\mathcal{O}(1)$ due to the choice of $h$ in~\cref{cor:exponential_linear}. 
Then, by~\cite[Corollary 60]{GilyenSuLowEtAl2019}, we may obtain a $(1,n_a+3,\epsilon')$-block-encoding of $e^{-ihH(s)}$, using the unitaries $V_0$ and $V_1$ a total number of times $\mathcal{O}(\log(1/\epsilon'))$. 
Finally, the overall numerical propagator can be obtained by multiplying all the block-encoded unitaries, which, by~\cite[Corollary 55]{GilyenSuLowEtAl2019arXiv}, yields a $(1,n_a+3,4(T/h)^2\epsilon')$-block-encoding of it. 
To bound the overall error by $\epsilon$, it suffices to choose $\epsilon' = \epsilon h^2/T^2$, and this contributes to an extra factor of $\log((\|H_0\|+\|H_1\|)/(\epsilon \Delta_*) )$ in the complexity of the gate-based implementation.

\subsection{Product formula}\label{sec:complexity_Trotter}

\subsubsection{General case with gap assumption on the walk operator}

We now consider the product formula 
\begin{equation}\label{eqn:high_order_trotter_general_2}
    U_{\text{pf},p}(t+h,t) = \prod_{k=0}^{K_p} \exp\left(  -i\beta_{p,k} h f((t+\delta_{p,k}h)/T) H_1\right)  \exp\left( -i\alpha_{p,k}h (1-f((t+\gamma_{p,k}h)/T)) H_0 \right)
\end{equation}
and its simplified version 
\begin{equation}\label{eqn:high_order_trotter_simplified_2}
    U_{\text{spf},p}(t+h,t) = \prod_{k=0}^{K_p} \exp\left(  -i\beta_{p,k} h f(t/T) H_1\right)  \exp\left( -i\alpha_{p,k}h (1-f(t/T)) H_0 \right). 
\end{equation} 
We start with the simplest case where we assume the multistep gap condition for the product walk operator. 

\begin{theorem}\label{thm:trotter_linear}
    Consider~\cref{prob:state_prep} using digital adiabatic simulation with the product formula $U_{\text{pf},p}$ or the simplified product formula $U_{\text{spf},p}$. 
    Suppose that both $|f'(s)|$ and $|f''(s)|$ are uniformly bounded over $[0,1]$. 
    Let $\widetilde{\Delta}(s)$ denote the multistep spectral gap of $U_{\text{pf},p}(sT+h,sT)$ or $U_{\text{spf},p}(sT+h,sT)$, and $\widetilde{\Delta}_* = \min_{s\in[0,1-h/T]} \widetilde{\Delta}(s)$. 
    Then, the overall error between the actual and the ideal evolution can be bounded by 
    \begin{equation}
        \mathcal{O}\left( \frac{h}{T}\left(\frac{ h (\|H_0\|+\|H_1\|)}{\widetilde{\Delta}_*^2} + \frac{ h^2  (\|H_0\|+\|H_1\|)^2}{\widetilde{\Delta}_*^2} + \frac{ h^2 (\|H_0\|+\|H_1\|)^2}{\widetilde{\Delta}_*^3}\right) \right). 
    \end{equation}
\end{theorem}
\begin{proof}
    For $s\in[0,1-h/T]$ and any $0 \leq j \leq T_{d}-1$, the walk operator $W(j/T_{d})$ is $U_{\text{pf},p}((j+1)h,jh)$ or $U_{\text{pfs},p}((j+1)h,jh)$. 
    We show in~\cref{lem:c1_c2_product_formula} and~\cref{lem:c1_c2_product_formula_simplified} that, for the walk operator $W(s)$ in either case, we can choose $c_1(s)$ and $c_2(s)$ (defined in~\cref{eqn:def_ck}) as 
    \begin{equation}
        c_1(s) = \mathcal{O}(h(\|H_0\|+\|H_1\|)), \quad c_2(s) = \mathcal{O}(h(\|H_0\|+\|H_1\|) + h^2(\|H_0\|+\|H_1\|)^2). 
    \end{equation}
    Then we can apply~\cref{lem:DAE_linear} to obtain the desired error bound. 
\end{proof}

\subsubsection{General case without gap assumption on the walk operator}

\cref{thm:trotter_linear} assumes the gap condition for the product walk operator, which is not natural when the time step size $h$ is large. 
When the gap condition is not assumed, we can indeed decrease the time step size to make the walk operator closer to the evolution $e^{-i h H(s)}$, of which the gap is closely related to that of $H(s)$. 
Here we only consider the simplified product formula~\cref{eqn:high_order_trotter_simplified}. 

We first prove the result for first- and second-order product formulae. 
An interesting phenomenon is that the first- and second-order formulae share the same eigenvalues as in the following lemma. 

\begin{lemma}\label{lem:trotter_first_second_eigenvalues}
    Consider the simplified first-order product formula $U_{\text{spf},1}(t+h,t) = e^{-ihf(t/T)H_1}e^{-ih(1-f(t/T))H_0}$ and the simplified second-order product formula  $U_{\text{spf},2}(t+h,t) = e^{-ih(1-f(t/T))H_0/2}e^{-ihf(t/T)H_1}e^{-ih(1-f(t/T))H_0/2}$. 
    The eigenvalues of $U_{\text{spf},1}(t+h,t)$ and $U_{\text{spf},2}(t+h,t)$ are exactly the same. 
\end{lemma}
\begin{proof}
    Notice that $U_{\text{spf},2}$ is just a rotation of $U_{\text{spf},1}$ according to a unitary, i.e., 
    \begin{equation}
        U_{\text{spf},2}(t+h,t) = V U_{\text{spf},1}(t+h,t) V^{\dagger}, \quad V = e^{-ih(1-f(t/T))H_0/2}, 
    \end{equation}
    and such a rotation does not change eigenvalues. 
\end{proof}

Now we can estimate the gap of the product walk operators by reducing the time step size and using perturbation theory. 

\begin{lemma}\label{lem:gap_Trotter_1st_2nd}
    For $s\in[0,1-h/T]$ and a scheduling function $f(s)$, let $H(s) = (1-f(s))H_0+f(s)H_1$ be the interpolating Hamiltonian, $U_{\text{spf},1}$ and $U_{\text{spf},2}$ be the simplified time-dependent product formulae as $U_{\text{spf},1}(sT+h,sT) = e^{-ihf(s)H_1}e^{-ih(1-f(s))H_0}$ and $U_{\text{spf},2}(sT+h,sT) = e^{-ih(1-f(s))H_0/2}e^{-ihf(s)H_1}e^{-ih(1-f(s))H_0/2}$. 
    Let $\Delta_H(s)$ and $\Delta_{U,p}(s)$ denote the spectral gap of $H(s)$ and $U_{\text{spf},p}$, respectively. 
    Then, for any $h \leq 1/(\|H_0\|+\|H_1\|)$ and $p = 1,2$, we have 
\begin{align}
&h\Delta_H(s) - \frac{ h^3}{95}\left( 2\|[H_1,[H_1,H_0]]\| +  \|[H_0,[H_0,H_1]]\|\right) \nonumber\\
&\leq \Delta_{U,p}(s)\nonumber\\ 
&\leq h\Delta_H(s) + \frac{h^3}{95}\left( 2\|[H_1,[H_1,H_0]]\| +  \|[H_0,[H_0,H_1]]\|\right). 
\end{align}
\end{lemma}
\begin{proof}
According to to~\cref{lem:trotter_first_second_eigenvalues}, it suffices to consider only the second-order formula. 

Notice that the spectral gaps of $e^{-i h H(s)}$ and $h H(s)$, though defined in slightly different ways, are of the same value for all $h \leq 1/(\|H_0\|+\|H_1\|)$. 
So it suffices to investigate the spectral gaps of $U_{\text{spf},2}(sT+h,sT)$ and $e^{-ihH(s)}$. 
For every fixed $s$, we may bound the difference between $U_{\text{spf},2}(sT+h,sT)$ and $e^{-ihH(s)}$ by using the error bound for the second-order time-independent Trotter method because the discrete time points are fixed in different exponentials. 
For example, by~\cite[Proposition 16]{ChildsSuTranEtAl2019}, 
\begin{align}
\left\| U_{\text{spf},2}(sT+h,sT)-e^{-i h H(s)} \right\| & \leq \frac{h^3}{12} \left\| [f(s)H_1,[f(s)H_1,(1-f(s))H_0]] \right\|\nonumber\\ 
&\quad+ \frac{h^3}{24} \left\| [(1-f(s))H_0,[(1-f(s))H_0,f(s)H_1]] \right\| \nonumber\\
 & \leq \frac{h^3}{192}\left( 2\|[H_1,[H_1,H_0]]\| +  \|[H_0,[H_0,H_1]]\|\right)\label{eqn:gap_time_independent_Trotter_bound}. 
\end{align}
Equation \eqref{eqn:gap_time_independent_Trotter_bound} implies that the difference between $U_{\text{spf},2}(sT+h,sT)$ and $e^{-iH(s)}$ is small if we reduce the time step size $h$, and further implies that the eigenvalues of $U_{\text{spf},2}(sT+h,sT)$ and $e^{-iH(s)}$ are close according to~\cite{BhatiaDavis1984,ElsnerHe1993}. 
Specifically, according to the eigenvalue perturbation theorem of Refs.~\cite{BhatiaDavis1984,ElsnerHe1993}, there exists an ordering $\lambda_j(s)$ and $\widetilde{\lambda}_j(s)$ of the eigenvalues of $U_{\text{spf},p}(sT+h,sT)$ and $e^{-iH(s)}$, respectively, such that
\begin{align}\label{eqn:proof_gap_trotter_error}
    \max_j |\lambda_j(s) - \widetilde{\lambda}_j(s)| &\leq \left\| U_{\text{spf},2}(sT+h,sT) -  e^{-iH(s)}\right\| \nn
    &\leq \frac{h^3}{192}\left( 2\|[H_1,[H_1,H_0]]\| +  \|[H_0,[H_0,H_1]]\|\right) \nn
    &\leq \frac{h^3}{192}\left( 8 \|H_0\|\|H_1\|^2 +  4 \|H_0\|^2\|H_1\|\right) \nn
    &\leq \frac{h^3}{192}\left( \frac 8{3\sqrt{3}} (\|H_0\|+\|H_1\|)^3\right) \nn
    &\leq \frac 1{72\sqrt{3}} \, .
    \end{align}
Now let $\mathbb{S}^1$ denote the unit circle, and $|z_1-z_2|_{\mathbb{S}^1}$ denote the angular distance between $z_1,z_2 \in \mathbb{S}^1$.
Then we have
\begin{equation}
    |\lambda_j(s) - \widetilde{\lambda}_j(s)|_{\mathbb{S}^1}= 2\arcsin\left( \frac 1{2}|\lambda_j(s) - \widetilde{\lambda}_j(s)| \right)\, .
\end{equation}
Because the arcsine has a monotonically increasing first derivative, the inequality in Eq.~\eqref{eqn:proof_gap_trotter_error} implies
\begin{align}
    |\lambda_j(s) - \widetilde{\lambda}_j(s)|_{\mathbb{S}^1}&\leq 144\sqrt{3} \arcsin\left( \frac 1{155\sqrt{3}} \right) |\lambda_j(s) - \widetilde{\lambda}_j(s)| \nn
    &\leq \frac {96}{95} |\lambda_j(s) - \widetilde{\lambda}_j(s)| \, .
\end{align}
    
    Let $\lambda_0(s)$ and $\lambda_1(s)$ denote the two eigenvalues that determine the gap. 
    Then the gap $U_{\text{spf},2}(sT+h,sT)$ can be bounded by
    \begin{align}
    |\lambda_1(s) - \lambda_0(s)|_{\mathbb{S}^1} & \geq |\widetilde{\lambda}_1(s) - \widetilde{\lambda}_0(s)|_{\mathbb{S}^1} - |\widetilde{\lambda}_1(s) - \lambda_1(s)|_{\mathbb{S}^1} - |\widetilde{\lambda}_0(s) - \lambda_0(s)|_{\mathbb{S}^1} \nonumber\\
    & \geq h\Delta_H(s) - \frac{96}{95} |\widetilde{\lambda}_1(s) - \lambda_1(s)| - \frac{96}{95}|\widetilde{\lambda}_0(s) - \lambda_0(s)|  \nonumber\\
    & \geq h\Delta_H(s) - \frac{h^3}{95}\left( 2\|[H_1,[H_1,H_0]]\| +  \|[H_0,[H_0,H_1]]\|\right), 
    \end{align}
and 
\begin{align}
    |\lambda_1(s) - \lambda_0(s)|_{\mathbb{S}^1} & \leq |\widetilde{\lambda}_1(s) - \widetilde{\lambda}_0(s)|_{\mathbb{S}^1} + |\widetilde{\lambda}_1(s) - \lambda_1(s)|_{\mathbb{S}^1} + |\widetilde{\lambda}_0(s) - \lambda_0(s)|_{\mathbb{S}^1}\nonumber \\
    & \leq h\Delta_H(s) + \frac{96}{95} |\widetilde{\lambda}_1(s) - \lambda_1(s)| + \frac{96}{95}|\widetilde{\lambda}_0(s) - \lambda_0(s)|  \nonumber\\
    & \leq h\Delta_H(s) + \frac{h^3}{95}\left( 2\|[H_1,[H_1,H_0]]\| +  \|[H_0,[H_0,H_1]]\|\right), 
\end{align}
where we use the fact again that the gaps of the Hamiltonian and the walk operator are the same. 
\end{proof}

According to~\cref{lem:gap_Trotter_1st_2nd}, we can reduce the time step size $h$ to make the product walk operator satisfy the gap condition as well. 
An important observation is that the time step size will not depend on the tolerated error $\epsilon$, because the only reason to reduce the time step size is for the gap condition, which is independent of the numerical error. 

Now we are ready to prove the following complexity estimate for the product formula without assuming the gap condition of the product walk operator. 

\begin{corollary}\label{cor:1st_2nd_product_formula}
    Consider~\cref{prob:state_prep} using digital adiabatic simulation with the simplified first- and second-order product formulae as in~\cref{eqn:high_order_trotter_simplified}. 
    Suppose that both $|f'(s)|$ and $|f''(s)|$ are uniformly bounded over $[0,1]$. 
    Let $\Delta(s)$ denote the spectral gap of $H(s)$, and $\Delta_* = \min_{s\in[0,1]} \Delta(s)$. 
    Then, in order to bound the overall error between the actual and the ideal evolution by $\epsilon \in (0,1)$, it suffices to choose 
        \begin{equation}
            T = \mathcal{O}\left( \frac{(\|H_0\|+\|H_1\|)^2}{ \Delta_*^3 \epsilon}  \right), 
        \end{equation}
        and 
        \begin{equation}
            h = \min\left\{ \frac{1}{\|H_0\|+\|H_1\|}, \sqrt{\frac{95}{2} } \sqrt{\frac{\Delta_*}{2\|[H_1,[H_1,H_0]]\| +  \|[H_0,[H_0,H_1]]\|}}\right\}, 
        \end{equation}
        and the overall number of steps becomes 
        \begin{equation}
            T_{d} = \mathcal{O}\left( \frac{(\|H_0\|+\|H_1\|)^3}{ \Delta_*^3 \epsilon} \max\left\{ 1, \sqrt{\frac{2\|[H_1,[H_1,H_0]]\| +  \|[H_0,[H_0,H_1]]\|}{\Delta_* (\|H_0\|+\|H_1\|)^2}}\right\} \right). 
        \end{equation}
\end{corollary}
\begin{proof}
    Let $T_{d} = T/h$, and we assume $T_{d}$ is an integer. 
    Define the walk operator 
    \begin{equation}
        W(s) = e^{-ih(1-f(s))H_0/2}e^{-ihf(s)H_1}e^{-ih(1-f(s))H_0/2}. 
    \end{equation}
    We first bound the fixed-time gap of $W(s)$. 
    According to the choice of $h$, 
    we always have 
    \begin{equation}
        \frac{h^3}{95}\left( 2\left\|[H_1,[H_1,H_0]]\right\| +  \left\|[H_0,[H_0,H_1]]\right\|\right) \leq \frac{1}{2} h \Delta_*. 
    \end{equation}
    Using~\cref{lem:gap_Trotter_1st_2nd}, the gap of $W(s)$ is bounded from below by the gap of $\frac{1}{2} h \Delta(s)$, which has a further lower bound $\frac{1}{2}h \Delta_*$. 
    
    For the multistep gap, we can use~\cref{lem:multistep_gap}, and we need to verify that $T_d$ is sufficiently large. 
    Specifically, suppose that our choice of $T$ is 
    \begin{equation}
        T = C \frac{(\|H_0\|+\|H_1\|)^2}{ \Delta_*^3 \epsilon}
    \end{equation}
    for a constant $C > 0$. 
    Then we have 
    \begin{equation}
        T_d = T/h = C \frac{(\|H_0\|+\|H_1\|)^2}{ h \Delta_*^3 \epsilon} \geq C \frac{(\|H_0\|+\|H_1\|)^2}{ h \Delta_*^3}. 
    \end{equation}
    Notice that $\Delta_* \leq 2\|H(s)\| \leq 2(\|H_0\|+\|H_1\|)$, so we have 
    \begin{equation}
        T_d \geq \frac{C}{4} \frac{1}{ h\Delta_*}. 
    \end{equation}
    According to~\cref{lem:c1_c2_product_formula_simplified}, we have $c_1(s) = \mathcal{O}(h(\|H_0\|+\|H_1\|)) = \mathcal{O}(1)$, so we can choose a sufficiently large constant $C$ such that $T_d \geq \frac{2\pi}{h\Delta_*/2} \sup c_1(s)$. 
    Then~\cref{lem:multistep_gap} ensures that the multistep gap of $W(s)$ is bounded from below by $h\Delta_*/4$. 
    
    Therefore, the error bound in~\cref{thm:trotter_linear} becomes 
    \begin{align}
        & \quad \mathcal{O}\left( \frac{h}{T}\left(\frac{ h (\|H_0\|+\|H_1\|)}{h^2\Delta_*^2} + \frac{ h^2  (\|H_0\|+\|H_1\|)^2}{h^2\Delta_*^2} + \frac{ h^2 (\|H_0\|+\|H_1\|)^2}{h^3\Delta_*^3}\right) \right)\nonumber \\
        & \leq \mathcal{O}\left( \frac{1}{T}\left(\frac{  (\|H_0\|+\|H_1\|)}{\Delta_*^2} + \frac{ h (\|H_0\|+\|H_1\|)^2}{\Delta_*^2} + \frac{ (\|H_0\|+\|H_1\|)^2}{\Delta_*^3}\right) \right) \nonumber\\
        & \leq \mathcal{O}\left( \frac{ (\|H_0\|+\|H_1\|)^2}{T\Delta_*^3}\right), 
    \end{align}
    and the desired choices of $T$ and $h$ directly follow from this error bound and the conditions for the multistep gap. 
\end{proof}

For high-order product formulae, we can estimate the complexity by the same strategy as the first- and second-order formulae and replacing the second nested commutators by higher-order nested commutators due to the high-order Trotter error bound in~\cite[Theorem 11]{ChildsSuTranEtAl2019}. 
Now we state the result in the following corollary, and give its proof in~\cref{app:proof_high_order_Trotter}. 

\begin{corollary}\label{cor:Trotter_linear}
    Consider~\cref{prob:state_prep} using digital adiabatic simulation with simplified $p$-th order product formulae as in~\cref{eqn:high_order_trotter_simplified}. 
    Suppose that both $|f'(s)|$ and $|f''(s)|$ are uniformly bounded over $[0,1]$. 
    Let $\Delta(s)$ denote the spectral gap of $H(s)$, and $\Delta_* = \min_{s\in[0,1]} \Delta(s)$. 
    Then, for any $p \geq 1$, 
    \begin{enumerate}
        \item the overall error between the actual and the ideal evolution can be bounded by 
    \begin{equation}
        \mathcal{O}\left( \frac{ (\|H_0\|+\|H_1\|)^2}{T\Delta_*^3}\right). 
    \end{equation}
        \item in order to bound the overall error by $\epsilon$, it suffices to choose 
        \begin{equation}
            T = \mathcal{O}\left( \frac{(\|H_0\|+\|H_1\|)^2}{ \Delta_*^3 \epsilon}  \right), 
        \end{equation}
        and 
        \begin{equation}
            h = \Theta\left(\min\left\{ \frac{1}{\|H_0\|+\|H_1\|}, \frac{\Delta_*^{1/p}}{\left(\sum_{\gamma_0,\cdots,\gamma_{p} \in \left\{0,1\right\}}\|[H_{\gamma_p},\cdots,[H_{\gamma_1},H_{\gamma_0}]]\|\right)^{1/p}}\right\}\right), 
        \end{equation}
        and the overall number of steps becomes 
        \begin{equation}
            T_{d} = \mathcal{O}\left( \frac{(\|H_0\|+\|H_1\|)^3}{ \Delta_*^3 \epsilon} \max\left\{ 1, \frac{\left(\sum_{\gamma_0,\cdots,\gamma_{p} \in \left\{0,1\right\}}\|[H_{\gamma_p},\cdots,[H_{\gamma_1},H_{\gamma_0}]]\|\right)^{1/p}}{\Delta_*^{1/p}(\|H_0\|+\|H_1\|)}  \right\} \right). 
        \end{equation}
    \end{enumerate}
\end{corollary}

\section{Improved time step size and complexity estimate with boundary cancellation}\label{sec:analysis_boundary_cancellation}

\subsection{A review of high-order discrete adiabatic theorem}

The high-order discrete adiabatic theorem aims to establish a super-polynomial convergence in $1/T$ under further assumptions. 
In particular, we suppose the discrete evolution satisfies the following conditions. 
\begin{enumerate}
    \item For $k \geq 1$, there exist real-valued functions $c_k(s)$ such that $\|D^{(k)}W(s)\| \leq c_k(s)/T^k$ for $0 \leq s \leq 1-k/T$, where $D^{(k)}W(s)$ represents the $k$-th order finite difference of $W(s)$ with step size $1/T$. \label{cond:1_exp}
    \item The eigenvalues of $W(s)$ can be separated into two groups $\sigma_P(s)$ and $\sigma_Q(s)$ such that each continuous eigenpath $\lambda(s)$ of $W(s)$ consistently belongs to one of the two groups and the distance between $\sigma_P(s)$ and $\sigma_Q(s)$ is bounded from below by $\Delta(s) > 0$. \label{cond:2_exp}
    \item The input state $\ket{\psi}$ is within the eigenspace corresponding to $\sigma_P(0)$. \label{cond:3_exp}
\end{enumerate}

Under these conditions, Ref.~\cite{DKS98} claims an upper bound of the discrete diabatic error, achieving $\mathcal{O}(1/T^{k})$ for any positive integer $k$. 
Unfortunately, after a careful examination, we find a missing step in the proof provided in Ref.~\cite{DKS98}, so the claimed high-order discrete adiabatic theorem is not rigorously supported. 
Nevertheless, we also perform a numerical test, verifying that the conclusion on the high-order convergence is numerically correct. 
Therefore, to be precise, here we present the result in~\cite{DKS98} as a conjecture of which we have strong confidence on the correctness, and refer to~\cref{app:DAT_exp} for more details on the missing step in the proof and our numerical validation. 

\begin{conjecture}\label{lem:DAE_exp}
    Suppose that $W(s)$ satisfies the boundary cancellation condition in the sense that $W^{(k)}(0) = W^{(k)}(1) = 0$ for all $k \geq 1$. 
    Then, for any integer $T$ and any positive integer $k$, we have 
    \begin{equation}
        \left\|U(1)\ket{\psi} - P(1)U(1)\ket{\psi}\right\| \leq \frac{C_k}{T^k},  
    \end{equation}
    where $U(s)$ is the discrete evolution defined in~\cref{eqn:discrete_evolution}, $P(s)$ is the spectral projection of $W(s)$ onto $\sigma_P(s)$, and $C_k$ is a constant independent of $T$. 
\end{conjecture}

Similar to the continuous case,~\cref{lem:DAE_exp} bounds the leakage rate of the actual state out of the desired eigenspace. 
Compared with~\cref{lem:DAE_linear}, \cref{lem:DAE_exp} shows an error bound with super-polynomial convergence in $T$ under an additional assumption analogous to the continuous boundary cancellation condition. 
However,~\cref{lem:DAE_exp} does not give an explicit dependence on the spectral gap. 
It does not show how the pre-constant depends on the order $k$ either, which prevents us from obtaining an exponentially small error because there might be an exponential dependence on $k$ hidden in the constant $C_k$.

\subsection{Our results}

We now focus on the scenario with the additional boundary cancellation condition and establish the super-polynomial convergence result. 
Due to the lack of a discrete adiabatic theorem with super-polynomial convergence and explicit gap dependence at the same time, it is hard to simultaneously establish both the gap dependence and the precision dependence. 
Notice that if we want to track only the gap dependence, we can assume the tolerated error to be at a constant level and directly use~\cref{cor:exponential_linear} or~\cref{cor:Trotter_linear} in the general case, which gives an $\mathcal{O}(\Delta_*^{-3})$ dependence. 
Throughout this section, we assume that the spectral norm of the Hamiltonians $H_0$ and $H_1$ are bounded by $1$ and the minimum spectral gap to be bounded from below by a constant, and we focus on the dependence on $\epsilon$. 

For the first-order exponential integrator and product formulae, we can prove the following two theorems, which show that for any method it suffices to choose an $\mathcal{O}(1)$ time step size and the number of steps can be $\mathcal{O}(1/\epsilon^{o(1)})$. 
Remarkably, such results also hold even for first-order numerical methods. 

\begin{theorem}\label{thm:exp_exp}
    Consider~\cref{prob:state_prep} using exponential propagators~\cref{eqn:exponential_integrator_1st}. 
    Suppose that~\cref{lem:DAE_exp} is true, both $\|H_0\|$ and $\|H_1\|$ are bounded by $1$, $H(s)$ satisfies the gap condition with constant-level minimum gap, and the smooth interpolation function has $f^{(k)}(0)=f^{(k)}(1) = 0$ for all $k \geq 1$. 
    Then, in order to prepare an $\epsilon$-approximation of the eigenstate of $H_1$, it suffices to choose $h = 1$, and the evolution time and overall number of steps as 
    \begin{equation}
        T = T_{d} = \mathcal{O}\left( \frac{1}{\epsilon^{1/k}} \right)
    \end{equation}
    for any $k \geq 1$. 
\end{theorem}
\begin{proof}
    Let $W(s) = e^{-iH(s)}$ where $H(s) = (1-f(s))H_0 + f(s) H_1$. 
    We only need to verify that the walk operator $W(s)$ satisfies all the conditions of~\cref{lem:DAE_exp}. 
    Since the function $f(s)$ is smooth over $[0,1]$, the operator $W(s)$ is also smooth. 
    The walk operator $W(s)$ also satisfies the gap condition since $H(s)$ is assumed to satisfy the gap condition and $\|H(s)\| \leq 1$. 
    Finally, the boundary cancellation condition of $f(s)$ implies the boundary cancellation of $W(s)$. 
    Therefore, our claim directly follows from~\cref{lem:DAE_exp}. 
\end{proof}

\begin{theorem}\label{thm:trotter_exp}
    Consider~\cref{prob:state_prep} using simplified product formula as in~\cref{eqn:high_order_trotter_simplified}. 
    Suppose that~\cref{lem:DAE_exp} is true, both $\|H_0\|$ and $\|H_1\|$ are bounded by $1$, $H(s)$ satisfies the gap condition with constant-level minimum gap, and the smooth interpolation function $f(s)$ has $f^{(k)}(0)=f^{(k)}(1) = 0$ for all $k \geq 1$. 
    Then, in order to prepare an $\epsilon$-approximation of the eigenstate of $H_1$, it suffices to choose $h = \mathcal{O}(1)$, and the evolution time and overall number of steps as 
    \begin{equation}
        T = \mathcal{O}\left( \frac{1}{\epsilon^{1/k}}\right),  \quad T_{d} = \mathcal{O}\left( \frac{1}{\epsilon^{1/k}}\right)
    \end{equation}
    for any $k \geq 1$. 
\end{theorem}

\begin{proof}
    Define the walk operator 
    \begin{equation}
        W(s) = \prod_{k=0}^{K} \exp\left( -i\beta_{p,k} h f(s) H_1\right)  \exp\left( -i\alpha_{p,k}h (1-f(s)) H_0 \right). 
    \end{equation}
    Since the function $f(s)$ is smooth over $[0,1]$, the operator $W(s)$ is also smooth. 
    Using~\cref{lem:gap_Trotter_high_order}, the walk operator $W(s)$ also satisfies the gap condition as $H(s)$ does by choosing a sufficiently small (but independent of $\epsilon$) time step size $h$. 
    Finally, according to the chain rule, each term in the derivatives of $W(s)$ contains the derivatives of $f(s)$ as a multiplicative factor, which implies that $W(s)$ also satisfies the boundary cancellation condition. 
    Therefore, our claim directly follows from~\cref{lem:DAE_exp}. 
\end{proof}

As a side product, we can obtain an improved error bound on the numerical discretization errors as well. 
The idea is to use the triangle inequality 
to bound the numerical discretization error by the sum of the continuous and discrete adiabatic errors. 
Here we state and prove an error bound for the first-order product formula. 
Similar results for other numerical methods can be established in the same way. 

\begin{corollary}\label{cor:improved_trotter}
    Let $\ket{\psi(T)}$ be the solution of~\cref{eqn:AQC_dynamics} at the final time, and $\ket{\widetilde{\phi}}$ be the numerical solution obtained by the first-order product formula with step size $\mathcal{O}(1)$. 
    Suppose that~\cref{lem:DAE_exp} is true, both $\|H_0\|$ and $\|H_1\|$ are bounded by $1$, the Hamiltonian $(1-s)H_0+sH_1$ satisfies the gap condition, the spectrum of interest only consists of one simple eigenvalue, and the smooth interpolation function $f(s)$ has $f^{(k)}(0)=f^{(k)}(1) = 0$ for any $k \geq 1$. 
    Then, for any $k \geq 1$, we have 
    \begin{equation}
        \left\|\ket{\widetilde{\phi}}\bra{\widetilde{\phi}} - \ket{\psi(T)}\bra{\psi(T)}\right\| \leq \mathcal{O}\left(\frac{1}{T^k}\right). 
    \end{equation}
\end{corollary}
\begin{proof}
    First, according to the continuous high-order adiabatic theorem, which we detail in~\cref{lem:AT_exp} in~\cref{app:continuous_adiabatic}, there exists an eigenstate $\ket{u}$ of $H_1$ corresponding to the eigenvalue of interest such that for any $k$, 
    \begin{equation}
        \left\| \ket{\psi(T)} - \ket{u} \right\| \leq \mathcal{O}\left(\exp\left(-cT^{\frac{1}{1+\alpha}}\right)\right) \leq \mathcal{O}\left(\frac{1}{T^k}\right). 
    \end{equation}
    Similarly, for the numerical solution, according to~\cref{thm:trotter_exp}, there exists an eigenstate $\ket{v}$ of $e^{-i h H_1}$ corresponding to the eigenvalue of interest such that for any $k$, 
    \begin{equation}
        \left\| \ket{\widetilde{\phi}} - \ket{v} \right\| \leq  \mathcal{O}\left(\frac{1}{T^k}\right). 
    \end{equation}
    Notice that $\ket{u}$ and $\ket{v}$, if not exactly the same, only differ by a phase factor, since we assume that $\|H_1\| \leq 1$ and the spectrum of interest only consists of one simple eigenvalue. 
    Therefore, the triangle inequality gives 
    \begin{align}
        \left\|\ket{\widetilde{\phi}}\bra{\widetilde{\phi}} - \ket{\psi(T)}\bra{\psi(T)}\right\| &\leq \left\| \ket{\widetilde{\phi}}\bra{\widetilde{\phi}} - \ket{v}\bra{v}\right\| + \left\| \ket{\psi(T)}\bra{\psi(T)} - \ket{u}\bra{u}\right\| \nonumber\\
        & \leq 2 \left\| \ket{\widetilde{\phi}} - \ket{v} \right\| + 2 \left\| \ket{\psi(T)} - \ket{u} \right\|  \nonumber \\
        & \leq \mathcal{O}\left(\frac{1}{T^k}\right). 
    \end{align}
\end{proof}

We compare~\cref{cor:improved_trotter} with the standard Trotter error bound in two aspects. 
First, standard Trotter error bounds are typically in terms of operator norm, and thus simultaneously bound the fidelity error and the phase error. 
However, our result in~\cref{cor:improved_trotter} only bounds the difference in the density matrix, and there is nothing we can conclude from~\cref{cor:improved_trotter} about the error within the phase. 
Nevertheless, in the scenario of near adiabatic evolution, the phase factor is usually not of importance and our bound on the density matrices is usually sufficient. 
Second, in terms of the density matrices, our error bound significantly improves the standard ones. 
Our Trotter error bound for near adiabatic evolution is $\mathcal{O}(T^{-k})$ for an arbitrarily large $k$, while the standard error bound of first-order Trotter with time step size $\mathcal{O}(1)$ is only $\mathcal{O}(1)$. 
We remark that such an improvement requires additional conditions, including the regularity conditions and the Hamiltonian gap conditions specified in~\cref{cor:improved_trotter}.

\section{Gap conditions of product formula and discrete adiabatic evolution for gapless Hamiltonian}\label{sec:gap}

Our analysis requires the gap condition of the local numerical propagator due to the usage of the discrete adiabatic theorems. 
For product formula, if only the gap in the Hamiltonian $H(s)$ is assumed rather than the product walk operator, we can reduce the time step size $h$ to make sure that the product walk operator also has a spectral gap (as shown in~\cref{lem:gap_Trotter_high_order} and~\cref{cor:Trotter_linear}). 

However, when the time step size is large (that is, as large as $\mathcal{O}(1)$), the relation between the gap of the product walk operator and the gap of the Hamiltonian is not straightforward yet since the perturbation theory does not directly apply. 
We show by toy numerical examples that, in the worst case, even if $\|H_0\|$ and $\|H_1\|$ are bounded, the gap condition of the Hamiltonian does not necessarily imply that of the first-order Trotter operator, and vice versa. 

\begin{figure}
    \centering
    \includegraphics[width=0.45\textwidth]{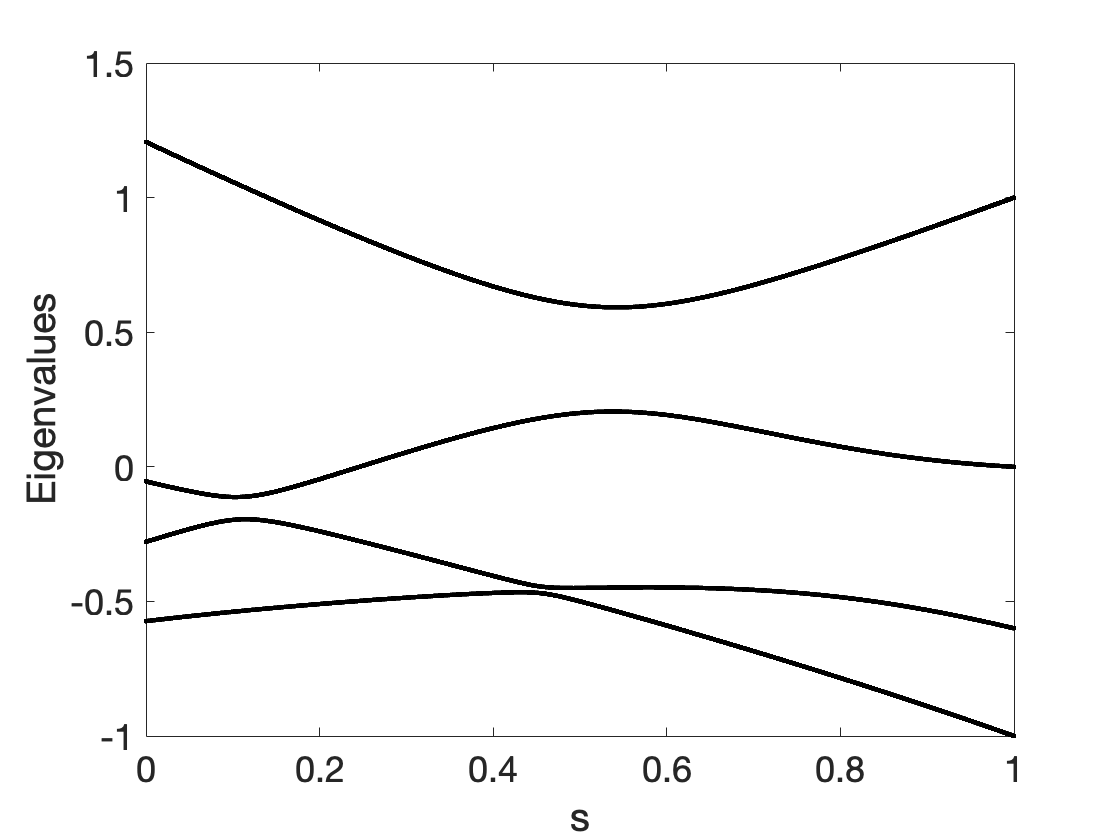}
    \includegraphics[width=0.45\textwidth]{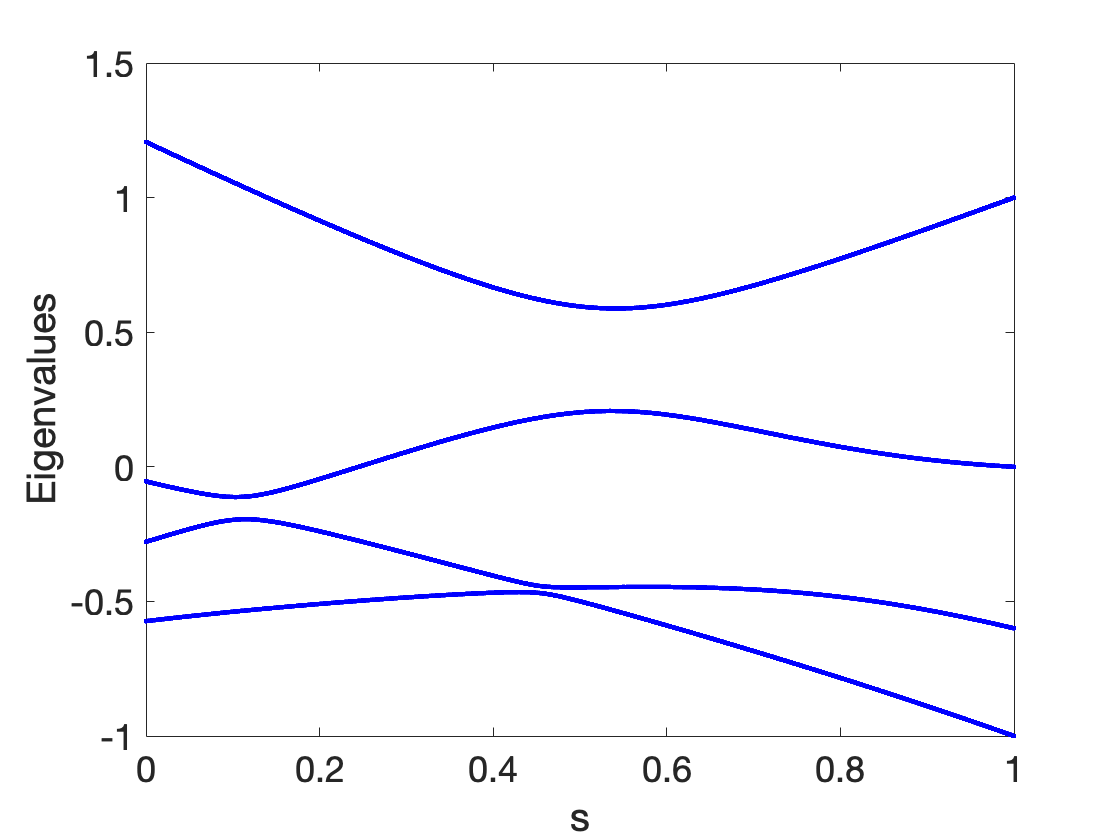}
    \caption{Left: spectrum of $H(s)$. Right: spectrum of $i\log(W(s))$. Here $\epsilon = 0.05$. }
    \label{fig:spectrum}
\end{figure}

The first example is that the gap of $H(s)$ remains constant while the gap of $W(s)$ closes. 
We consider a $4$-dimensional system, and $H_0$ and $H_1$ are constructed as follows. 
Let $Q$ be the eigenbasis of the matrix
\begin{equation}
    \left(\begin{array}{cccc}
        2 & -1 & 0 & 0\\
        -1 & 2 & -1 & 0 \\
        0 & -1 & 2 & -1 \\
        0 & 0 & -1 & 2
    \end{array}\right),
\end{equation}
$D = \text{diag}(-0.5,-0.5+\epsilon,0.2,0.6)$ for a small parameter $\epsilon\geq 0$, and define $\widetilde{W} = Qe^{-iD}Q^{\dagger}$. 
Then we choose $H_1 = \text{diag}(-1,-0.6,0,1)$, $H_0 = 2i\log(e^{iH_1/2} \widetilde{W} )$, and the scheduling function is linear (\emph{i.e.}, $f(s) = s$). 
By definition, $W(0.5)$ is exactly $\widetilde{W}$, so we can control the gap of $W(s)$ by tuning $\epsilon$. 
In particular, the gap vanishes when we choose $\epsilon = 0$. 
As an illustration, the spectrum of $H(s)$ and $W(s)$ are shown in~\cref{fig:spectrum}. 
Notice that the gap of $W(s)$ is measured using the distance in $\mathbb{S}^1$ (\emph{i.e.}, the spectral gap of $i\log(W(s))$). 
At a first glance of~\cref{fig:spectrum}, the spectrum of $H(s)$ and $i\log(W(s))$ are surprisingly very similar. 
However, a careful computation yields small differences. 
We numerically compute the spectral gap between ground and first excited states, and results with different $\epsilon$ are shown in~\cref{tab:num_gap_1}. 
As $\epsilon$ decreases, the gap of $W(s)$ decreases, but the gap of $H(s)$ remains at a constant level (although as small as $\sim 10^{-3}$) when $\epsilon$ is small enough. 
In the extreme case where $\epsilon = 0$, the gap of $W(s)$ becomes $0$ as expected (numerical gap is $\sim 10^{-16}$, reaching machine epsilon), while the gap of $H(s)$ is still $\sim 10^{-3}$. 

\begin{table}[tbh]
    \centering
    \scalebox{0.88}{
    \begin{tabular}{c|c|c|c|c|c|c|c|c|c|c|c}\hline\hline 
        $\epsilon$ & $10^{-1}$ & $5\times 10^{-2}$ & $2\times 10^{-2}$ & $10^{-2}$ & $5\times 10^{-3}$ & $2\times 10^{-3}$ & $10^{-3}$ & $5\times 10^{-4}$ & $2\times 10^{-4}$ & $10^{-4}$ & 0 \\ \hline
        Gap of $H(s)$ & $5.1 \times 10^{-2}$ & $2.3\times 10^{-2}$ & $7.9\times 10^{-3}$ & $3.0\times 10^{-2}$ & $5.6\times 10^{-4}$ & $8.9\times 10^{-4}$ & $1.4\times 10^{-3}$ & $1.6\times 10^{-3}$ & $1.8\times 10^{-3}$ & $1.8\times 10^{-3}$ & $1.9\times 10^{-3}$ \\ \hline
        Gap of $W(s)$ & $5.2\times 10^{-2}$ & $2.5\times 10^{-2}$ & $9.7\times 10^{-3}$ & $4.8\times 10^{-2}$ & $2.6\times 10^{-3}$ & $9.5\times 10^{-4}$ & $4.8\times 10^{-4}$ & $2.4\times 10^{-4}$ & $1.0\times 10^{-4}$ & $5.2\times 10^{-5}$ & $1.1\times 10^{-16}$ \\ \hline\hline
    \end{tabular}
    }
    \caption{Spectral gaps of $H(s)$ and $W(s)$ in the first example. Here the gap reported is the minimal gap over equi-distant points with step size $10^{-4}$. }
    \label{tab:num_gap_1}
\end{table}

Similarly, we can construct a second example where the gap of $W(s)$ remains constant while the gap of $H(s)$ closes. 
Let $Q$ and $D = \text{diag}(-0.5,-0.5+\epsilon,0.2,0.6)$ be the same as in the first example. 
Define $\widetilde{H} = QDQ^{\dagger}$, $H_1 = \text{diag}(-1,-0.6,0,1)$, $H_0 = 2\widetilde{H} - H_1$, and the scheduling function $f(s) = s$. 
Again by definition, $H(0.5)$ is exactly $\widetilde{H}$ and thus can be gapless when $\epsilon = 0$. 
\cref{tab:num_gap_2} shows the spectral gaps with different $\epsilon$, and we can observe that, as $\epsilon$ decreases, the gap of $H(s)$ vanishes while the gap of $W(s)$ remains at a constant level. 

\begin{table}[tbh]
    \centering
    \scalebox{0.88}{
    \begin{tabular}{c|c|c|c|c|c|c|c|c|c|c|c}\hline\hline 
        $\epsilon$ & $10^{-1}$ & $5\times 10^{-2}$ & $2\times 10^{-2}$ & $10^{-2}$ & $5\times 10^{-3}$ & $2\times 10^{-3}$ & $10^{-3}$ & $5\times 10^{-4}$ & $2\times 10^{-4}$ & $10^{-4}$ & 0 \\ \hline
        Gap of $H(s)$ & $5.1\times 10^{-2}$ & $2.5\times 10^{-2}$ & $9.6\times 10^{-3}$ & $4.7\times 10^{-3}$ & $2.4\times 10^{-3}$ & $9.4\times 10^{-4}$ & $4.7\times 10^{-4}$ & $2.3\times 10^{-4}$ & $1.0\times 10^{-4}$ & $5.1\times 10^{-5}$ & $3.3\times 10^{-16}$ \\ \hline
        Gap of $W(s)$ & $5.3\times 10^{-2}$ & $2.6\times 10^{-2}$ & $1.1\times 10^{-2}$ & $6.6\times 10^{-3}$ & $4.2\times 10^{-3}$ & $2.8\times 10^{-3}$ & $2.4\times 10^{-3}$ & $2.1\times 10^{-3}$ & $2.0\times 10^{-3}$ & $1.9\times 10^{-3}$ & $1.9\times 10^{-3}$ \\ \hline\hline
    \end{tabular}
    }
    \caption{Spectral gaps of $H(s)$ and $W(s)$ in the second example. Here the gap reported is the minimal gap over equi-distant points with step size $10^{-4}$. }
    \label{tab:num_gap_2}
\end{table}

The second example where the gap of $H(s)$ closes while the gap of $W(s)$ of the first-order Trotter operator remains constant suggests the possibility of efficient adiabatic optimization for systems with exponentially small gap or even gapless systems. 
Specifically, for gapless systems, we can try first-order Trotter with time step size $1$. 
Thanks to the discrete adiabatic theorems, such an approach will work and approximate the desired eigenstate if the product walk operators satisfy the gap condition, which is possible even for gapless systems (and we have observed this in our second example). 
Meanwhile, decreasing the time step sizes might worsen the computation because it makes the numerical evolution closer to the continuous dynamics which is known to be away from the desired eigenstate.

\begin{figure}
    \centering
    \includegraphics[width = 0.45\textwidth]{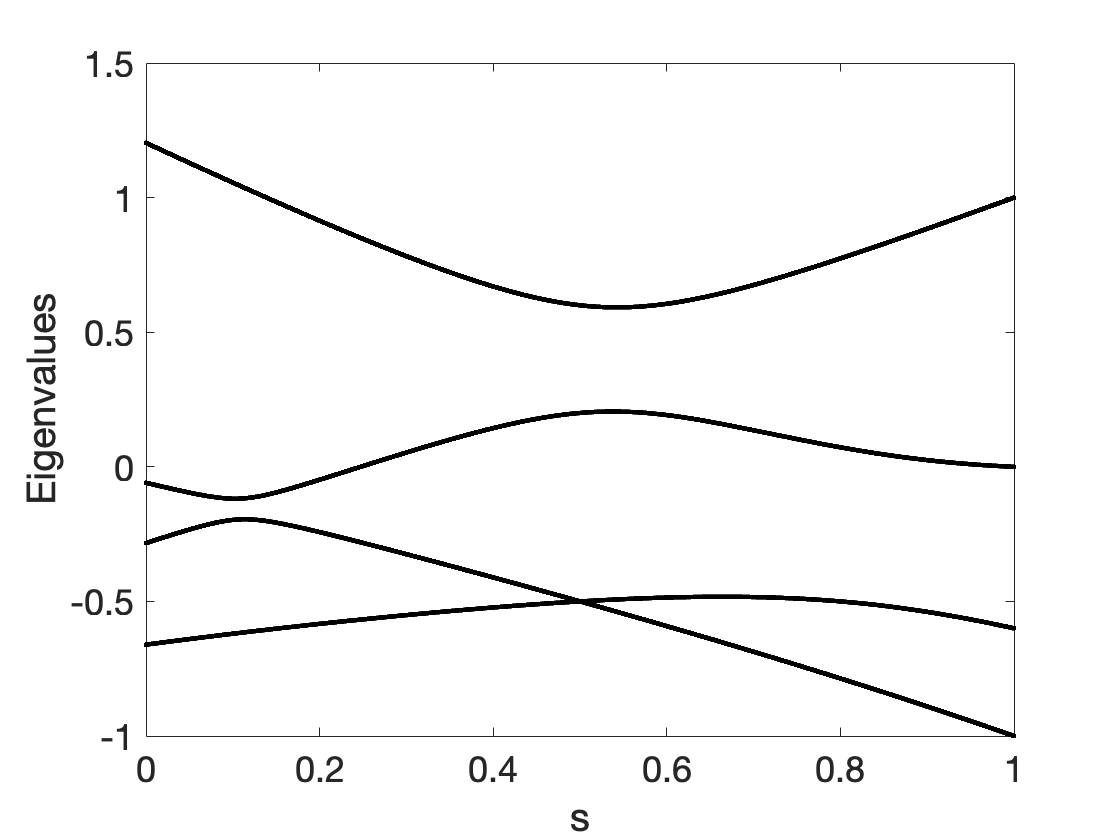}
    \includegraphics[width = 0.45\textwidth]{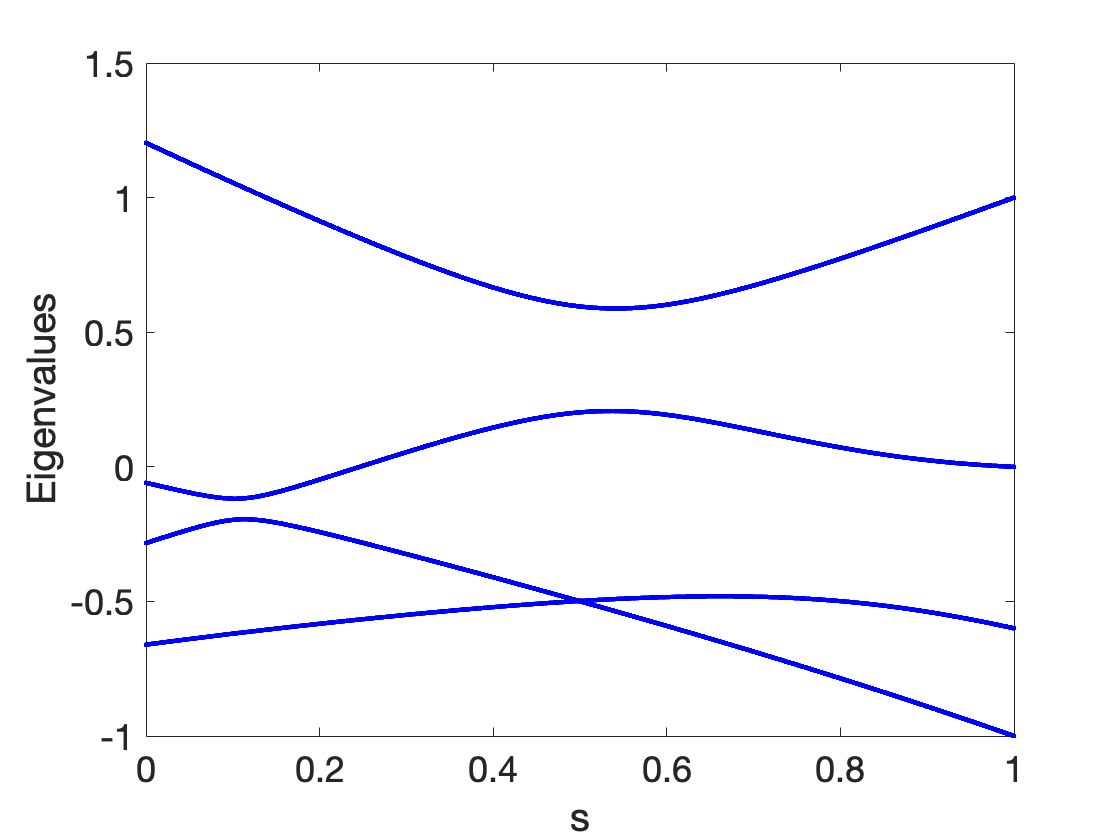}\\
    \includegraphics[width = 0.45\textwidth]{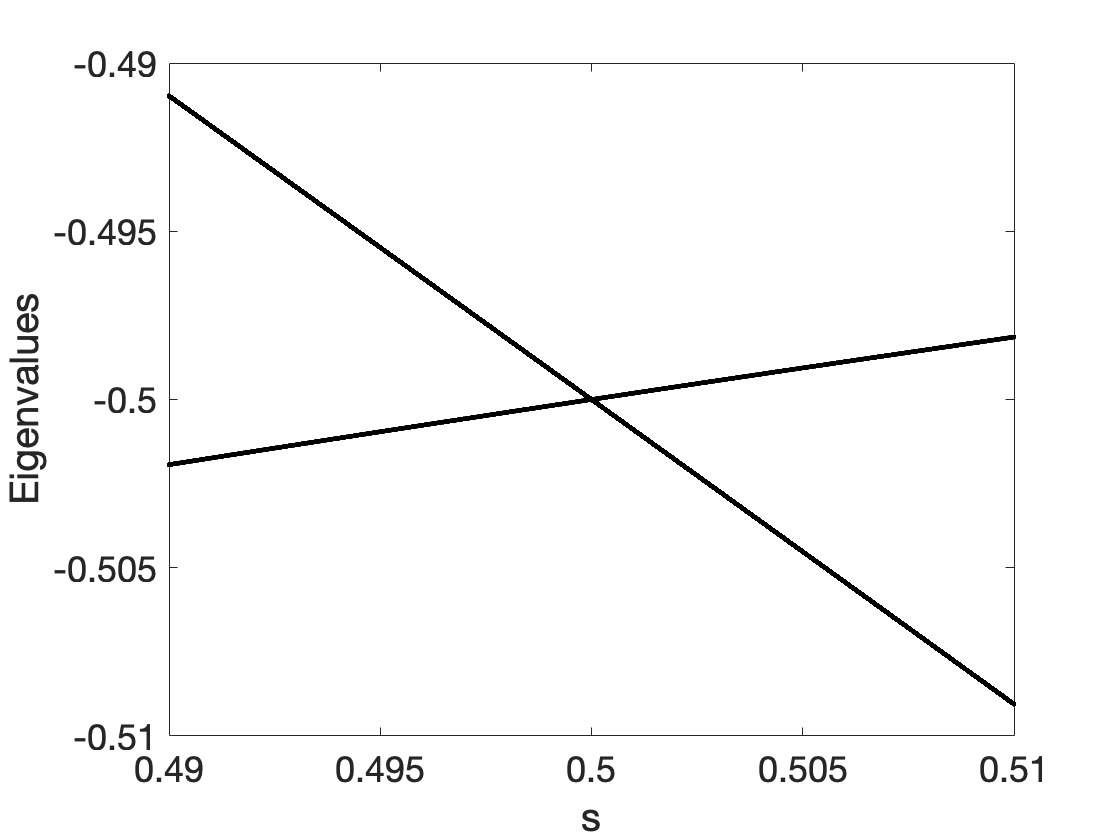}
    \includegraphics[width = 0.45\textwidth]{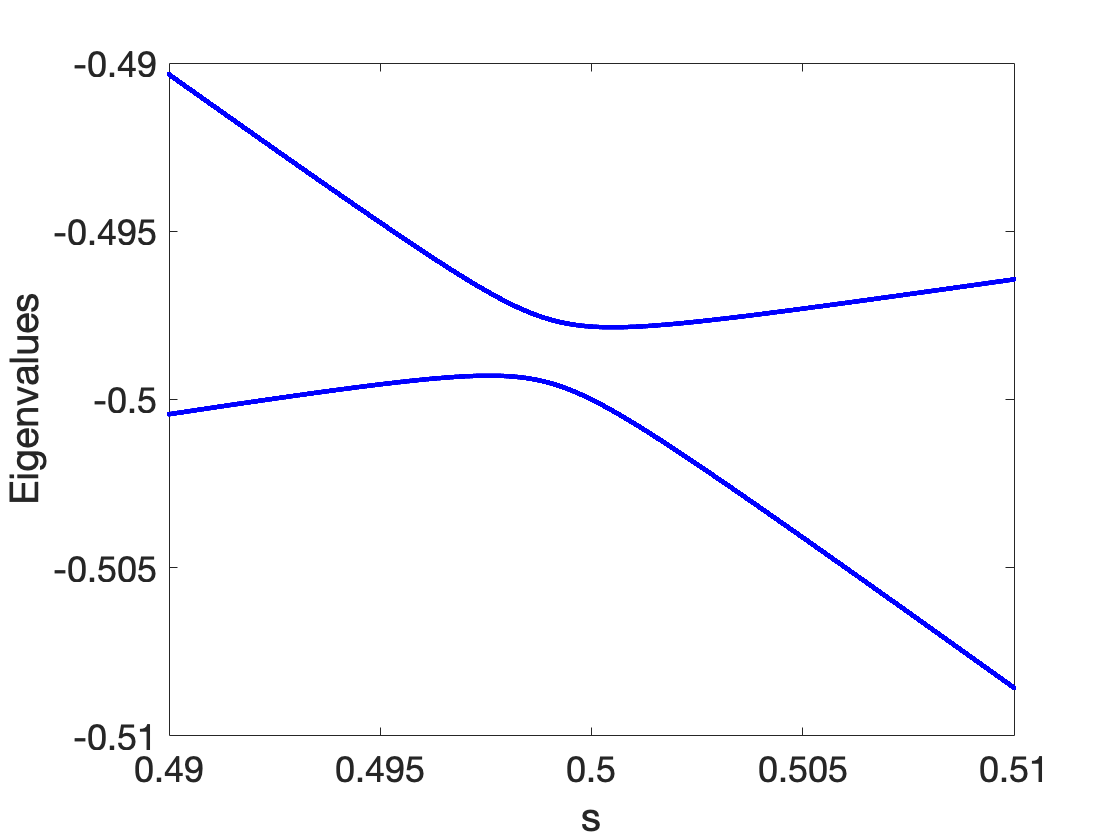}
    \caption{Spectrums of $H(s)$ and $i \log(W(s))$ of first-order Trotter with time step size $1$ in the second toy example with $\epsilon = 0$. 
    Here the top left is of $H(s)$, the top right is of $i \log(W(s))$, and the bottom two are zoom-in views. }
    \label{fig:spectrum_second_toy}
\end{figure}

\begin{figure}
    \centering
    \includegraphics[width=0.45\textwidth]{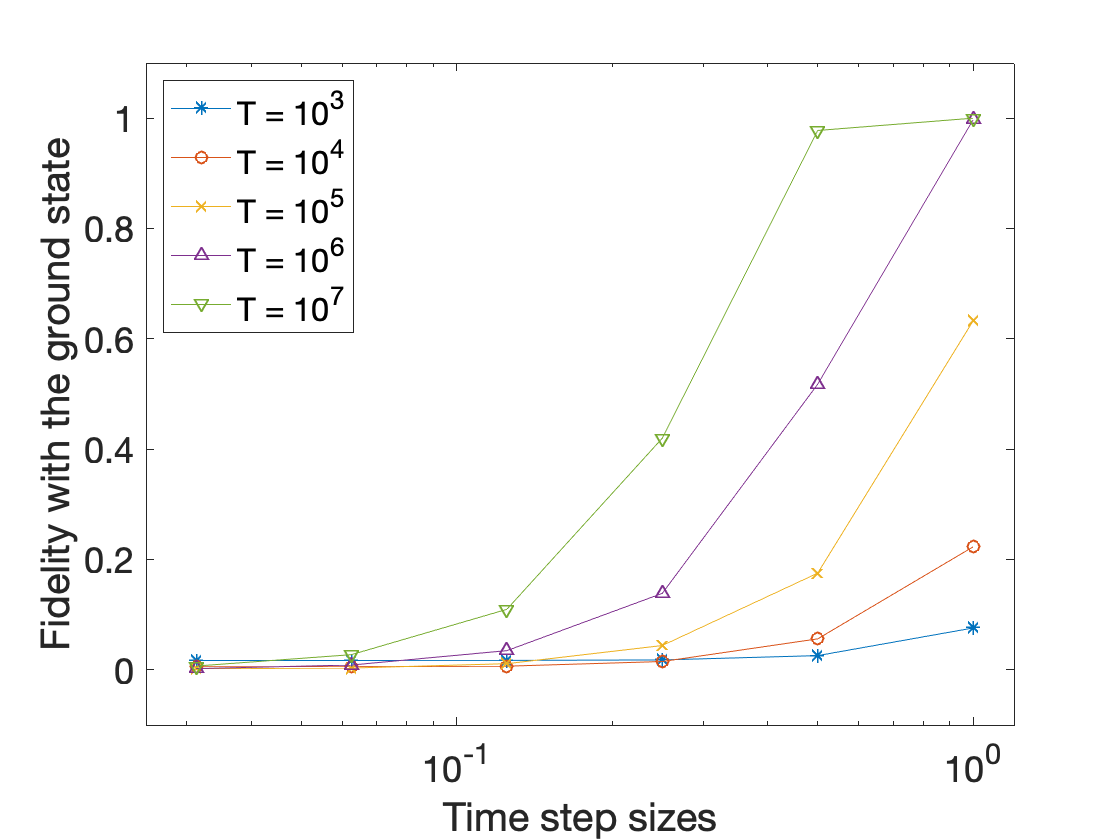}
    \includegraphics[width=0.45\textwidth]{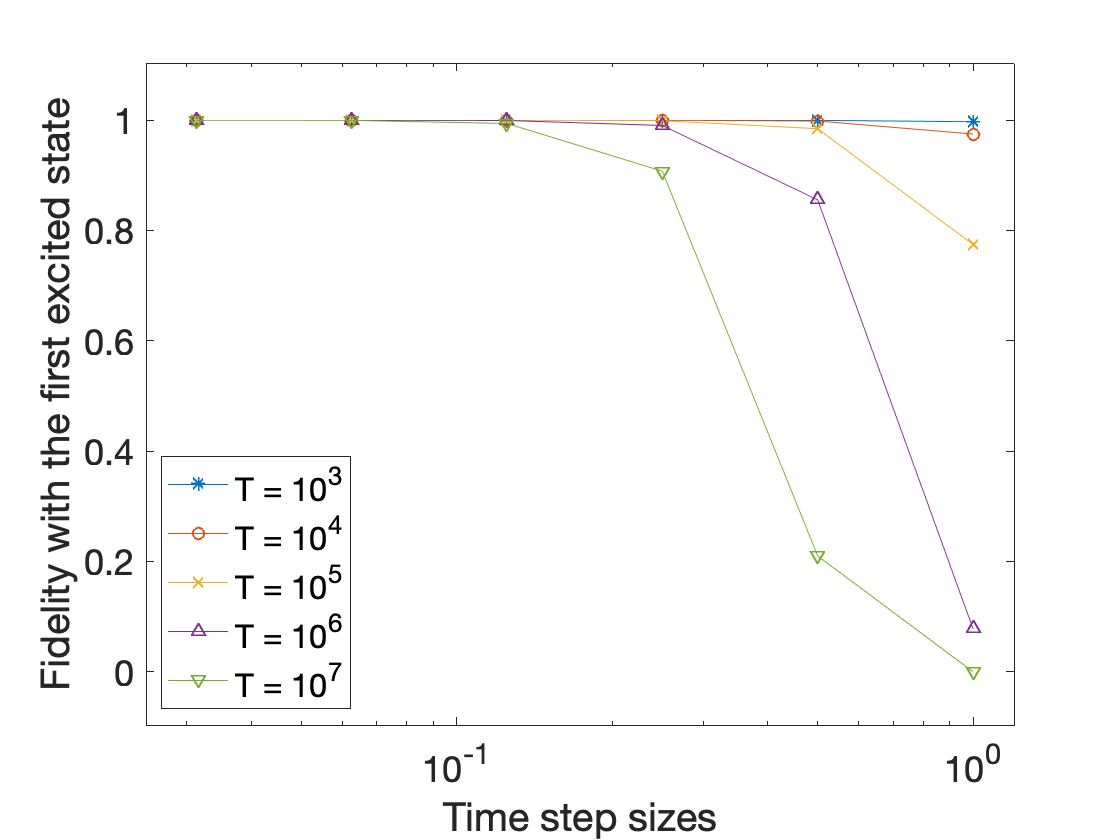}
    \caption{Fidelity with the ground state (left) and the first excited state (right) of the numerical solutions from first-order Trotter method with different time $T$ and time step sizes $h$. }
    \label{fig:fidelity_second_toy}
\end{figure}

We illustrate this phenomenon using our second toy model again. 
Here we directly choose $\epsilon = 0$, so the Hamiltonian $H(s)$ is gapless and the walk operator $W(s)$ of first-order Trotter with time step size $1$ still satisfies the gap condition. 
~\cref{fig:spectrum_second_toy} shows the spectrums of $H(s)$ and $
W(s)$, and confirms that $H(s)$ is gapless and $W(s)$ is gapped (though the gap is indeed very small). 
We choose time $T = 10^k$ for $k = 3,4,5,6,7$, and implement first-order Trotter method with time step sizes $h = 2^{-k}$ for $k = 0,1,2,3,4,5$. 
In~\cref{fig:fidelity_second_toy}, we plot the fidelity $|\braket{\widetilde{\phi}}{e_j}|$ for $j = 0,1$ where $\ket{\widetilde{\phi}}$ is our final actual state and $\ket{e_j}$ is the $j$-th eigenstate of the final Hamiltonian $H_1$. 
When the time step size is fixed to be $1$, the actual state will better approximate the ground state of $H_1$ as $T$ becomes larger. 
This can be explained by the discrete adiabatic theorems and the analysis we establish in this work. 
However, if we decrease the time step size, then the numerical solution will go away from the ground state and tend to the first excited state instead. 
This is because, as the time step size decreases, the numerical solution will be closer to the solution of the continuous dynamics governed by $H(s)$, which is gapless and does not even approximately follow the ground eigenpath. 
The observation that it goes to the first excited state can be explained according to the adiabatic theorem without gap condition~\cite{AvronElgart1999} since the ground state of the initial Hamiltonian $H_0$ and the first excited state of the final Hamiltonian $H_1$ is smoothly connected. 
In summary, this toy example shows that for gapless Hamiltonian, continuous adiabatic evolution does not approximate the ground state but the Trotterization with $\mathcal{O}(1)$ time step size may still work as long as the product walk operator satisfies the gap condition.

\section{Optimized gap dependence with application to Grover search and comparison with QAOA}\label{sec:Grover}

We have shown that the time step size in discretizing near-adiabatic dynamics can be larger as previously expected by viewing the numerical integrator as the set of discrete adiabatic walk operators. 
In our analysis, we mainly focus on the improved scalings in $\epsilon$, and we simply bound the time-dependent spectral gap by the minimal gap, resulting in a generally super-linear dependence on the gap. 
However, the gap dependence can be improved to be linear if we are given sufficient information of the simultaneous gap and carefully choose the scheduling function $f(s)$ according to the size of the gap. 
This has been observed in several applications of the continuous adiabatic evolution as well as using the first-order exponential integrator to construct optimal linear system solver. 
Here we demonstrate by the adiabatic Grover search example that the linear gap dependence can also be achieved in first-order Trotter method as well.

We consider the unstructured search problem in $N$ dimensional Hilbert space. 
Let $\left\{\ket{j}\right\}_{j=0}^{N-1}$ denote the computational basis, and $\mathcal{M} \subset \left\{\ket{j}\right\}_{j=0}^{N-1}$ is the set of the marked states with $|\mathcal{M}| = M$. 
We assume $N \geq 2M$. 
Starting with the uniform superposition of the computational basis, our goal is to find the uniform superposition of all marked states. 
Here we do not assume $\emph{a priori}$ knowledge on $M$ since estimating the bound of $M$ can be costly -- instead we are interested in designing robust algorithm that does not use the information of $M$ in implementing the algorithm and automatically becomes more accurate when $M$ is larger.

Adiabatic formulation~\cite{RevModPhys.90.015002} takes the initial Hamiltonian to be 
\begin{equation}
    H_0 = I - \ket{u}\bra{u}, 
\end{equation}
where $\ket{u} = \frac{1}{\sqrt{N}}\sum_{j=0}^{N-1} \ket{j}$ is the uniform superposition, and the final Hamiltonian to be 
\begin{equation}
    H_1 = I - \sum_{j\in\mathcal{M}} \ket{j}\bra{j}. 
\end{equation}
Notice that both $H_0$ and $H_1$ only effectively act on a two-dimensional subspace. 
The orthonormal basis of this subspace is $\left\{\ket{e_0},\ket{e_1}\right\}$ where 
\begin{equation}
    \ket{e_0} = \frac{1}{\sqrt{M}} \sum_{j\in\mathcal{M}} \ket{j}, \quad \ket{e_1} = \frac{1}{\sqrt{N-M}} \sum_{j\notin \mathcal{M}}\ket{j}, 
\end{equation}
and the corresponding matrix representations of $H_0$ and $H_1$ are 
\begin{equation}\label{eqn:Grover_H0_H1_2dim}
    H_0 = \left(\begin{array}{cc}
        1-M/N & -\sqrt{M(N-M)}/N  \\
        -\sqrt{M(N-M)}/N & M/N
    \end{array}\right), 
    \quad H_1 = \left(
    \begin{array}{cc}
        0 & 0 \\
        0 & 1
    \end{array}
    \right). 
\end{equation}
Furthermore, $\ket{u}$ is the ground state of $H_0$ and our target state $\ket{e_0}$ is the ground state of $H_1$. 
Let $H(s) = (1-f(s))H_0 + f(s)H_1$ where $f(s)$ is the scheduling function such that $f(0) = 0$ and $f(1) = 1$. 
We may solve the time evolution governed by $H(s)$ with the initial state $\ket{u}$ to approximate the target state. 
The corresponding digital simulation algorithm, using first-order Trotter method with time step size $1$ (so $T = T_d$ in this case), goes as 
\begin{equation}
    \prod_{j=0}^{T-1} \exp\left(-i f(j/T) H_1\right)\exp\left(-i (1-f(j/T)) H_0\right) \ket{u}. 
\end{equation}

For time-dependent Hamiltonian $H(s) = (1-f(s))H_0 + f(s)H_1$, straightforward computations yield that two relevant eigenvalues which determine the diabatic errors are~\cite{RevModPhys.90.015002} 
\begin{equation}
    \lambda_{\pm} = \frac{1}{2} \pm \frac{1}{2}\sqrt{(1-2f(s))^2 + \frac{4M}{N}f(s)(1-f(s))}. 
\end{equation}
The gap of $H(s)$ is 
\begin{equation}\label{eqn:Grover_gap_H}
    \Delta_{H}(s;N,M,f) = \sqrt{(1-2f(s))^2 + \frac{4M}{N}f(s)(1-f(s))}. 
\end{equation}
However, in the discrete evolution, the gap of the walk operator $W(s) = \exp\left(-i f(s) H_1\right)\exp\left(-i (1-f(s)) H_0\right)$ is what determines the scaling of $T$. 
Since $W(s)$ is represented by a $2$-dimensional matrix, it is possible to explicitly compute its eigenvalues and spectral gap (see~\cref{app:Grover_gap}). 
However, the explicit expression of the gap turns out to be complicated for us to estimate the adiabatic error using~\cref{lem:DAE_linear}. 
To facilitate further computations, we alternatively relate the gap of $W(s)$ to that of $H(s)$ in the following lemma, for which the proof can be found in~\cref{app:Grover_gap}. 
\begin{lemma}\label{lem:Grover_gap_walk}
    Consider the adiabatic Grover search problem. 
    Let $H(s)$ be the Hamiltonian with the gap $\Delta_{H}(s;N,M,f)$ specified in~\cref{eqn:Grover_gap_H}, and $\Delta_{W}(s;N,M,f)$ denote the gap of the walk operator $W(s)$. 
    Then 
    \begin{equation}\label{eqn:Grover_gap_Trotter}
        \Delta_{W}(s;N,M,f) \geq \frac{2}{3} \Delta_{H}(s;N,M,f). 
    \end{equation}
\end{lemma}

The minimum gap appears when $s = f^{-1}(1/2)$ and $ \sim \sqrt{M/N}$.  
As a result, a linear interpolation ($f(s) = s$) leads to the non-optimal scaling $T \geq N/M$ which is at least quadratically worse than Grover's algorithm. 
To recover the Grover speedup, we need to choose a carefully designed scheduling function such that it slows down when the gap closes. 

We consider two options. 
The first one is to use the scheduling function designed for $M = 1$, so the scheduling function is independent of general values of $M$. 
The intuitive reason why this may work is that the case $M=1$ is the hardest one with smallest gap near $s = f^{-1}(1/2)$, so the resulting scheduling function is also sufficiently slow near the smallest gap for general $M > 1$. 
One example is defined through the differential equation 
\begin{equation}\label{eqn:Grover_scheduling_def}
    \partial_s f(s) = d_{N,p}\Delta_{H}(s;N,1,f(s))^p, \quad f(0) = 0, 
\end{equation}
where the normalization constant $d_{N,p}$ is chosen such that $f(1) = 1$ and has the expression 
\begin{equation}\label{eqn:Grover_def_dN}
    d_{N,p} = \int_0^1 \Delta_{H}(s;N,1,s)^{-p} ds. 
\end{equation}
Such a scheduling function was originally proposed in~\cite{RolandCerf2002} with $p = 2$ and generalized in~\cite{jansen2007bounds} with $1<p<2$ for continuous AQC. 
The work~\cite{DalzellYoderChuang2017} considers the time discretization error with $p = 2$ and shows that the time step size can be as large as $1$ via efficient Hamiltonian approach. 
For $1 \leq p < 2$, we give the error bound in the following theorem, and the proof can be found in~\cref{app:Grover_proof_linear}.

\begin{theorem}\label{thm:Grover_adiabatic}
     Consider solving the Grover search problem in $N$ dimension with $M$ marked states using AQC discretized by first-order product formula with time step size $1$. 
     Suppose that $M$ is unknown. 
     \begin{enumerate}
         \item Suppose that $T \geq \mathcal{O}(\log(N))$ and the scheduling function is defined in~\cref{eqn:Grover_scheduling_def} with $p=1$. 
         Then the overall error is bounded by 
         \begin{equation}
             \mathcal{O}\left( \frac{\log(N)}{T} \sqrt{\frac{N}{M}}  \right). 
         \end{equation}
         \item Suppose that $T \geq \mathcal{O}(N^{\frac{p-1}{2}})$ and the scheduling function is defined in~\cref{eqn:Grover_scheduling_def} with $1<p<2$. 
         Then the overall error is bounded by 
         \begin{equation}
             \mathcal{O}\left(\frac{1}{T} \frac{\sqrt{N}}{M^{1-\frac{p}{2}}}  + \frac{N^{p-1}}{T^2}\max\left\{ 1,\frac{N^{3/2-p}}{M^{3/2-p}} \right\}  \right). 
         \end{equation}
     \end{enumerate}
\end{theorem}

~\cref{thm:Grover_adiabatic} implies that, to achieve an $\epsilon$ error,  the smallest number of Trotter steps required is $\mathcal{O}\left(\sqrt{\frac{N}{M}} \frac{\log(N)}{\epsilon}\right)$ for $p = 1$, and $\mathcal{O}\left( \frac{\sqrt{N}}{M^{1-p/2}} \frac{1}{\epsilon} \right)$ for $1<p<2$. 
The discrete adiabatic evolution with $p=1$ is nearly optimal (up to a logarithmic factor) in both $N$ and $M$, and that with $1<p<2$ is optimal in $N$ but not optimal in $M$. 
Notice that these critical points are only for comparison with other algorithms and not for designing algorithms since we do not know $M$ \emph{a priori} -- instead we will use sufficiently large $T$ and the error becomes automatically smaller if $T$ is too large.

Now we consider another scheduling function satisfying boundary cancellation condition. 
We also want it to slow down at $s = 1/2$ to preserve the Grover speedup. 
We start with the function
\begin{equation}\label{eqn:Grover_glue_BC}
    g(s) = c_e^{-1} \int_0^s \exp \left( - \frac{1}{s'(1-s')}\right) ds', 
\end{equation}
where 
\begin{equation}\label{eqn:def_ce}
    c_e = \int_0^1 \exp \left( - \frac{1}{s'(1-s')}\right) ds'. 
\end{equation}
One can show that $g(0) = 0$, $g(1) = 1$, and $g^{(k)}(0) = g^{(k)}(1) = 0$ for all $k \geq 1$. 
Consider
\begin{equation}\label{eqn:Grover_scheduling_def_BC}
    f(s) = \begin{cases}
     \frac{1}{2} g(2s), & \quad s \in [0, \frac{1}{2}], \\
     \frac{1}{2} + \frac{1}{2}g(2s-1), & \quad s \in (\frac{1}{2}, 1], 
    \end{cases}
\end{equation}
which is illustrated in~\cref{fig:Grover_scheduling_exp}. 
\begin{figure}
    \centering
    \includegraphics[width = 0.6\textwidth]{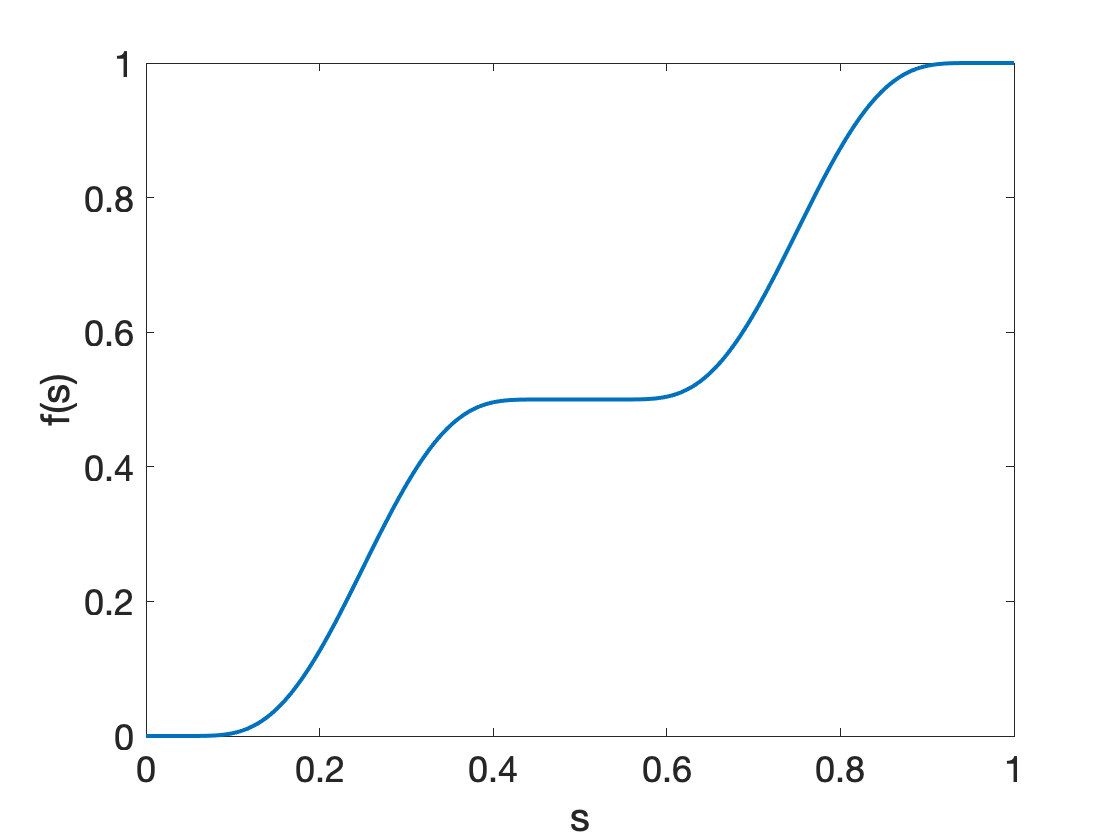}
    \caption{The scheduling function in robust Grover search with boundary cancellation. }
    \label{fig:Grover_scheduling_exp}
\end{figure}
Intuitively, we symmetrically connect two $g$'s and the resulting function $f(s)$ is smooth, satisfies the boundary cancellation and becomes slow in the middle. 
With this scheduling function~\cref{eqn:Grover_scheduling_def_BC}, we can bound the error as follows, and the proof can be found in~\cref{app:Grover_proof_exp}. 

\begin{theorem}\label{thm:Grover_adiabatic_BC}
    Consider solving the Grover search problem in $N$ dimension with $M$ marked states using AQC with the scheduling function~\cref{eqn:Grover_scheduling_def_BC}, discretized by first-order product formula with time step size $1$. 
    Suppose that~\cref{lem:DAE_exp} is true, and $T \geq \mathcal{O}(\sqrt{N/M})$. 
    Then, for any integer $k \geq 1$, the overall error can be bounded by 
    \begin{equation}
        \min \left\{ \frac{C}{T} \sqrt{\frac{N}{M}} \left(\log\frac{N}{M}\right)^4, \frac{C_k(N,M)}{T^k}  \right\}. 
    \end{equation}
    Here $C$ is an absolute constant, and $C_k(N,M)$ is a constant that depend on $k,N,M$. 
\end{theorem}

~\cref{thm:Grover_adiabatic_BC} can be understood from two angles. 
On the one hand, for constant level of error $\epsilon$, we use the first error bound and the number of Trotter steps required scales $\widetilde{\mathcal{O}}(\sqrt{N/M})$, achieving the Grover scaling up to logarithmic factors. 
On the other hand, for fixed $N$ and $M$, we use the second error bound to conclude that the convergence in $T$ is super-polynomial and $T$ scales $\widetilde{\mathcal{O}}(1/\epsilon^{o(1)})$. 
Notice that the scheduling function~\cref{eqn:Grover_scheduling_def_BC} is independent of both $M$ and $N$.

Finally we compare our results with QAOA. 
Our results imply an approach to construct the angles in QAOA via adiabatic evolution for Grover search problem. 
Specifically, QAOA for Grover search takes the ansatz 
\begin{equation}\label{eqn:Grover_QAOA_ansatz}
    \ket{\psi} = \prod_{j=0}^{p-1} e^{-i\gamma_j H_1} e^{-i\beta_j H_0} \ket{\psi_0}
\end{equation}
where the parameters $(\beta_j,\gamma_j)$ are optimized to minimize $\braket{\psi}{H_1|\psi}$. 
According to our discrete adiabatic result, the first-order Trotter method with time step 1 outputs 
\begin{equation}
    \ket{\widetilde{\phi}} = \prod_{j=0}^{T-1} e^{-if(j/T)H_1} e^{-i(1-f(j/T))H_0} \ket{u}, 
\end{equation}
and $\ket{\widetilde{\phi}}$ is a good approximation of the ground state of $H_1$. 
This implies a set of near-optimal QAOA parameters 
\begin{equation}
    \beta_j = 1-f(j/T), \quad \gamma_j = f(j/T), 
\end{equation}
where $f(s)$ is the scheduling function in~\cref{eqn:Grover_scheduling_def} with $1<p<2$. 
When $M = 1$, the depth and time complexity of QAOA are both $p = T \sim \sqrt{N}$ which already matches the Grover lower bound. 
Therefore, our result shows that discrete adiabatic evolution can asymptotically match the performance of optimized QAOA for Grover search problem, and the speedup of QAOA is at most a constant factor.

\section{Conclusions}\label{sec:conclusion}

In this work, we establish improved error bounds and complexity analysis for the first-order exponential integrator and product formula when applied to simulating near-adiabatic quantum dynamics. 
The key step in our analysis is to view the numerical integrator as the set of slowly varying unitary walk operators, and to directly bound the distance between the actual output state and the ideal eigenstate. 
Compared to the standard error analysis which combines the continuous adiabatic theorem and the bounds for time discretization error in the general setting, our method provides a much more direct point of view and thus yields improvements. 
In the first-order exponential integrator or product formula of any order, for accurate simulation, it suffices to choose a uniform time step size independent of the error $\epsilon$ and the evolution time $T$. 
This greatly improves the previous estimate which scales as $\epsilon^{1/p}/T^{1/p}$ for a $p$-th order product formula. 
Remarkably, when the Hamiltonian is bounded by $1$, a first-order method unit time step size would suffice. 
Furthermore, under the boundary cancellation condition, even first-order methods with uniform time step size can achieve exponential convergence. 

The key technical tool in our analysis is the discrete adiabatic theorems. 
In particular, under the boundary cancellation condition, we use a high-order discrete adiabatic theorem to show exponential convergence for first-order methods. 
However, there are two open questions in the high-order discrete adiabatic theorem. 
First, as discussed in~\cref{app:DKS_missing_step}, although the conclusion of the high-order discrete adiabatic theorem is numerically correct, the proof presented in~\cite{DKS98} seems not valid due to a missing step. 
It is our future work to fix this gap and establish a rigorous proof for the high-order discrete adiabatic theorem. 
Another drawback of the current high-order discrete adiabatic theorem is that it only analyzes the scaling of errors in terms of the evolution time $T$. 
It would be interesting and practically useful to enhance the result with explicit dependence on the spectral gap as well, similar to what has been done for first-order discrete adiabatic theorem in~\cite{CostaAnYuvalEtAl2022}. 

To apply the discrete adiabatic theorems, we require the numerical integrator to satisfy the gap condition as well. 
When analyzing high-order product formula, we match the gap condition for the numerical integrator with that of the time-dependent Hamiltonian $H(s)$, as long as the time step size is smaller than a threshold which only depends on the Hamiltonian $H(s)$ itself. 
However, when the time step size is large, there is no clear relation between the gap of the numerical integrator and the gap of the Hamiltonian, as shown in~\cref{sec:gap} via toy models. 
Remarkably, it is possible that, even if the Hamiltonian $H(s)$ has small gap or becomes gapless, its trotterization with large time step may regain the gap condition, indicating that adiabatic evolution with suitable discretization may achieve exponential speedup over the standard AQC framework. 
It is very interesting to find if there is any application of practical interest that has such feature in its gap condition. 

We apply our analysis to the example of adiabatic Grover search problem. 
For this problem, we find that the QAOA approach does not yield better asymptotic scaling compared to the Trotterized adiabatic approach. 
Relation and comparison between AQC and QAOA for eigenstate preparation and optimization is a long-standing topic, and a thorough and comprehensive understanding may significantly advance the theoretical study on QAOA or even more general variational quantum algorithms. 
Our work here is the first step to bridge the connection, in the scenario where the Hamiltonian is bounded and well gapped. 
The next step would be to further explore AQC (and especially discrete AQC) and QAOA for gapped Hamiltonian but with large spectral norm, and then even for gapless Hamiltonian.

\section*{Acknowledgements}
The authors thank Lucas Brady, Jacob Bringewatt, Andrew Childs, Alexey Gorshkov, Michael Jarret, and Lin Lin for helpful discussions. 
DA acknowledges funding from Innovation Program for Quantum Science and Technology via Project 2024ZD0301900, the support by The Fundamental Research Funds for the Central Universities, Peking University, and the support by the Department of Defense through the Hartree Postdoctoral Fellowship at QuICS. 
DWB is supported by Australian Research Council Discovery Projects DP210101367 and DP220101602.

\bibliographystyle{apsrev4-1}	
\bibliography{Adiab}

\begin{thebibliography}{38}%
\makeatletter
\providecommand \@ifxundefined [1]{%
 \@ifx{#1\undefined}
}%
\providecommand \@ifnum [1]{%
 \ifnum #1\expandafter \@firstoftwo
 \else \expandafter \@secondoftwo
 \fi
}%
\providecommand \@ifx [1]{%
 \ifx #1\expandafter \@firstoftwo
 \else \expandafter \@secondoftwo
 \fi
}%
\providecommand \natexlab [1]{#1}%
\providecommand \enquote  [1]{``#1''}%
\providecommand \bibnamefont  [1]{#1}%
\providecommand \bibfnamefont [1]{#1}%
\providecommand \citenamefont [1]{#1}%
\providecommand \href@noop [0]{\@secondoftwo}%
\providecommand \href [0]{\begingroup \@sanitize@url \@href}%
\providecommand \@href[1]{\@@startlink{#1}\@@href}%
\providecommand \@@href[1]{\endgroup#1\@@endlink}%
\providecommand \@sanitize@url [0]{\catcode `\\12\catcode `\$12\catcode
  `\&12\catcode `\#12\catcode `\^12\catcode `\_12\catcode `\%12\relax}%
\providecommand \@@startlink[1]{}%
\providecommand \@@endlink[0]{}%
\providecommand \url  [0]{\begingroup\@sanitize@url \@url }%
\providecommand \@url [1]{\endgroup\@href {#1}{\urlprefix }}%
\providecommand \urlprefix  [0]{URL }%
\providecommand \Eprint [0]{\href }%
\providecommand \doibase [0]{http://dx.doi.org/}%
\providecommand \selectlanguage [0]{\@gobble}%
\providecommand \bibinfo  [0]{\@secondoftwo}%
\providecommand \bibfield  [0]{\@secondoftwo}%
\providecommand \translation [1]{[#1]}%
\providecommand \BibitemOpen [0]{}%
\providecommand \bibitemStop [0]{}%
\providecommand \bibitemNoStop [0]{.\EOS\space}%
\providecommand \EOS [0]{\spacefactor3000\relax}%
\providecommand \BibitemShut  [1]{\csname bibitem#1\endcsname}%
\let\auto@bib@innerbib\@empty
\bibitem [{\citenamefont {Albash}\ and\ \citenamefont
  {Lidar}(2018)}]{RevModPhys.90.015002}%
  \BibitemOpen
  \bibfield  {author} {\bibinfo {author} {\bibfnamefont {T.}~\bibnamefont
  {Albash}}\ and\ \bibinfo {author} {\bibfnamefont {D.~A.}\ \bibnamefont
  {Lidar}},\ }\href {\doibase 10.1103/RevModPhys.90.015002} {\bibfield
  {journal} {\bibinfo  {journal} {Rev. Mod. Phys.}\ }\textbf {\bibinfo {volume}
  {90}},\ \bibinfo {pages} {015002} (\bibinfo {year} {2018})}\BibitemShut
  {NoStop}%
\bibitem [{\citenamefont {Hochbruck}\ and\ \citenamefont
  {Ostermann}(2010)}]{HochbruckOstermann2010}%
  \BibitemOpen
  \bibfield  {author} {\bibinfo {author} {\bibfnamefont {M.}~\bibnamefont
  {Hochbruck}}\ and\ \bibinfo {author} {\bibfnamefont {A.}~\bibnamefont
  {Ostermann}},\ }\href {\doibase 10.1017/S0962492910000048} {\bibfield
  {journal} {\bibinfo  {journal} {Acta Numerica}\ }\textbf {\bibinfo {volume}
  {19}},\ \bibinfo {pages} {209–286} (\bibinfo {year} {2010})}\BibitemShut
  {NoStop}%
\bibitem [{\citenamefont {Schiffer}\ \emph {et~al.}(2022)\citenamefont
  {Schiffer}, \citenamefont {Tura},\ and\ \citenamefont
  {Cirac}}]{SchifferTuraCirac2022}%
  \BibitemOpen
  \bibfield  {author} {\bibinfo {author} {\bibfnamefont {B.~F.}\ \bibnamefont
  {Schiffer}}, \bibinfo {author} {\bibfnamefont {J.}~\bibnamefont {Tura}}, \
  and\ \bibinfo {author} {\bibfnamefont {J.~I.}\ \bibnamefont {Cirac}},\ }\href
  {\doibase 10.1103/PRXQuantum.3.020347} {\bibfield  {journal} {\bibinfo
  {journal} {PRX Quantum}\ }\textbf {\bibinfo {volume} {3}},\ \bibinfo {pages}
  {020347} (\bibinfo {year} {2022})}\BibitemShut {NoStop}%
\bibitem [{\citenamefont {Farhi}\ \emph {et~al.}(2000)\citenamefont {Farhi},
  \citenamefont {Goldstone}, \citenamefont {Gutmann},\ and\ \citenamefont
  {Sipser}}]{farhi2000quantum}%
  \BibitemOpen
  \bibfield  {author} {\bibinfo {author} {\bibfnamefont {E.}~\bibnamefont
  {Farhi}}, \bibinfo {author} {\bibfnamefont {J.}~\bibnamefont {Goldstone}},
  \bibinfo {author} {\bibfnamefont {S.}~\bibnamefont {Gutmann}}, \ and\
  \bibinfo {author} {\bibfnamefont {M.}~\bibnamefont {Sipser}},\ }\href
  {https://arxiv.org/abs/quant-ph/0001106} {\bibfield  {journal} {\bibinfo
  {journal} {arXiv:quant-ph/0001106}\ } (\bibinfo {year} {2000})}\BibitemShut
  {NoStop}%
\bibitem [{\citenamefont {van Dam}\ \emph {et~al.}(2001)\citenamefont {van
  Dam}, \citenamefont {Mosca},\ and\ \citenamefont
  {Vazirani}}]{vanDamMoscaVazirani2001}%
  \BibitemOpen
  \bibfield  {author} {\bibinfo {author} {\bibfnamefont {W.}~\bibnamefont {van
  Dam}}, \bibinfo {author} {\bibfnamefont {M.}~\bibnamefont {Mosca}}, \ and\
  \bibinfo {author} {\bibfnamefont {U.}~\bibnamefont {Vazirani}},\ }in\ \href
  {\doibase 10.1109/SFCS.2001.959902} {\emph {\bibinfo {booktitle} {Proceedings
  42nd IEEE Symposium on Foundations of Computer Science}}}\ (\bibinfo {year}
  {2001})\ pp.\ \bibinfo {pages} {279--287}\BibitemShut {NoStop}%
\bibitem [{\citenamefont {Childs}\ \emph {et~al.}(2021)\citenamefont {Childs},
  \citenamefont {Su}, \citenamefont {Tran}, \citenamefont {Wiebe},\ and\
  \citenamefont {Zhu}}]{ChildsSuTranEtAl2019}%
  \BibitemOpen
  \bibfield  {author} {\bibinfo {author} {\bibfnamefont {A.~M.}\ \bibnamefont
  {Childs}}, \bibinfo {author} {\bibfnamefont {Y.}~\bibnamefont {Su}}, \bibinfo
  {author} {\bibfnamefont {M.~C.}\ \bibnamefont {Tran}}, \bibinfo {author}
  {\bibfnamefont {N.}~\bibnamefont {Wiebe}}, \ and\ \bibinfo {author}
  {\bibfnamefont {S.}~\bibnamefont {Zhu}},\ }\href {\doibase
  10.1103/PhysRevX.11.011020} {\bibfield  {journal} {\bibinfo  {journal}
  {Physical Review X}\ }\textbf {\bibinfo {volume} {11}},\ \bibinfo {pages}
  {011020} (\bibinfo {year} {2021})}\BibitemShut {NoStop}%
\bibitem [{\citenamefont {An}\ \emph {et~al.}(2021)\citenamefont {An},
  \citenamefont {Fang},\ and\ \citenamefont {Lin}}]{AnFangLin2021}%
  \BibitemOpen
  \bibfield  {author} {\bibinfo {author} {\bibfnamefont {D.}~\bibnamefont
  {An}}, \bibinfo {author} {\bibfnamefont {D.}~\bibnamefont {Fang}}, \ and\
  \bibinfo {author} {\bibfnamefont {L.}~\bibnamefont {Lin}},\ }\href {\doibase
  10.22331/q-2021-05-26-459} {\bibfield  {journal} {\bibinfo  {journal}
  {Quantum}\ }\textbf {\bibinfo {volume} {5}},\ \bibinfo {pages} {459}
  (\bibinfo {year} {2021})}\BibitemShut {NoStop}%
\bibitem [{\citenamefont {Suzuki}(1993)}]{suzuki1993general}%
  \BibitemOpen
  \bibfield  {author} {\bibinfo {author} {\bibfnamefont {M.}~\bibnamefont
  {Suzuki}},\ }\href {\doibase 10.2183/pjab.69.161} {\bibfield  {journal}
  {\bibinfo  {journal} {Proceedings of the Japan Academy, Series B}\ }\textbf
  {\bibinfo {volume} {69}},\ \bibinfo {pages} {161–166} (\bibinfo {year}
  {1993})}\BibitemShut {NoStop}%
\bibitem [{\citenamefont {Kieferov{\'{a}}}\ \emph {et~al.}(2019)\citenamefont
  {Kieferov{\'{a}}}, \citenamefont {Scherer},\ and\ \citenamefont
  {Berry}}]{KieferovSchererBerry2019}%
  \BibitemOpen
  \bibfield  {author} {\bibinfo {author} {\bibfnamefont {M.}~\bibnamefont
  {Kieferov{\'{a}}}}, \bibinfo {author} {\bibfnamefont {A.}~\bibnamefont
  {Scherer}}, \ and\ \bibinfo {author} {\bibfnamefont {D.~W.}\ \bibnamefont
  {Berry}},\ }\href {\doibase 10.1103/physreva.99.042314} {\bibfield  {journal}
  {\bibinfo  {journal} {Physical Review A}\ }\textbf {\bibinfo {volume} {99}},\
  \bibinfo {pages} {042314} (\bibinfo {year} {2019})}\BibitemShut {NoStop}%
\bibitem [{\citenamefont {Low}\ and\ \citenamefont
  {Wiebe}(2018)}]{LowWiebe2019}%
  \BibitemOpen
  \bibfield  {author} {\bibinfo {author} {\bibfnamefont {G.~H.}\ \bibnamefont
  {Low}}\ and\ \bibinfo {author} {\bibfnamefont {N.}~\bibnamefont {Wiebe}},\
  }\href {https://doi.org/10.48550/arXiv.1805.00675} {\bibfield  {journal}
  {\bibinfo  {journal} {arXiv:1805.00675}\ } (\bibinfo {year}
  {2018})}\BibitemShut {NoStop}%
\bibitem [{\citenamefont {Su}\ \emph {et~al.}(2021)\citenamefont {Su},
  \citenamefont {Berry}, \citenamefont {Wiebe}, \citenamefont {Rubin},\ and\
  \citenamefont {Babbush}}]{SuBerryWiebeEtAl2021}%
  \BibitemOpen
  \bibfield  {author} {\bibinfo {author} {\bibfnamefont {Y.}~\bibnamefont
  {Su}}, \bibinfo {author} {\bibfnamefont {D.~W.}\ \bibnamefont {Berry}},
  \bibinfo {author} {\bibfnamefont {N.}~\bibnamefont {Wiebe}}, \bibinfo
  {author} {\bibfnamefont {N.}~\bibnamefont {Rubin}}, \ and\ \bibinfo {author}
  {\bibfnamefont {R.}~\bibnamefont {Babbush}},\ }\href {\doibase
  10.1103/prxquantum.2.040332} {\bibfield  {journal} {\bibinfo  {journal}
  {{PRX} Quantum}\ }\textbf {\bibinfo {volume} {2}},\ \bibinfo {pages} {040332}
  (\bibinfo {year} {2021})}\BibitemShut {NoStop}%
\bibitem [{\citenamefont {Rajput}\ \emph {et~al.}(2022)\citenamefont {Rajput},
  \citenamefont {Roggero},\ and\ \citenamefont
  {Wiebe}}]{RajputRoggeroWiebe2022}%
  \BibitemOpen
  \bibfield  {author} {\bibinfo {author} {\bibfnamefont {A.}~\bibnamefont
  {Rajput}}, \bibinfo {author} {\bibfnamefont {A.}~\bibnamefont {Roggero}}, \
  and\ \bibinfo {author} {\bibfnamefont {N.}~\bibnamefont {Wiebe}},\ }\href
  {\doibase 10.22331/q-2022-08-17-780} {\bibfield  {journal} {\bibinfo
  {journal} {Quantum}\ }\textbf {\bibinfo {volume} {6}},\ \bibinfo {pages}
  {780} (\bibinfo {year} {2022})}\BibitemShut {NoStop}%
\bibitem [{\citenamefont {Farhi}\ \emph {et~al.}(2014)\citenamefont {Farhi},
  \citenamefont {Goldstone},\ and\ \citenamefont
  {Gutmann}}]{FarhiGoldstoneGutmann2014}%
  \BibitemOpen
  \bibfield  {author} {\bibinfo {author} {\bibfnamefont {E.}~\bibnamefont
  {Farhi}}, \bibinfo {author} {\bibfnamefont {J.}~\bibnamefont {Goldstone}}, \
  and\ \bibinfo {author} {\bibfnamefont {S.}~\bibnamefont {Gutmann}},\ }\href
  {https://doi.org/10.48550/arXiv.1411.4028} {\bibfield  {journal} {\bibinfo
  {journal} {arXiv:1411.4028}\ } (\bibinfo {year} {2014})}\BibitemShut
  {NoStop}%
\bibitem [{\citenamefont {Yi}(2021)}]{Yi2021}%
  \BibitemOpen
  \bibfield  {author} {\bibinfo {author} {\bibfnamefont {C.}~\bibnamefont
  {Yi}},\ }\href {\doibase 10.1103/PhysRevA.104.052603} {\bibfield  {journal}
  {\bibinfo  {journal} {Physical Review A}\ }\textbf {\bibinfo {volume}
  {104}},\ \bibinfo {pages} {052603} (\bibinfo {year} {2021})}\BibitemShut
  {NoStop}%
\bibitem [{\citenamefont {Kovalsky}\ \emph {et~al.}(2023)\citenamefont
  {Kovalsky}, \citenamefont {Calderon-Vargas}, \citenamefont {Grace},
  \citenamefont {Magann}, \citenamefont {Larsen}, \citenamefont {Baczewski},\
  and\ \citenamefont {Sarovar}}]{kocia2022digital}%
  \BibitemOpen
  \bibfield  {author} {\bibinfo {author} {\bibfnamefont {L.~K.}\ \bibnamefont
  {Kovalsky}}, \bibinfo {author} {\bibfnamefont {F.~A.}\ \bibnamefont
  {Calderon-Vargas}}, \bibinfo {author} {\bibfnamefont {M.~D.}\ \bibnamefont
  {Grace}}, \bibinfo {author} {\bibfnamefont {A.~B.}\ \bibnamefont {Magann}},
  \bibinfo {author} {\bibfnamefont {J.~B.}\ \bibnamefont {Larsen}}, \bibinfo
  {author} {\bibfnamefont {A.~D.}\ \bibnamefont {Baczewski}}, \ and\ \bibinfo
  {author} {\bibfnamefont {M.}~\bibnamefont {Sarovar}},\ }\href {\doibase
  10.1103/PhysRevLett.131.060602} {\bibfield  {journal} {\bibinfo  {journal}
  {Physical Review Letters}\ }\textbf {\bibinfo {volume} {131}},\ \bibinfo
  {pages} {060602} (\bibinfo {year} {2023})}\BibitemShut {NoStop}%
\bibitem [{\citenamefont {Wiebe}\ \emph {et~al.}(2010)\citenamefont {Wiebe},
  \citenamefont {Berry}, \citenamefont {H{\o}yer},\ and\ \citenamefont
  {Sanders}}]{WiebeBerryHoyerEtAl2010}%
  \BibitemOpen
  \bibfield  {author} {\bibinfo {author} {\bibfnamefont {N.}~\bibnamefont
  {Wiebe}}, \bibinfo {author} {\bibfnamefont {D.}~\bibnamefont {Berry}},
  \bibinfo {author} {\bibfnamefont {P.}~\bibnamefont {H{\o}yer}}, \ and\
  \bibinfo {author} {\bibfnamefont {B.~C.}\ \bibnamefont {Sanders}},\ }\href
  {\doibase 10.1088/1751-8113/43/6/065203} {\bibfield  {journal} {\bibinfo
  {journal} {Journal of Physics A: Mathematical and Theoretical}\ }\textbf
  {\bibinfo {volume} {43}},\ \bibinfo {pages} {065203} (\bibinfo {year}
  {2010})}\BibitemShut {NoStop}%
\bibitem [{\citenamefont {Yoshida}(1990)}]{Yoshida1990}%
  \BibitemOpen
  \bibfield  {author} {\bibinfo {author} {\bibfnamefont {H.}~\bibnamefont
  {Yoshida}},\ }\href {\doibase https://doi.org/10.1016/0375-9601(90)90092-3}
  {\bibfield  {journal} {\bibinfo  {journal} {Physics Letters A}\ }\textbf
  {\bibinfo {volume} {150}},\ \bibinfo {pages} {262} (\bibinfo {year}
  {1990})}\BibitemShut {NoStop}%
\bibitem [{\citenamefont {Morales}\ \emph {et~al.}(2025)\citenamefont
  {Morales}, \citenamefont {Costa}, \citenamefont {Pantaleoni}, \citenamefont
  {Burgarth}, \citenamefont {Sanders},\ and\ \citenamefont
  {Berry}}]{MoralesCostaBurgarthEtAl2022}%
  \BibitemOpen
  \bibfield  {author} {\bibinfo {author} {\bibfnamefont {M.~E.~S.}\
  \bibnamefont {Morales}}, \bibinfo {author} {\bibfnamefont {P.~C.~S.}\
  \bibnamefont {Costa}}, \bibinfo {author} {\bibfnamefont {G.}~\bibnamefont
  {Pantaleoni}}, \bibinfo {author} {\bibfnamefont {D.~K.}\ \bibnamefont
  {Burgarth}}, \bibinfo {author} {\bibfnamefont {Y.~R.}\ \bibnamefont
  {Sanders}}, \ and\ \bibinfo {author} {\bibfnamefont {D.~W.}\ \bibnamefont
  {Berry}},\ }\href {\doibase 10.2478/qic-2025-0001} {\bibfield  {journal}
  {\bibinfo  {journal} {Quantum Information \& Computation}\ }\textbf {\bibinfo
  {volume} {25}},\ \bibinfo {pages} {1} (\bibinfo {year} {2025})}\BibitemShut
  {NoStop}%
\bibitem [{\citenamefont {Blanes}\ \emph {et~al.}(2024)\citenamefont {Blanes},
  \citenamefont {Casas},\ and\ \citenamefont {Murua}}]{BlanesCasasMurua2024}%
  \BibitemOpen
  \bibfield  {author} {\bibinfo {author} {\bibfnamefont {S.}~\bibnamefont
  {Blanes}}, \bibinfo {author} {\bibfnamefont {F.}~\bibnamefont {Casas}}, \
  and\ \bibinfo {author} {\bibfnamefont {A.}~\bibnamefont {Murua}},\ }\href
  {\doibase 10.1017/S0962492923000077} {\bibfield  {journal} {\bibinfo
  {journal} {Acta Numerica}\ }\textbf {\bibinfo {volume} {33}},\ \bibinfo
  {pages} {1–161} (\bibinfo {year} {2024})}\BibitemShut {NoStop}%
\bibitem [{\citenamefont {Jansen}\ \emph {et~al.}(2007)\citenamefont {Jansen},
  \citenamefont {Ruskai},\ and\ \citenamefont {Seiler}}]{jansen2007bounds}%
  \BibitemOpen
  \bibfield  {author} {\bibinfo {author} {\bibfnamefont {S.}~\bibnamefont
  {Jansen}}, \bibinfo {author} {\bibfnamefont {M.-B.}\ \bibnamefont {Ruskai}},
  \ and\ \bibinfo {author} {\bibfnamefont {R.}~\bibnamefont {Seiler}},\ }\href
  {\doibase https://doi.org/10.1063/1.2798382} {\bibfield  {journal} {\bibinfo
  {journal} {Journal of Mathematical Physics}\ }\textbf {\bibinfo {volume}
  {48}},\ \bibinfo {pages} {102111} (\bibinfo {year} {2007})}\BibitemShut
  {NoStop}%
\bibitem [{\citenamefont {Dranov}\ \emph {et~al.}(1998)\citenamefont {Dranov},
  \citenamefont {Kellendonk},\ and\ \citenamefont {Seiler}}]{DKS98}%
  \BibitemOpen
  \bibfield  {author} {\bibinfo {author} {\bibfnamefont {A.}~\bibnamefont
  {Dranov}}, \bibinfo {author} {\bibfnamefont {J.}~\bibnamefont {Kellendonk}},
  \ and\ \bibinfo {author} {\bibfnamefont {R.}~\bibnamefont {Seiler}},\ }\href
  {\doibase 10.1063/1.532382} {\bibfield  {journal} {\bibinfo  {journal}
  {Journal of Mathematical Physics}\ }\textbf {\bibinfo {volume} {39}},\
  \bibinfo {pages} {1340} (\bibinfo {year} {1998})}\BibitemShut {NoStop}%
\bibitem [{\citenamefont {Costa}\ \emph {et~al.}(2022)\citenamefont {Costa},
  \citenamefont {An}, \citenamefont {Sanders}, \citenamefont {Su},
  \citenamefont {Babbush},\ and\ \citenamefont {Berry}}]{CostaAnYuvalEtAl2022}%
  \BibitemOpen
  \bibfield  {author} {\bibinfo {author} {\bibfnamefont {P.~C.}\ \bibnamefont
  {Costa}}, \bibinfo {author} {\bibfnamefont {D.}~\bibnamefont {An}}, \bibinfo
  {author} {\bibfnamefont {Y.~R.}\ \bibnamefont {Sanders}}, \bibinfo {author}
  {\bibfnamefont {Y.}~\bibnamefont {Su}}, \bibinfo {author} {\bibfnamefont
  {R.}~\bibnamefont {Babbush}}, \ and\ \bibinfo {author} {\bibfnamefont
  {D.~W.}\ \bibnamefont {Berry}},\ }\href {\doibase
  10.1103/PRXQuantum.3.040303} {\bibfield  {journal} {\bibinfo  {journal} {PRX
  Quantum}\ }\textbf {\bibinfo {volume} {3}},\ \bibinfo {pages} {040303}
  (\bibinfo {year} {2022})}\BibitemShut {NoStop}%
\bibitem [{\citenamefont {Nenciu}(1993)}]{Nenciu1993}%
  \BibitemOpen
  \bibfield  {author} {\bibinfo {author} {\bibfnamefont {G.}~\bibnamefont
  {Nenciu}},\ }\href {\doibase 10.1007/bf02096616} {\bibfield  {journal}
  {\bibinfo  {journal} {Communications in Mathematical Physics}\ }\textbf
  {\bibinfo {volume} {152}},\ \bibinfo {pages} {479–496} (\bibinfo {year}
  {1993})}\BibitemShut {NoStop}%
\bibitem [{\citenamefont {Ge}\ \emph {et~al.}(2016)\citenamefont {Ge},
  \citenamefont {Moln\'ar},\ and\ \citenamefont {Cirac}}]{GeMolnarCirac2016}%
  \BibitemOpen
  \bibfield  {author} {\bibinfo {author} {\bibfnamefont {Y.}~\bibnamefont
  {Ge}}, \bibinfo {author} {\bibfnamefont {A.}~\bibnamefont {Moln\'ar}}, \ and\
  \bibinfo {author} {\bibfnamefont {J.~I.}\ \bibnamefont {Cirac}},\ }\href
  {\doibase 10.1103/PhysRevLett.116.080503} {\bibfield  {journal} {\bibinfo
  {journal} {Physical Review Letters}\ }\textbf {\bibinfo {volume} {116}},\
  \bibinfo {pages} {080503} (\bibinfo {year} {2016})}\BibitemShut {NoStop}%
\bibitem [{\citenamefont {Avron}\ and\ \citenamefont
  {Elgart}(1999)}]{AvronElgart1999}%
  \BibitemOpen
  \bibfield  {author} {\bibinfo {author} {\bibfnamefont {J.~E.}\ \bibnamefont
  {Avron}}\ and\ \bibinfo {author} {\bibfnamefont {A.}~\bibnamefont {Elgart}},\
  }\href {\doibase 10.1007/s002200050620} {\bibfield  {journal} {\bibinfo
  {journal} {Communications in Mathematical Physics}\ }\textbf {\bibinfo
  {volume} {203}},\ \bibinfo {pages} {445} (\bibinfo {year}
  {1999})}\BibitemShut {NoStop}%
\bibitem [{\citenamefont {Roland}\ and\ \citenamefont
  {Cerf}(2002)}]{RolandCerf2002}%
  \BibitemOpen
  \bibfield  {author} {\bibinfo {author} {\bibfnamefont {J.}~\bibnamefont
  {Roland}}\ and\ \bibinfo {author} {\bibfnamefont {N.~J.}\ \bibnamefont
  {Cerf}},\ }\href {\doibase 10.1103/PhysRevA.65.042308} {\bibfield  {journal}
  {\bibinfo  {journal} {Physical Review A}\ }\textbf {\bibinfo {volume} {65}},\
  \bibinfo {pages} {042308} (\bibinfo {year} {2002})}\BibitemShut {NoStop}%
\bibitem [{\citenamefont {An}\ and\ \citenamefont {Lin}(2022)}]{an2019quantum}%
  \BibitemOpen
  \bibfield  {author} {\bibinfo {author} {\bibfnamefont {D.}~\bibnamefont
  {An}}\ and\ \bibinfo {author} {\bibfnamefont {L.}~\bibnamefont {Lin}},\
  }\href {\doibase 10.1145/3498331} {\bibfield  {journal} {\bibinfo  {journal}
  {ACM Transactions on Quantum Computing}\ }\textbf {\bibinfo {volume} {3}},\
  \bibinfo {pages} {5} (\bibinfo {year} {2022})}\BibitemShut {NoStop}%
\bibitem [{\citenamefont {Dalzell}\ \emph {et~al.}(2017)\citenamefont
  {Dalzell}, \citenamefont {Yoder},\ and\ \citenamefont
  {Chuang}}]{DalzellYoderChuang2017}%
  \BibitemOpen
  \bibfield  {author} {\bibinfo {author} {\bibfnamefont {A.~M.}\ \bibnamefont
  {Dalzell}}, \bibinfo {author} {\bibfnamefont {T.~J.}\ \bibnamefont {Yoder}},
  \ and\ \bibinfo {author} {\bibfnamefont {I.~L.}\ \bibnamefont {Chuang}},\
  }\href {\doibase 10.1103/PhysRevA.95.012311} {\bibfield  {journal} {\bibinfo
  {journal} {Physical Review A}\ }\textbf {\bibinfo {volume} {95}},\ \bibinfo
  {pages} {012311} (\bibinfo {year} {2017})}\BibitemShut {NoStop}%
\bibitem [{\citenamefont {Yoder}\ \emph {et~al.}(2014)\citenamefont {Yoder},
  \citenamefont {Low},\ and\ \citenamefont {Chuang}}]{YoderLowChuang2014}%
  \BibitemOpen
  \bibfield  {author} {\bibinfo {author} {\bibfnamefont {T.~J.}\ \bibnamefont
  {Yoder}}, \bibinfo {author} {\bibfnamefont {G.~H.}\ \bibnamefont {Low}}, \
  and\ \bibinfo {author} {\bibfnamefont {I.~L.}\ \bibnamefont {Chuang}},\
  }\href {\doibase 10.1103/PhysRevLett.113.210501} {\bibfield  {journal}
  {\bibinfo  {journal} {Physical Review Letters}\ }\textbf {\bibinfo {volume}
  {113}},\ \bibinfo {pages} {210501} (\bibinfo {year} {2014})}\BibitemShut
  {NoStop}%
\bibitem [{\citenamefont {Grover}(1996)}]{Grover1996}%
  \BibitemOpen
  \bibfield  {author} {\bibinfo {author} {\bibfnamefont {L.~K.}\ \bibnamefont
  {Grover}},\ }in\ \href {\doibase 10.1145/237814.237866} {\emph {\bibinfo
  {booktitle} {Proceedings of the Twenty-Eighth Annual ACM Symposium on Theory
  of Computing}}},\ \bibinfo {series and number} {STOC '96}\ (\bibinfo
  {publisher} {Association for Computing Machinery},\ \bibinfo {address} {New
  York, NY, USA},\ \bibinfo {year} {1996})\ pp.\ \bibinfo {pages}
  {212--219}\BibitemShut {NoStop}%
\bibitem [{\citenamefont {Jiang}\ \emph {et~al.}(2017)\citenamefont {Jiang},
  \citenamefont {Rieffel},\ and\ \citenamefont {Wang}}]{JiangRieffelWang2017}%
  \BibitemOpen
  \bibfield  {author} {\bibinfo {author} {\bibfnamefont {Z.}~\bibnamefont
  {Jiang}}, \bibinfo {author} {\bibfnamefont {E.~G.}\ \bibnamefont {Rieffel}},
  \ and\ \bibinfo {author} {\bibfnamefont {Z.}~\bibnamefont {Wang}},\ }\href
  {\doibase 10.1103/PhysRevA.95.062317} {\bibfield  {journal} {\bibinfo
  {journal} {Physical Review A}\ }\textbf {\bibinfo {volume} {95}},\ \bibinfo
  {pages} {062317} (\bibinfo {year} {2017})}\BibitemShut {NoStop}%
\bibitem [{\citenamefont {Gily{\'e}n}\ \emph
  {et~al.}(2019{\natexlab{a}})\citenamefont {Gily{\'e}n}, \citenamefont {Su},
  \citenamefont {Low},\ and\ \citenamefont {Wiebe}}]{GilyenSuLowEtAl2019}%
  \BibitemOpen
  \bibfield  {author} {\bibinfo {author} {\bibfnamefont {A.}~\bibnamefont
  {Gily{\'e}n}}, \bibinfo {author} {\bibfnamefont {Y.}~\bibnamefont {Su}},
  \bibinfo {author} {\bibfnamefont {G.~H.}\ \bibnamefont {Low}}, \ and\
  \bibinfo {author} {\bibfnamefont {N.}~\bibnamefont {Wiebe}},\ }in\ \href
  {\doibase 10.1145/3313276.3316366} {\emph {\bibinfo {booktitle} {Proceedings
  of the 51st Annual ACM SIGACT Symposium on Theory of Computing}}}\ (\bibinfo
  {year} {2019})\ pp.\ \bibinfo {pages} {193--204}\BibitemShut {NoStop}%
\bibitem [{\citenamefont {Low}\ and\ \citenamefont
  {Chuang}(2019)}]{Low2019hamiltonian}%
  \BibitemOpen
  \bibfield  {author} {\bibinfo {author} {\bibfnamefont {G.~H.}\ \bibnamefont
  {Low}}\ and\ \bibinfo {author} {\bibfnamefont {I.~L.}\ \bibnamefont
  {Chuang}},\ }\href {\doibase 10.22331/q-2019-07-12-163} {\bibfield  {journal}
  {\bibinfo  {journal} {{Quantum}}\ }\textbf {\bibinfo {volume} {3}},\ \bibinfo
  {pages} {163} (\bibinfo {year} {2019})}\BibitemShut {NoStop}%
\bibitem [{\citenamefont {Gily{\'e}n}\ \emph
  {et~al.}(2019{\natexlab{b}})\citenamefont {Gily{\'e}n}, \citenamefont {Su},
  \citenamefont {Low},\ and\ \citenamefont {Wiebe}}]{GilyenSuLowEtAl2019arXiv}%
  \BibitemOpen
  \bibfield  {author} {\bibinfo {author} {\bibfnamefont {A.}~\bibnamefont
  {Gily{\'e}n}}, \bibinfo {author} {\bibfnamefont {Y.}~\bibnamefont {Su}},
  \bibinfo {author} {\bibfnamefont {G.~H.}\ \bibnamefont {Low}}, \ and\
  \bibinfo {author} {\bibfnamefont {N.}~\bibnamefont {Wiebe}},\ }\href@noop {}
  {\enquote {\bibinfo {title} {Quantum singular value transformation and
  beyond: exponential improvements for quantum matrix arithmetics},}\ }
  (\bibinfo {year} {2019}{\natexlab{b}}),\ \bibinfo {note}
  {\href{https://arxiv.org/abs/1806.01838}{arXiv:1806.01838}}\BibitemShut
  {NoStop}%
\bibitem [{\citenamefont {Bhatia}\ and\ \citenamefont
  {Davis}(1984)}]{BhatiaDavis1984}%
  \BibitemOpen
  \bibfield  {author} {\bibinfo {author} {\bibfnamefont {R.}~\bibnamefont
  {Bhatia}}\ and\ \bibinfo {author} {\bibfnamefont {C.}~\bibnamefont {Davis}},\
  }\href {\doibase 10.1080/03081088408817578} {\bibfield  {journal} {\bibinfo
  {journal} {Linear and Multilinear Algebra}\ }\textbf {\bibinfo {volume}
  {15}},\ \bibinfo {pages} {71} (\bibinfo {year} {1984})}\BibitemShut {NoStop}%
\bibitem [{\citenamefont {Elsner}\ and\ \citenamefont
  {He}(1993)}]{ElsnerHe1993}%
  \BibitemOpen
  \bibfield  {author} {\bibinfo {author} {\bibfnamefont {L.}~\bibnamefont
  {Elsner}}\ and\ \bibinfo {author} {\bibfnamefont {C.}~\bibnamefont {He}},\
  }\href {\doibase https://doi.org/10.1016/0024-3795(93)90469-5} {\bibfield
  {journal} {\bibinfo  {journal} {Linear Algebra and its Applications}\
  }\textbf {\bibinfo {volume} {188-189}},\ \bibinfo {pages} {207} (\bibinfo
  {year} {1993})}\BibitemShut {NoStop}%
\bibitem [{\citenamefont {An}(2021)}]{An2021thesis}%
  \BibitemOpen
  \bibfield  {author} {\bibinfo {author} {\bibfnamefont {D.}~\bibnamefont
  {An}},\ }\href {https://escholarship.org/uc/item/8h70p30t} {\bibfield
  {journal} {\bibinfo  {journal} {UC Berkeley Electronic Theses and
  Dissertations}\ } (\bibinfo {year} {2021})}\BibitemShut {NoStop}%
\bibitem [{\citenamefont {Milne-Thomson}(1933)}]{milne1933calculus}%
  \BibitemOpen
  \bibfield  {author} {\bibinfo {author} {\bibfnamefont {L.}~\bibnamefont
  {Milne-Thomson}},\ }\href@noop {} {\emph {\bibinfo {title} {The Calculus of
  Finite Differences}}}\ (\bibinfo  {publisher} {Macmillan and Co},\ \bibinfo
  {year} {1933})\BibitemShut {NoStop}%
\end{thebibliography}%

\clearpage

\appendix

\section{Continuous adiabatic theorems}\label{app:continuous_adiabatic}

Here we briefly review the previous results about continuous adiabatic theorems, which aim to bound the distance between the dynamics in~\cref{eqn:AQC_dynamics} and the ideal eigenspace. 

For a general interpolation function $f(s)$, the continuous adiabatic error can be bounded linearly in the inverse evolution time. 
Here we present a result from Ref.~\cite{jansen2007bounds}. 
Let $U(t)$ denote the exact evolution operator of~\cref{eqn:AQC_dynamics}, and $P(s)$ denote the spectral projection onto the eigenspace of $H(s)$ of interest. 
Suppose that $\Delta(s)$ is the spectral gap of the Hamiltonian $H(s)$. 

\begin{lemma}[Continuous adiabatic theorem]\label{lem:AT_JRS}
    Suppose that $H(s)$ is second-order continuously differentiable. Let $\Delta(s)$ denote the spectral gap of $H(s)$. Then for any $t\in[0,T]$ and $s = t/T$, we have 
    \begin{equation}
        \left\|\ket{\psi(t)} - P(s)\ket{\psi(t)}\right\| \leq 
        \frac{C}{T} \left(\frac{ \|H'(0)\|}{\Delta(0)^2} + \frac{ \|H'(s)\|}{\Delta(s)^2} + \int_0^s \frac{\|H''(\tau)\|}{\Delta(\tau)^2} d\tau + \int_0^s \frac{\|H'(\tau)\|^2}{\Delta(\tau)^3} d\tau \right) 
    \end{equation}
    for a constant $C > 0$ which only depends on the number of the eigenstates contained in $P(s)$. 
\end{lemma}

When the Hamiltonian $H(s)$ satisfies the boundary cancellation condition, namely that the derivatives of $H(s)$ of any order vanish at the boundary, convergence of the final adiabatic error can be improved to be exponential in time. 
Such a high-order adiabatic theorem was rigorously proved in~\cite{Nenciu1993} without explicit gap dependence and in~\cite{GeMolnarCirac2016} with cubic inverse gap dependence. 
Based on their approaches, a recent work~\cite{An2021thesis} improved the gap dependence to be quadratic. 
Here we present the main result in~\cite{An2021thesis}. 

\begin{definition}[Gevrey class]\label{def:gevrey}
    A function $g(s)$ defined on $[0,1]$ is in the Gevrey class $G^{\alpha}$ for $\alpha > 0$ if there exist constants $C,D>0$ such that for all $k \geq 0$, 
    \begin{equation}
        \max_{s\in[0,1]} \|g^{(k)}(s)\| \leq C D^k \frac{(k!)^{1+\alpha}}{(k+1)^2}. 
    \end{equation}
    Here $\|\cdot \|$ represents the absolute value if $g$ is scalar-valued and the spectral norm if $g$ is operator-valued.
\end{definition}

\begin{lemma}[Continuous adiabatic theorem with boundary cancellation{~\cite[Theorem 7]{An2021thesis}}]\label{lem:AT_exp}
     Suppose that $H(s)$ is in the Gevrey class $G^{\alpha}$ for $\alpha > 0$ and assume that $H^{(k)}(0) = H^{(k)}(1) = 0$ for all $k \geq 1$. 
    Let $\Delta(s)$ denote a lower bound of the spectral gap of $H(s)$ and let $\Delta_* = \min_{s\in[0,1]}\Delta(s)$. Then we have 
    \begin{equation}
        \left\| \ket{\psi(T)} - P(1)\ket{\psi(T)} \right\| \leq 
        \frac{c}{\Delta_*} \exp\left(-\left(c'T\Delta_*^2\right)^{\frac{1}{1+\alpha}}\right)
    \end{equation}
    where $c$ and $c'$ are positive constants only depending on $C$, $D$ (specified according to~\cref{def:gevrey}) and $\alpha$. 
\end{lemma}

Both~\cref{lem:AT_JRS} and~\cref{lem:AT_exp} bound the leakage of the dynamics out of the desired eigenspace. 
We remark that, up to a constant factor $2$, it can also serve as the upper bound of the distance between the dynamics and an ideal eigenstate $P(s)\ket{\psi(t)}/\|P(s)\ket{\psi(t)}\|$ using a linear algebra result that 
\begin{equation}
    \left\|\ket{\psi(t)} - \frac{P(s)\ket{\psi(t)}}{\|P(s)\ket{\psi(t)}\|}\right\| \leq 2\|\ket{\psi(t)}-P(s)\ket{\psi(t)}\|. 
\end{equation}

\section{Multistep gap condition}\label{app:multistep_gap}

The discrete adiabatic theorem in~\cref{lem:DAE_linear} (\cref{sec:DAT_linear}) needs the multistep gap condition, i.e., the gap is assumed between successive steps, which is a further subtlety. 
Here we show that, if the number of steps $T_d$ is large enough and each $W(s)$ has a uniformly bounded spectral gap, then the multistep gap condition can be guaranteed. 

\begin{lemma}\label{lem:multistep_gap}
    Let $W(s)$ be the walk operator, $T_d$ be the number of discrete walk steps, and $c_1(s)$ be the function such that $\|DW(s)\| \leq c_1(s)/T_d$. 
    Suppose that the gap between two groups $\sigma_P(s)$ and $\sigma_Q(s)$ of the eigenvalues of $W(s)$ is bounded from below by $\Delta_* > 0$. 
    Then, as long as $T_d \geq (2\pi/\Delta_*) \sup_{\widetilde{s} \in \{s,s+1/T_d\}\cap [0,1] }c_1(\widetilde{s})$, the multistep gap between the arcs $\sigma_P^{(2)}(s)$ and $\sigma_Q^{(2)}(s)$ is bounded from below by $\Delta_*/2$.
\end{lemma}
\begin{proof}
    Let $\lambda(s)$ and $\mu(s)$ denote the two eigenvalues of $W(s)$ that determine the gap, i.e., at the corresponding boundary of the arcs. 
    For any $s',s'' \in \{s,s+1/T_d,s+2/T_d\}\cap [0,1]$, according to~\cite{BhatiaDavis1984} and the definition of $c_1(s)$, we have 
    \begin{equation}
        |\lambda(s')-\lambda(s'')| \leq \|W(s')-W(s'')\| \leq \frac{2}{T_d} \sup_{\widetilde{s} \in \{s,s+1/T_d\}\cap [0,1] }c_1(\widetilde{s}). 
    \end{equation}
    Let $\mathbb{S}^1$ denote the unit circle, and define $|z_1-z_2|_{\mathbb{S}^1}$ to be the angular distance between $z_1,z_2\in \mathbb{S}^1$. 
    Using $|z_1-z_2|_{\mathbb{S}^1} \leq \frac{\pi}{2} |z_1-z_2|$, we have 
    \begin{equation}
        |\lambda(s')-\lambda(s'')|_{\mathbb{S}^1} \leq \frac{\pi}{T_d} \sup_{\widetilde{s} \in \{s,s+1/T_d\}\cap [0,1] }c_1(\widetilde{s}). 
    \end{equation}
    We have required that $T_d \geq \frac{2\pi}{\Delta_*} \sup_{\widetilde{s} \in \{s,s+1/T_d\}\cap [0,1] }c_1(\widetilde{s})$, which gives 
    \begin{equation}
        |\lambda(s')-\lambda(s'')|_{\mathbb{S}^1} \leq \frac{\Delta_*}{2}. 
    \end{equation}
    Therefore, by the triangle inequality, we have 
    \begin{equation}
        |\lambda(s') - \mu(s'')|_{\mathbb{S}^1} \geq |\lambda(s'') - \mu(s'')|_{\mathbb{S}^1} - |\lambda(s') - \lambda(s'')|_{\mathbb{S}^1} \geq \frac{\Delta_*}{2}, 
    \end{equation}
    which means that the multistep gap is at least $\Delta_*/2$. 
\end{proof}

\section{Bounds for finite differences}\label{app:bounds_finite_diff}

To apply the discrete adiabatic theorem stated in~\cref{lem:DAE_linear}, we need to estimate the scalings of $c_1(s)$ and $c_2(s)$. 
Here we show how to bound $c_1(s)$ and $c_2(s)$ by the derivatives of the walk operator $W(s)$, and derive explicit estimates for the first-order exponential integrator and the product formula. 

We first show that for each $k$, $c_k(s)$ is closely related to the $k$-th order derivative of $W(s)$. 

\begin{lemma}\label{lem:finite_difference_bound_general}
    Let $W(s)$ be a set of parameterized unitary walk operators for $s \in [0,1]$, and $D^{(k)}$ represents the $k$-th order finite difference of $W(s)$ with step size $1/T_d$ for an integer $T_d$. 
    For any positive integer $k$, if $W(s)$ is $k$-th order continuously differentiable, then we have 
    \begin{equation}
        \|D^{(k)}W(s)\| \leq \frac{c_k(s)}{T_d^k} 
    \end{equation}
    where 
    \begin{equation}\label{eq:ckeq}
        c_k(s) = \max_{s' \in [s, s+k/T_d]} \|W^{(k)}(s)\|. 
    \end{equation}
\end{lemma}

\begin{proof}
    For $k = 1$, using the fundamental theorem of calculus, we have 
    \begin{equation}
        DW(s) = W(s+1/T_d) - W(s) = \int_s^{s+1/T_d} W'(s') d s'.  
    \end{equation}
More generally,
    \begin{equation}
        D^{(k)}W(s) = \int_0^{1/T_d} ds_1\int_0^{1/T_d} ds_2 \cdots \int_0^{1/T_d}ds_k \, W^{(k)}(s+s_1+s_2+\cdots+s_k) \, .  
    \end{equation}
This is the matrix form of a standard result for scalar finite differences \cite[p.\ 10]{milne1933calculus}, and is easily found by induction with
\begin{align}
    D^{(k)}W(s) &= D^{(k-1)}W(s+1/T_d) - D^{(k-1)}W(s) \nn
    &=  \int_0^{1/T_d} ds_1\int_0^{1/T_d} ds_2 \cdots \int_0^{1/T_d}ds_{k-1} \, W^{(k-1)}(s+1/T_d+s_1+s_2+\cdots+s_{k-1})\nn
 & \quad   - \int_0^{1/T_d} ds_1\int_0^{1/T_d} ds_2 \cdots \int_0^{1/T_d}ds_{k-1} \, W^{(k-1)}(s+s_1+s_2+\cdots+s_{k-1}) \nn
    &= \int_0^{1/T_d} ds_1\int_0^{1/T_d} ds_2 \cdots \int_0^{1/T_d}ds_{k-1} \, DW^{(k-1)}(s+s_1+s_2+\cdots+s_{k-1}) \nn
    &= \int_0^{1/T_d} ds_1\int_0^{1/T_d} ds_2 \cdots \int_0^{1/T_d}ds_k \, W^{(k)}(s+s_1+s_2+\cdots+s_k) \, .
\end{align}
This expression gives, by the triangle inequality,
    \begin{align}
        \|D^{(k)}W(s)\| &\leq \int_0^{1/T_d} ds_1\int_0^{1/T_d} ds_2 \cdots \int_0^{1/T_d}ds_k \, \|W^{(k)}(s+s_1+s_2+\cdots+s_k)\| \nn
        &\leq \frac 1{T_d^k} \max_{s'\in[s,s+k/T_d]}\|W^{(k)}(s')\| \, .
    \end{align}
Therefore $\|D^{(k)}W(s)\|$ is upper bounded with $c_k$ as in Eq.~\eqref{eq:ckeq}. 
\end{proof}

Now we give explicit choices of $c_1(s)$ and $c_2(s)$ for the first-order exponential integrator and the product formula. 

\begin{lemma}\label{lem:c1_c2_exp}
    Let $W(s) = U_{\exp}(sT+h,sT)$ be the first-order exponential integrator defined in~\cref{eqn:exponential_integrator_1st}, where $h = T/T_d$ is the time step size. 
    Suppose that both $|f'(s)|$ and $|f''(s)|$ are uniformly bounded over $[0,1]$. 
    Then we can choose 
    \begin{equation}
        c_1(s) = \mathcal{O}(h(\|H_0\|+\|H_1\|)), \quad c_2(s) = \mathcal{O}(h(\|H_0\|+\|H_1\|) + h^2(\|H_0\|+\|H_1\|)^2). 
    \end{equation}
\end{lemma}
\begin{proof}
    First, using the formula that for any operator $X(s)$, 
    \begin{equation}
        \frac{d}{ds} e^{X(s)} = \int_0^1 e^{\alpha X(s)} \frac{d X(s)}{ds} e^{(1-\alpha)X(s)} d\alpha, 
    \end{equation}
    we can compute that 
    \begin{align}
        W'(s) &= -i h f'(s)\int_0^1 e^{-i \alpha h H(s)} (H_1-H_0) e^{-i  (1-\alpha) h H(s)} d\alpha, 
    \end{align}
    and, according to the chain rule, 
    \begin{align}
        W''(s) &= -ihf''(s)\int_0^1 e^{-i \alpha h H(s)} (H_1-H_0) e^{-i(1-\alpha)h H(s)} d\alpha \nonumber \\
        & \quad - h^2 f'(s)^2 \int_0^1\int_0^1 \alpha e^{-i \alpha\beta h H(s)}(H_1-H_0)e^{-i \alpha(1-\beta) h H(s)} (H_1-H_0) e^{-i(1-\alpha) h H(s)} d\beta \, d\alpha \nonumber \\
        & \quad - h^2 f'(s)^2\int_0^1\int_0^1 (1-\alpha) e^{-i \alpha h H(s)} (H_1-H_0) e^{-i(1-\alpha)\beta h H(s)} (H_1-H_0) e^{-i(1-\alpha)(1-\beta) h H(s)} d\beta \, d\alpha. 
    \end{align} 
    Therefore, using the bound $\|H_1-H_0\| \leq (\|H_0\|+\|H_1\|)$, we have 
    \begin{align}\label{eqn:exp_ite_W1}
        \|W'(s)\| \leq  h |f'(s)|\int_0^1 \|H_1-H_0\| d\alpha \leq \mathcal{O}(h(\|H_0\|+\|H_1\|))
    \end{align}
    and 
    \begin{align}
        & \quad \|W''(s)\|\nonumber \\
        &\leq h|f''(s)|\int_0^1 \|H_1-H_0\|  d\alpha + h^2 |f'(s)|^2 \int_0^1\int_0^1 \alpha\|H_1-H_0\|^2 d\beta \, d\alpha + h^2 |f'(s)|^2\int_0^1\int_0^1 (1-\alpha) \|H_1-H_0\|^2  d\beta \, d\alpha \nonumber \\
        & \leq \mathcal{O}(h(\|H_0\|+\|H_1\|) + h^2(\|H_0\|+\|H_1\|)^2). \label{eqn:exp_ite_W2}
    \end{align} 
    Then we can choose $c_1(s)$ and $c_2(s)$ as stated according to~\cref{lem:finite_difference_bound_general} and~\cref{eqn:exp_ite_W1,eqn:exp_ite_W2}. 
\end{proof}

\begin{lemma}\label{lem:c1_c2_product_formula}
    Let 
    \begin{equation}
        W(s) = U_{\text{pf},p}(sT+h,sT) = \prod_{k=0}^{K_p} \exp\left( -i\beta_{p,k} h f((sT+\delta_{p,k}h)/T) H_1\right)  \exp\left( -i\alpha_{p,k}h (1-f((sT+\gamma_{p,k}h)/T)) H_0 \right)
    \end{equation}
   be the product formula defined in~\cref{eqn:high_order_trotter_general}, where $h = T/T_d$ is the time step size. 
    Suppose that both $|f'(s)|$ and $|f''(s)|$ are uniformly bounded over $[0,1]$. 
    Then we can choose 
    \begin{equation}
        c_1(s) = \mathcal{O}(h(\|H_0\|+\|H_1\|)), \quad c_2(s) = \mathcal{O}(h(\|H_0\|+\|H_1\|) + h^2(\|H_0\|+\|H_1\|)^2). 
    \end{equation}
\end{lemma}
\begin{proof}
    Let 
    \begin{equation}
        V_{k,0}(s) = \exp\left( -i\alpha_{p,k}h (1-f((sT+\gamma_{p,k}h)/T)) H_0 \right), \quad V_{k,1}(s) = \exp\left( -i\beta_{p,k} h f((sT+\delta_{p,k}h)/T) H_1\right), 
    \end{equation}
    and 
    \begin{equation}
        V_k(s) = V_{k,1}(s) V_{k,0}(s) \, ,
    \end{equation}
    so that $W(s) = \prod_{k=0}^{K_p} V_k(s)$. 
    By the product rule, we have 
    \begin{equation}
        W'(s) = \sum_{k=0}^{K_p} \left( \prod_{k'=k+1}^{K_p} V_{k'}(s) \right) V_k'(s)  \left( \prod_{k'=0}^{k-1} V_{k'}(s)  \right) = \sum_{k=0}^{K_p} \left( \prod_{k'=k+1}^{K_p} V_{k'}(s) \right) \left( V_{k,1}'(s) V_{k,0}(s) + V_{k,1}(s) V_{k,0}'(s) \right) \left( \prod_{k'=0}^{k-1} V_{k'}(s)  \right), 
    \end{equation}
    and 
    \begin{equation}
        \|W'(s)\| = \mathcal{O}\left( \sum_{k=0}^{K_p} \left(\|V_{k,1}'(s)\| +  \| V_{k,0}'(s) \| \right) \right). 
    \end{equation}
    Because the exponent of each operator $V_{k,j}$, $j=0,1$, involves a single operator, they commute with their derivatives, so we can use
\begin{equation}
\frac{d}{ds}e^{X(s)}= e^{X(s)}X'(s),    
\end{equation}
so
\begin{align}
\frac{d}{ds}V_{k,0}(s)&= V_{k,0}(s)\left(i\alpha_{p,k}h f'(\varphi_k(s))H_0\right),  \nonumber\\
\frac{d}{ds}V_{k,1}(s)&= V_{k,1}(s)\left(-i\beta_{p,k}h f'(\phi_k(s))H_1\right), 
\end{align}
where $\varphi_k(s)=(sT+\gamma_{p,k}h)/T$ and $\phi_k(s)=(sT+\delta_{p,k}h)/T$. 
So we have $\|V_{k,0}'(s)\| = \mathcal{O}(h\|H_0\|)$ and $\|V_{k,1}'(s)\| = \mathcal{O}(h\|H_1\|)$, and 
    \begin{equation}\label{eqn:product_formula_ite_W1}
        \|W'(s)\| = \mathcal{O}\left( h(\|H_0\| + \|H_1\| ) \right). 
    \end{equation}
    Then, according to~\cref{lem:finite_difference_bound_general}, we can choose $c_1(s) = \max\|W'(s)\| = \mathcal{O}\left( h(\|H_0\| + \|H_1\| ) \right)$. 

    Now we estimate the second-order derivative of $W(s)$. 
    Using the product rule again, we can compute 
    \begin{align}
        W''(s) & = \sum_{k=0}^{K_p} \left( \prod_{k'=k+1}^{K_p} V_{k'}(s) \right) V_k''(s)  \left( \prod_{k'=0}^{k-1} V_{k'}(s)  \right) \nonumber\\
        & \quad + 2 \sum_{0 \leq k_0 < k_1 \leq K_p} \left( \prod_{k'=k_1+1}^{K_p} V_{k'}(s) \right)V_{k_1}'(s) \left( \prod_{k'=k_0+1}^{k_1-1} V_{k'}(s) \right) V_{k_0}'(s)  \left( \prod_{k'=0}^{k_0-1} V_{k'}(s)  \right) \nonumber\\
        & = \sum_{k=0}^{K_p} \left( \prod_{k'=k+1}^{K_p} V_{k'}(s) \right) \left( V_{k,1}''(s) V_{k,0}(s) + 2 V_{k,1}'(s) V_{k,0}'(s) + V_{k,1}(s) V_{k,0}''(s) \right)  \left( \prod_{k'=0}^{k-1} V_{k'}(s)  \right) \nonumber\\
        & \quad + 2 \sum_{0 \leq k_0 < k_1 \leq K_p} \left( \prod_{k'=k_1+1}^{K_p} V_{k'}(s) \right)\left( V_{k_1,1}'(s) V_{k_1,0}(s) + V_{k_1,1}(s) V_{k_1,0}'(s) \right) \nonumber \\
        &\quad\quad\quad\quad\quad\quad \times \left( \prod_{k'=k_0+1}^{k_1-1} V_{k'}(s) \right) \left( V_{k_0,1}'(s) V_{k_0,0}(s) + V_{k_0,1}(s) V_{k_0,0}'(s) \right)  \left( \prod_{k'=0}^{k_0-1} V_{k'}(s)  \right). 
    \end{align}
    Then 
    \begin{align}
        \|W''(s)\| & = \mathcal{O} \left( \sum_{k=0}^{K_p} \left( \|V_{k,1}''(s)\| + \| V_{k,1}'(s) \| \| V_{k,0}'(s)\| + \| V_{k,0}''(s) \| \right) \right) \nonumber\\
        & \quad \quad + \mathcal{O} \left( \sum_{0 \leq k_0 < k_1 \leq K_p} \left( \|V_{k_1,1}'(s)\| + \|V_{k_1,0}'(s) \| \right) \left( \| V_{k_0,1}'(s) \| + \| V_{k_0,0}'(s)\| \right)  \right). 
    \end{align}
    Notice that 
    \begin{align}
        V_{k,0}''(s)  & =  \left( i\alpha_{p,k}h f''((sT+\gamma_{p,k}h)/T)) H_0 \right) \exp\left( -i\alpha_{p,k}h (1-f((sT+\gamma_{p,k}h)/T)) H_0 \right)  \nonumber\\
        & \quad + \left( i\alpha_{p,k}h f'((sT+\gamma_{p,k}h)/T)) H_0 \right)^2 \exp\left( -i\alpha_{p,k}h (1-f((sT+\gamma_{p,k}h)/T)) H_0 \right) ,\\
         V_{k,1}''(s) &= \left( -i\beta_{p,k} h f''((sT+\delta_{p,k}h)/T) H_1 \right)  \exp\left( -i\beta_{p,k} h f((sT+\delta_{p,k}h)/T) H_1\right) \nonumber\\
         & \quad + \left( -i\beta_{p,k} h f'((sT+\delta_{p,k}h)/T) H_1 \right)^2  \exp\left( -i\beta_{p,k} h f((sT+\delta_{p,k}h)/T) H_1\right) , 
    \end{align}
    so we have $\|V_{k,0}''(s)\| = \mathcal{O} ( h\|H_0\| + h^2 \|H_0\|^2 )$ and $\|V_{k,1}''(s)\| = \mathcal{O} ( h\|H_1\| + h^2 \|H_1\|^2 )$. 
    Together with $\|V_{k,0}'(s)\| = \mathcal{O}(h\|H_0\|)$ and $\|V_{k,1}'(s)\| = \mathcal{O}(h\|H_1\|)$, we have 
    \begin{equation}
        \|W''(s)\| = \mathcal{O} \left( h(\|H_0\|+\|H_1\|) + h^2 (\|H_0\|+\|H_1\|)^2 \right). 
    \end{equation}
    Then we can choose $c_2(s) = \mathcal{O} \left( h(\|H_0\|+\|H_1\|) + h^2 (\|H_0\|+\|H_1\|)^2 \right)$ according to~\cref{lem:finite_difference_bound_general}. 
\end{proof}

For the simplified product formula, the same bound also holds because the simplified product formula is just a special case of the general product formula. 

\begin{lemma}\label{lem:c1_c2_product_formula_simplified}
    Let 
    \begin{equation}
        W(s) = U_{\text{spf},p}(sT+h,sT) = \prod_{k=0}^{K_p} \exp\left(  -i\beta_{p,k} h f(s) H_1\right)  \exp\left( -i\alpha_{p,k}h (1-f(s)) H_0 \right) 
    \end{equation}
   be the simplified product formula defined in~\cref{eqn:high_order_trotter_simplified}, where $h = T/T_d$ is the time step size. 
    Suppose that both $|f'(s)|$ and $|f''(s)|$ are uniformly bounded over $[0,1]$. 
    Then we can choose 
    \begin{equation}
        c_1(s) = \mathcal{O}(h(\|H_0\|+\|H_1\|)), \quad c_2(s) = \mathcal{O}(h(\|H_0\|+\|H_1\|) + h^2(\|H_0\|+\|H_1\|)^2). 
    \end{equation}
\end{lemma}

\section{Numerical methods and their standard complexity estimates}\label{app:standard_analysis}

In the main text, we focus on two types of numerical methods for discretizing~\cref{eqn:AQC_dynamics}, namely the first-order exponential integrator~\cref{eqn:exponential_integrator_1st} and the product formula~\cref{eqn:high_order_trotter_general}. 
Here we present their standard time discretization error bounds and complexity estimates. 

\subsection{Standard time discretization error bounds}

For the first-order exponential integrator, its standard error bound is given as follows. 

\begin{lemma}\label{lem:exponential_standard_bdd}
    Let $U_{\exp}(t+h,t)$ be the first-order exponential integrator defined in~\cref{eqn:exponential_integrator_1st}. 
    Suppose that $|f'(s)|$ is uniformly bounded over $[0,1]$. 
    Then we have $\left\|U_{\exp}(t+h,t) - \mathcal{T}\exp\left( \int_{t}^{t+h} H(\tau/T)d\tau\right)\right\| \leq C_{1,H} h^{2}$, where 
    \begin{equation}
        C_{1,H} = \mathcal{O}\left(\frac{\|H_0\|+\|H_1\|}{T}\right). 
    \end{equation}
\end{lemma}
\begin{proof}
    For a fixed time $t$ and a step size $h$, let $\tau \in [0,h]$. 
    The exact evolution operator from $t$ to $t+\tau$ is 
    \begin{equation}
        U(t+\tau,t) = \mathcal{T} \exp\left( -i \int_0^{\tau} H(t/T+\tau'/T)d\tau'\right), 
    \end{equation}
    and the first-order exponential integrator can be written as 
    \begin{equation}
        U_{\exp}(t+\tau,t) = \exp(-i\tau H(t/T)), 
    \end{equation}
    These two operators satisfy the differential equations, respectively, 
    \begin{equation}
        \frac{d}{d\tau} U(t+\tau,t) = -i H(t/T+\tau/T) U(t+\tau,t), \quad U(t,t) = I,
    \end{equation}
    and 
    \begin{equation}
        \frac{d}{d\tau} U_{\exp}(t+\tau,t) = -iH(t/T)  U_{\exp}(t+\tau,t), \quad U_{\exp}(t,t) = I. 
    \end{equation}
    According to the variation of parameters formula (as known as Duhamel's principle), we may write $U_{\exp}$ as a perturbation of $U$, that is, 
    \begin{align}
        U_{\exp}(t+\tau,t) &= U(t+\tau,t) + i\int_0^\tau U(t+\tau,t+\tau') (H(t/T+\tau'/T) - H(t/T)) U_{\exp}(t+\tau',t) d\tau' \nonumber\\
        & = U(t+\tau,t) + i\int_0^\tau U(t+\tau,t+\tau') \int_{0}^{\tau'} \frac{1}{T} H^{(1)}(t/T+\tau''/T) U_{\exp}(t+\tau',t) d\tau'' d\tau'. 
    \end{align}
    In the second line we use the fundamental theorem of calculus, and we are getting an extra $1/T$ scaling because the Hamiltonian is changing slowly on the scale of $T$ (technically, we are taking the derivative with respect to $\tau'$ for the function $H(t/T+\tau'/T)$, in which the argument $\tau'$ is divided by $T$). 
    Since $\|H^{(1)}(s)\| = |f(s)|\|H_1-H_0\| \leq \mathcal{O}(\|H_0\|+\|H_1\|)$, we have 
    \begin{equation}
        \|U_{\exp}(t+h,t) - U(t+h,t)\| \leq \mathcal{O}\left( \frac{(\|H_0\|+\|H_1\|)h^2}{T} \right). 
    \end{equation}
\end{proof}

We show an error bound for the Trotter-Suzuki formula in the next result, which is a direct corollary of the general analysis in~\cite{WiebeBerryHoyerEtAl2010}. 
We also expect a similar error bound for general product formulae beyond Trotter-Suzuki. 

\begin{lemma}
    Let $U_{\text{pf},p}(t+h,t)$ denote the $p$-th order time-dependent product formula obtained by Sukuzi recursion. 
    Suppose that the scheduling function $f(s)$ is smooth with uniformly bounded derivatives of any order, and $\|H_0\|+\|H_1\| \geq 1$. 
    Then, for any $p \geq 1$ and $T \geq 1$, we have 
    \begin{equation}
        \left\|U_{\text{pf},p}(t+h,t) - \mathcal{T}\exp\left( \int_{t}^{t+h} H(\tau/T)d\tau\right)\right\| \leq C_{p} (\|H_0\|+\|H_1\|)^{p+1} h^{p+1}, 
    \end{equation}
    where $C_p$ is a constant that depends on $p$ and the derivatives of $f(s)$. 
\end{lemma}

\begin{proof}
    According to~\cite[Theorem 3]{WiebeBerryHoyerEtAl2010}, there exists a $p$-dependent constant $\widetilde{C}_p$ such that 
    \begin{equation}
         \left\|U_{\text{pf},p}(t+h,t) - \mathcal{T}\exp\left( \int_{t}^{t+h} H(\tau/T)d\tau\right)\right\| \leq \widetilde{C}_{p} \Lambda_p^{p+1} h^{p+1}. 
    \end{equation}
    Here $\Lambda_p$ is a constant such that 
    \begin{equation}
        \Lambda_p \geq \max_{\tau \in [t,t+h]}\left\|\frac{d^m H(\tau/T)}{d\tau^m} \right\|^{1/(m+1)}, \quad \forall \, 0 \leq m \leq p. 
    \end{equation}
    By $H(s) = (1-f(s))H_0 + f(s)H_1$ and the triangle inequality, it suffices to choose
    \begin{equation}
        \Lambda_p \geq \max_{\tau\in[t,t+h]}\left( \left\| \frac{d^m (1-f(\tau/T))}{d\tau^m}  H_0 \right\| + \left\| \frac{d^m f(\tau/T)}{d\tau^m}  H_1 \right\| \right)^{1/(m+1)}, \quad \forall \, 0 \leq m \leq p. 
    \end{equation}
    Let $F_m$ be the upper bound of $|f^{(m)}(s)|$ for all $s\in[0,1]$, then we have 
    \begin{align}
        \left( \left\| \frac{d^m (1-f(t/T))}{dt^m}  H_0 \right\| + \left\| \frac{d^m f(t/T)}{dt^m}  H_1 \right\| \right)^{1/(m+1)} \leq \left( \frac{F_m}{T^m} (\left\|  H_0 \right\| + \left\| H_1 \right\| )\right)^{1/(m+1)} \leq F_m^{1/(m+1)} \left( \left\|  H_0 \right\| + \left\| H_1 \right\| \right). 
    \end{align}
    So we can choose $\Lambda_p = \max_{0\leq m \leq p} F_m^{1/(m+1)} \left( \left\|  H_0 \right\| + \left\| H_1 \right\| \right)$, and the claimed error bound can be obtained by defining $C_p = \max_{0\leq m \leq p} F_m^{(p+1)/(m+1)} \widetilde{C}_p$. 
\end{proof}

\subsection{Standard complexity estimates}

Now we show the standard approach to analyze the complexity of using discretized AQC to prepare an eigenstate. 
Specifically, let $\ket{\phi}$ denote the exact target eigenstate and
\begin{equation}
    \ket{\widetilde{\phi}} = \prod_{j=0}^{T_{d}-1} U_{\text{num}}((j+1)h,jh) \ket{\psi(0)}
\end{equation}
where $U_{\text{num}}$ is a $p$-th order local numerical propagator in the sense that the local truncation error can be bounded by 
\begin{equation}
    \text{LTE}_p \leq C_{p,H} h^{p+1}
\end{equation}
where $C_{p,H}$ is a constant that might depend on $p$ and certain norms of $H$, $H_0$, $H_1$, but is independent of the time step size $h$. 
Notice that in the first-order exponential integrator, this $C$ parameter also depends on $T$. 
We would like to study how many local numerical propagators, namely $M$, scale if we want $\|\ket{\widetilde{\phi}}\bra{\widetilde{\phi}} - \ket{\phi}\bra{\phi}\| \leq \epsilon$ for some $0 < \epsilon < 1$. 

A standard approach for estimating the final eigenstate error $\varepsilon_{\text{Eig}}$ is by using $\varepsilon_{\text{AQC}}$ as the upper bound from the continuous adiabatic theorem  plus the error from the time discretization $\varepsilon_{\text{Disc}}$ : 
\begin{equation}\label{eqn:error_relation}
\varepsilon_{\text{Eig}} \leq \varepsilon_{\text{AQC}} + \varepsilon_{\text{Disc}}. 
\end{equation}
To bound the overall approximation error by $\mathcal{O}(\epsilon)$, it suffices to bound both the continuous adiabatic error and the time discretization error by $\mathcal{O}(\epsilon)$. 
In order to bound the time discretization error, which is $C_{p,H}Th^p$ for a $p$-th order method, we can choose 
\begin{equation}
    h = \mathcal{O}\left(\frac{\epsilon^{1/p}}{C_{p,H}^{1/p}T^{1/p}}\right).
\end{equation}
Therefore the overall number of queries to $U_{\text{num}}$ can be bounded by 
\begin{equation}
    T_{d} = \frac{T}{h} = \mathcal{O} \left(\frac{C_{p,H}^{1/p} T^{1+1/p}}{\epsilon^{1/p}}\right),
\end{equation}
and we only need to further study the scaling of the evolution time $T$, which needs to be chosen large enough such that the adiabatic error is bounded as well. 
For the general second-order continuously differentiable interpolation function $f(s)$ which is independent of the Hamiltonians, \cref{lem:AT_JRS} implies that the adiabatic error is bounded by 
\begin{equation}
    \mathcal{O}\left( \frac{1}{T} \left( \frac{\|H_0\|+\|H_1\|}{\Delta_*^2} + \frac{(\|H_0\|+\|H_1\|)^2}{\Delta_*^3} \right) \right) \leq \mathcal{O}\left(  \frac{(\|H_0\|+\|H_1\|)^2}{T \Delta_*^3}  \right), 
\end{equation}
then the evolution time should be chosen as 
\begin{equation}
    T = \mathcal{O}\left(\frac{(\|H_0\|+\|H_1\|)^2}{\Delta_*^3 \epsilon}\right). 
\end{equation}
The overall number of steps is thus 
\begin{equation}
    T_{d} = \mathcal{O}\left( \frac{C_{p,H}^{1/p} (\|H_0\|+\|H_1\|)^{2+2/p} }{\Delta_*^{3+3/p} \epsilon^{1+2/p}} \right). 
\end{equation}

Now we consider the boundary cancellation condition. 
Here we further assume that both $\|H_0\|$ and $\|H_1\|$ are bounded by $1$. 
We can use~\cref{lem:AT_exp} to bound the adiabatic error, which implies the choice of evolution time to be 
\begin{equation}
    T = \mathcal{O}\left(\frac{1}{\Delta_*^2} \left(\log\left(\frac{1}{\Delta_* \epsilon}\right)\right)^{1+\alpha}\right)
\end{equation}
and thus the overall number of steps to be 
\begin{equation}
    T_{d} = \mathcal{O}\left(\frac{1}{\Delta_*^{2+2/p} \epsilon^{1/p} } \left(\log\left(\frac{1}{\Delta_* \epsilon}\right)\right)^{(1+\alpha)(1+1/p)} \right). 
\end{equation}

\section{Complexity estimate of the high-order simplified product formulae without gap condition on the walk operator}\label{app:proof_high_order_Trotter}

Here we give a proof of~\cref{cor:Trotter_linear}. 
The first step is to establish a connection between the gap of the Trotter walk operator and the Hamiltonian. 
We show this connection in the following lemma. 

\begin{lemma}\label{lem:gap_Trotter_high_order}
    For $s\in[0,1-h/T]$ and a scheduling function $f(s)$, let $H(s) = (1-f(s))H_0+f(s)H_1$ be the interpolating Hamiltonian and $U_{\text{spf},p}(sT+h,sT)$ be the simplified product formula as in~\cref{eqn:high_order_trotter_simplified}. 
    Let $\Delta_H(s)$ and $\Delta_{U,p}(s)$ denote the spectral gap of $H(s)$ and $U_{\text{spf},p}$, respectively. 
    Then, for any $h \leq 1/(\|H_0\|+\|H_1\|)$ and $p \geq 1$, there exist $p$-dependent constants $c_p$ such that 
    \begin{align}
        &h\Delta_H(s) - c_p h^{p+1} \sum_{\gamma_0,\cdots,\gamma_{p} \in \left\{0,1\right\}}\|[H_{\gamma_p},\cdots,[H_{\gamma_1},H_{\gamma_0}]]\| \nonumber\\
        &\leq \Delta_{U,p}(s) \nonumber\\
        &\leq h\Delta_H(s) + c_p h^{p+1} \sum_{\gamma_0,\cdots,\gamma_{p} \in \left\{0,1\right\}}\|[H_{\gamma_p},\cdots,[H_{\gamma_1},H_{\gamma_0}]]\|. 
    \end{align}
\end{lemma}
\begin{proof}
Notice that the spectral gaps of $e^{-i h H(s)}$ and $h H(s)$ are of the same value for all $h \leq 1/(\|H_0\|+\|H_1\|)$. 
So it suffices to investigate the spectral gaps of $U_{\text{spf},p}(sT+h,sT)$ and $e^{-ihH(s)}$. 
For every fixed $s$, we may bound the difference between $U_{\text{spf},p}(sT+h,sT)$ and $e^{-ihH(s)}$ by using the error bound of the time-independent product formula in~\cite[Theorem 11]{ChildsSuTranEtAl2019}, which gives 
    \begin{equation}
        \left\| U_{\text{spf},p}(sT+h,sT) -  e^{-iH(s)}\right\| \leq \widetilde{c}_p h^{p+1}\left( \sum_{\gamma_0,\cdots,\gamma_{p} \in \left\{0,1\right\}}\|[H_{\gamma_p},\cdots,[H_{\gamma_1},H_{\gamma_0}]]\| \right)
    \end{equation}
    for a $p$-dependent constant $\widetilde{c}_p$. 

    According to the eigenvalue perturbation theorem of Ref.~\cite{BhatiaDavis1984,ElsnerHe1993}, there exists an ordering $\lambda_j(s)$ and $\widetilde{\lambda}_j(s)$ of the eigenvalues of $U_{\text{spf},p}(sT+h,sT)$ and $e^{-iH(s)}$, respectively, such that
    \begin{equation}
    \max_j |\lambda_j(s) - \widetilde{\lambda}_j(s)| \leq \left\| U_{\text{spf},p}(sT+h,sT) -  e^{-iH(s)}\right\|. 
    \end{equation}
    Let $\lambda_0(s)$ and $\lambda_1(s)$ denote the two eigenvalues that determine the gap, then the gap of $U_{\text{spf},p}(sT+h,sT)$ can be bounded by
    \begin{align}
    |\lambda_1(s) - \lambda_0(s)|_{\mathbb{S}^1} & \geq |\widetilde{\lambda}_1(s) - \widetilde{\lambda}_0(s)|_{\mathbb{S}^1} - |\widetilde{\lambda}_1(s) - \lambda_1(s)|_{\mathbb{S}^1} - |\widetilde{\lambda}_0(s) - \lambda_0(s)|_{\mathbb{S}^1} \nonumber\\
    & \geq h\Delta_H(s) - \frac{\pi}{2} |\widetilde{\lambda}_1(s) - \lambda_1(s)| - \frac{\pi}{2}|\widetilde{\lambda}_0(s) - \lambda_0(s)|  \nonumber\\
    & \geq h\Delta_H(s) - \pi \widetilde{c}_p h^{p+1}\left( \sum_{\gamma_0,\cdots,\gamma_{p} \in \left\{0,1\right\}}\|[H_{\gamma_p},\cdots,[H_{\gamma_1},H_{\gamma_0}]]\| \right), 
    \end{align}
and 
\begin{align}
    |\lambda_1(s) - \lambda_0(s)|_{\mathbb{S}^1} & \leq |\widetilde{\lambda}_1(s) - \widetilde{\lambda}_0(s)|_{\mathbb{S}^1} + |\widetilde{\lambda}_1(s) - \lambda_1(s)|_{\mathbb{S}^1} + |\widetilde{\lambda}_0(s) - \lambda_0(s)|_{\mathbb{S}^1} \nonumber\\
    & \leq h\Delta_H(s) + \frac{\pi}{2} |\widetilde{\lambda}_1(s) - \lambda_1(s)| + \frac{\pi}{2}|\widetilde{\lambda}_0(s) - \lambda_0(s)|  \nonumber\\
    & \leq h\Delta_H(s) + \pi \widetilde{c}_p h^{p+1}\left( \sum_{\gamma_0,\cdots,\gamma_{p} \in \left\{0,1\right\}}\|[H_{\gamma_p},\cdots,[H_{\gamma_1},H_{\gamma_0}]]\| \right). 
\end{align} 
In the second line of both estimates, we use $\theta \leq \frac{\pi}{2}\sin \theta$ for all $\theta \in [0,\pi/2]$. 
The proof is completed by choosing $c_p = \pi \widetilde{c}_p$. 
\end{proof}

Now we are ready to prove~\cref{cor:Trotter_linear}. 

\begin{proof}[Proof of~\cref{cor:Trotter_linear}]
    The sketch of the proof is as follows. 
    We will first give a lower bound for the fixed-time spectral gap of the walk operator $W(s)$, then give a lower bound for the multistep gap by verifying the conditions in~\cref{lem:multistep_gap}, and apply the discrete adiabatic theorem to bound the overall error between the actual and ideal evolution, from which we can infer the choices of $T$ and $h$.  
    
    Let $T_{d} = T/h$, and we assume $T_{d}$ is an integer. 
    We first bound the fixed-time gap of $W(s)$ defined by~\cref{eqn:high_order_trotter_simplified}. 
    According to the choice of $h$ with proper constant factor, we can make 
    \begin{equation}
        c_p h^{p+1} \sum_{\gamma_0,\cdots,\gamma_{p} \in \left\{0,1\right\}}\|[H_{\gamma_p},\cdots,[H_{\gamma_1},H_{\gamma_0}]]\| \leq \frac{1}{2} h \Delta_*, 
    \end{equation}
    where $c_p$ is the constant factor in~\cref{lem:gap_Trotter_high_order}. 
    Then, using~\cref{lem:gap_Trotter_high_order}, the gap of $W(s)$ is bounded from below by the gap of $\frac{1}{2} h \Delta(s)$, which has a further lower bound $\frac{1}{2}h \Delta_*$. 
    
    For the multistep gap, we can use~\cref{lem:multistep_gap}, and we need to verify that $T_d$ is sufficiently large. 
    Specifically, suppose that our choice of $T$ is 
    \begin{equation}
        T \geq C \frac{(\|H_0\|+\|H_1\|)^2}{ \Delta_*^3 \epsilon}
    \end{equation}
    for a constant $C > 0$. 
    Then we have 
    \begin{equation}
        T_d = T/h \geq C \frac{(\|H_0\|+\|H_1\|)^2}{ h \Delta_*^3 \epsilon} \geq C \frac{(\|H_0\|+\|H_1\|)^2}{ h \Delta_*^3}. 
    \end{equation}
    Notice that $\Delta_* \leq 2\|H(s)\| \leq 2(\|H_0\|+\|H_1\|)$, so we have 
    \begin{equation}
        T_d \geq \frac{C}{4} \frac{1}{ h\Delta_*}. 
    \end{equation}
    According to~\cref{lem:c1_c2_product_formula_simplified}, we have $c_1(s) = \mathcal{O}(h(\|H_0\|+\|H_1\|)) = \mathcal{O}(1)$, so we can choose a sufficiently large constant $C$ such that $T_d \geq \frac{2\pi}{h\Delta_*/2} \sup c_1(s)$. 
    Then~\cref{lem:multistep_gap} ensures that the multistep gap of $W(s)$ is bounded from below by $h\Delta_*/4$. 
    
    Therefore, the error bound in~\cref{thm:trotter_linear} becomes 
    \begin{align}
        & \quad \mathcal{O}\left( \frac{h}{T}\left(\frac{ h (\|H_0\|+\|H_1\|)}{h^2\Delta_*^2} + \frac{ h^2  (\|H_0\|+\|H_1\|)^2}{h^2\Delta_*^2} + \frac{ h^2 (\|H_0\|+\|H_1\|)^2}{h^3\Delta_*^3}\right) \right)\nonumber \\
        & \leq \mathcal{O}\left( \frac{1}{T}\left(\frac{  (\|H_0\|+\|H_1\|)}{\Delta_*^2} + \frac{ h (\|H_0\|+\|H_1\|)^2}{\Delta_*^2} + \frac{ (\|H_0\|+\|H_1\|)^2}{\Delta_*^3}\right) \right) \nonumber\\
        & \leq \mathcal{O}\left( \frac{ (\|H_0\|+\|H_1\|)^2}{T\Delta_*^3}\right), 
    \end{align}
    and the desired choices of $T$ and $h$ directly follow from this error bound and the conditions for the multistep gap. 
\end{proof}

\section{High-order discrete adiabatic theorem}\label{app:DAT_exp}

Here we provide more discussion on the high-order discrete adiabatic theorem, which is claimed in~\cite{DKS98} and we state as~\cref{lem:DAE_exp} in the main text. 
First, in the statement of~\cref{lem:DAE_exp}, we adapt a slightly different statement compared to the original version in~\cite{DKS98} for the purpose of clearer presentation. 
We will explain how to get~\cref{lem:DAE_exp} from there. 
More importantly, we will then point out the missing steps in the analysis of~\cite{DKS98} and explain why these missing steps imply that the proof in~\cite{DKS98} is not valid. 
Meanwhile, we also show a numerical result verifying that, despite the flawed proof, the claim of the high-order discrete adiabatic theorem is very likely to be correct.

\subsection{A slight variant}

We show how to derive~\cref{lem:DAE_exp} from the original result in~\cite{DKS98}. 
We need to introduce some operators before proceeding. 
Let $P(s)$ and $Q(s)$ denote the spectral projections onto $\sigma_P(s)$ and $\sigma_Q(s)$, respectively. 
Define 
\begin{align}
    S(s) &= P(s+1/T_d)P(s) + Q(s+1/T_d)Q(s), \\
    v(s) &= \sqrt{S(s)S(s)^{\dagger}}, \\
    V(s) &= v(s)^{-1} S(s). \label{eqn:def_op_V}
\end{align}

The following is the original result from~\cite{DKS98}. 

\begin{lemma}[{Higher-order discrete adiabatic theorem,~\cite[Theorem 3]{DKS98}}]\label{lem:DAE_exp_DKS}
    Suppose that $V(s)-I$ is supported in $(0,s_{\star})$. 
    Then, for any integer $n$ such that $n/T_d \notin (0,s_{\star})$ and for any positive integer $k$, we have 
    \begin{equation}
        \|Q(0) U_A^{\dagger}(n/T_d) U(n/T_d) P(0)\| \leq \frac{C_k}{T_d^{k}},
    \end{equation}
    where $C_k$ is a constant independent of $n$ and $T_d$. 
\end{lemma}

Notice that~\cref{lem:DAE_exp_DKS} introduces a new parameter $s_{\star}$ which could be larger than $1$. 
Such an extension allows us to bound the discrete adiabatic error when the number of the walk operators is larger than $T$. 
Now we prove~\cref{lem:DAE_exp} from~\cref{lem:DAE_exp_DKS}. 

\begin{proof}[Proof of~\cref{lem:DAE_exp}]
    The idea is to smoothly extend $W(s)$ to a longer time interval $[0,3]$ to make it satisfy the condition in~\cref{lem:DAE_exp_DKS}. 
    Specifically, we consider a walk operator on the time interval $[0,3]$ and force it to be unchanged only on $[1,2]$ to make sure that it satisfies the boundary cancellation condition in the discrete sense as well. 
    On the time interval $[0,3]$, we define 
    \begin{equation}
        \widetilde{W}(s) = \begin{cases}
        W(0), & s\in[0,1),\\
        W(s-1), & s \in [1,2],\\
        W(1), & s \in (2,3]. 
        \end{cases}
    \end{equation}
    The operators associated with $\widetilde{W}(s)$ are denoted using the same letter with an upper tilde (for example $\widetilde{P}(s)$ and $\widetilde{Q}(s)$ denote the spectral projections of $\widetilde{W}(s)$). 
    Since $\widetilde{W}^{(k)}(1) = W^{(k)}(0) = 0$ and $\widetilde{W}^{(k)}(2) = W^{(k)}(1) = 0$ for all $k \geq 1$, the operator $\widetilde{W}(s)$ is a smooth extension of $W(s)$. 
    
    Notice that $\widetilde{W}(s)$ and thus $\widetilde{P}(s)$ and $\widetilde{Q}(s)$ remain unchanged on $[0,1] \cup [2,3]$. 
    Then, for any $s \in [0,1-1/T_d] \cup [2,3-1/T_d]$, we must have 
    \begin{equation}
        \widetilde{S}(s) = \widetilde{P}(s+1/T_d)\widetilde{P}(s) + \widetilde{Q}(s+1/T_d)\widetilde{Q}(s) = \widetilde{P}(s)^2 + \widetilde{Q}(s)^2 = I, 
    \end{equation}
    so $\widetilde{v}(s) = \sqrt{\widetilde{S}(s)\widetilde{S}(s)^{\dagger}} = I$ and $\widetilde{V}(s) = \widetilde{v}(s)^{-1} \widetilde{S}(s) = I$. 
    This means that $\widetilde{V}(s) - I$ is supported in $(0,2)$ and satisfies the assumption of~\cref{lem:DAE_exp_DKS}. 
    As a result, we have 
    \begin{equation}\label{eqn:DAE_exp_intermediate}
        \|\widetilde{Q}(0) {\widetilde{U}_A}^{ \dagger}(2) \widetilde{U}(2) \widetilde{P}(0)\| \leq \frac{C_k}{T_d^{k}},
    \end{equation}
    where $C_k$ is a constant independent of $T_d$. 
    
    We now show how we may obtain a bound for the desired quantity from~\cref{eqn:DAE_exp_intermediate}. 
    First, we have 
    \begin{align}
        \|\widetilde{Q}(0) \widetilde{U}_A^{\dagger}(2) \widetilde{U}(2) \widetilde{P}(0)\| & = \|\widetilde{U}_A^{ \dagger}(2) \widetilde{U}(2) \widetilde{P}(0) - \widetilde{P}(0) \widetilde{U}_A^{ \dagger}(2) \widetilde{U}(2) \widetilde{P}(0)\|  \nonumber \\
        & = \|\widetilde{U}_A^{\dagger}(2) \widetilde{U}(2) \widetilde{P}(0) - \widetilde{U}_A^{ \dagger}(2) \widetilde{P}(2) \widetilde{U}(2) \widetilde{P}(0)\| \nonumber \\
        & = \|\widetilde{U}(2) \widetilde{P}(0) - \widetilde{P}(2) \widetilde{U}(2) \widetilde{P}(0)\|, 
    \end{align}
    where we use~\cref{eqn:prop_UA} in the second equality. 
    Then, since $\ket{\psi}$ is within the eigenspace corresponding to $\sigma_{\widetilde{P}}(0)$, we have 
    \begin{align}
        \|\widetilde{U}(2)\ket{\psi} - \widetilde{P}(2) \widetilde{U}(2) \ket{\psi}\| & = \|\widetilde{U}(2)\widetilde{P}(0)\ket{\psi} - \widetilde{P}(2) \widetilde{U}(2) \widetilde{P}(0) \ket{\psi}\| \nonumber \\
        & \leq \|\widetilde{U}(2)\widetilde{P}(0)  - \widetilde{P}(2) \widetilde{U}(2) \widetilde{P}(0) \| \nonumber \\
        & = \|\widetilde{Q}(0) \widetilde{U}_A^{\dagger}(2) \widetilde{U}(2) \widetilde{P}(0)\| \nonumber \\
        & \leq \frac{C_k}{T_d^k},
    \end{align}
    where $C_k$ is the constant in~\cref{lem:DAE_exp_DKS}. 
    Finally, notice that, by definition, $\widetilde{P}(2) = P(1)$, and $\ket{\psi}$ is an eigenstate of $\widetilde{W}(s)$ for all $s\in[0,1]$. Suppose the corresponding eigenvalue is $e^{i\theta}$ for a real number $\theta$, then 
    \begin{align}
        \|\widetilde{U}(2)\ket{\psi} - \widetilde{P}(2) \widetilde{U}(2) \ket{\psi}\| & = \left\|\left(\prod_{j=T_d}^{2T_d-1} \widetilde{W}(j/T_d)\right)\left(\prod_{j=0}^{T_d-1} \widetilde{W}(j/T_d)\right)\ket{\psi} - P(1) \left(\prod_{j=T_d}^{2T_d-1} \widetilde{W}(j/T_d)\right)\left(\prod_{j=0}^{T_d-1} \widetilde{W}(j/T_d)\right) \ket{\psi}\right\| \nonumber\\
        & = \left\|\left(\prod_{j=0}^{T_d-1} W(j/T_d)\right)W(0)^{T_d}\ket{\psi} - P(1) \left(\prod_{j=0}^{T_d-1} W(j/T_d)\right)W(0)^{T_d}\ket{\psi}\right\| \nonumber\\
        & = \left\|e^{i T_d\theta} U(1)\ket{\psi} - e^{i T_d \theta} P(1) U(1) \ket{\psi}\right\| \nonumber\\
        & = \left\|U(1)\ket{\psi} - P(1) U(1) \ket{\psi}\right\|. 
    \end{align}
    Therefore we have 
    \begin{equation}
        \left\|U(1)\ket{\psi} - P(1) U(1) \ket{\psi}\right\| = \|\widetilde{U}(2)\ket{\psi} - \widetilde{P}(2) \widetilde{U}(2) \ket{\psi}\| \leq \frac{C_k}{T_d^k}. 
    \end{equation}
\end{proof}

\subsection{Missing steps in the proof of the high-order discrete adiabatic theorem and a numerical validation}\label{app:DKS_missing_step}

We first follow the notations in~\cite{DKS98} to define several operators. 
Let $P(s)$ and $Q(s)$ be the spectral projections as before, and $V(s)$ is defined through~\cref{eqn:def_op_V}. 
Let 
\begin{align}
    \Omega(s) &= U_A^{\dagger}(s)U(s), \\
    \Theta(s) &= U_A^{\dagger}(s+1/T_d) V^{\dagger}(s) U_A(s+1/T_d), \\
    K(s) &= T_d(1-\Theta(s)). 
\end{align}
It can be shown that the Volterra equation holds as 
\begin{equation}
    \Omega(n/T_d) = 1 - \frac{1}{T_d} \sum_{k=0}^{n-1} K(k/T_d) \Omega(k/T_d). 
\end{equation}
In~\cite{DKS98}, another set of operators $\left\{ \Omega_j(s) \right\}$ is defined in a fixed-point-iteration fashion -- let $\Omega_0(s) = 1$ for all $s$, and 
\begin{equation}
    \Omega_j(n/T_d) = - \frac{1}{T_d} \sum_{k=0}^{n-1} K(k/T_d) \Omega_{j-1}(k/T_d). 
\end{equation}
By taking the summation over $j$, we can see that $\sum_{j=0}^{\infty} \Omega_j(s)$ is expected to converge to $\Omega(s)$. 

Lemma 1 in~\cite{DKS98} attempts to show that $Q_0 \Omega_j(1) P_0$ is of order $\mathcal{O}(T_d^{-k})$ for arbitrary positive integer $k$ under the boundary cancellation condition, and the high-order discrete adiabatic theorem becomes a consequence of this lemma. 
This was proven using mathematical induction over $j$. 
However, in the proof of~\cite[Lemma 1]{DKS98}, the induction assumption (Eq. 40 in~\cite{DKS98}) was taken to be $Q_0 \Omega_j(n/T_d) P_0 = \mathcal{O}(T_d^{-k})$ for all $0 \leq n \leq T_d$, which is much stronger than the original claim which only holds when $n = 0$ and $T_d$.  
Furthermore, intuitively, the operator $Q_0 \Omega_j(n/T_d) P_0$ should not scale $\mathcal{O}(T_d^{-k})$ inside the time interval since the boundary cancellation condition is expected to accelerate the convergence only near the boundary. 

We verify our statements numerically using a $4$-level model. 
We consider the walk operator to be the first-order exponential integrator with time step size $1$, \emph{i.e.}, $W(s) = e^{-i H(s)}$, where $H(s) = (1-f(s))H_0 + f(s) H_1$. 
We choose the Hamiltonians $H_j = Q_j^{\dagger} D_j Q_j$ for $j = 0,1$, where $D_0 = \text{diag}(0.5,0.8,1.2,1.4)$ and $D_1 = \text{diag}(0.3,1.0,1.5,1.9)$ are two diagonal matrices, $Q_0$ is the orthonormal basis of the matrix
\begin{equation}
    \left( \begin{array}{cccc}
        2 & 1 & 0 & 1 \\
        1 & 2 & 1 & 0 \\
        0 & 1 & 2 & 1 \\
        1 & 0 & 1 & 2
    \end{array} \right), 
\end{equation}
and $Q_1$ is the orthonormal basis of 
\begin{equation}
    \left( \begin{array}{cccc}
        3 & -0.5 & 0 & -2 \\
        -0.5 & 3 & 1 & 0 \\
        0 & 1 & 3 & -1 \\
        -2 & 0 & -1 & 3
    \end{array} \right). 
\end{equation}
The scheduling function is the glue function $f(s) = c^{-1} \int_0^s \exp\left(-\frac{1}{s'(1-s')}\right) ds'$, where $c$ is the normalization constant such that $f(1) = 1$. 
Notice that such a scheduling function satisfies the boundary cancellation condition. 

\begin{figure}[t]
    \centering
    \includegraphics[width = 0.45\textwidth]{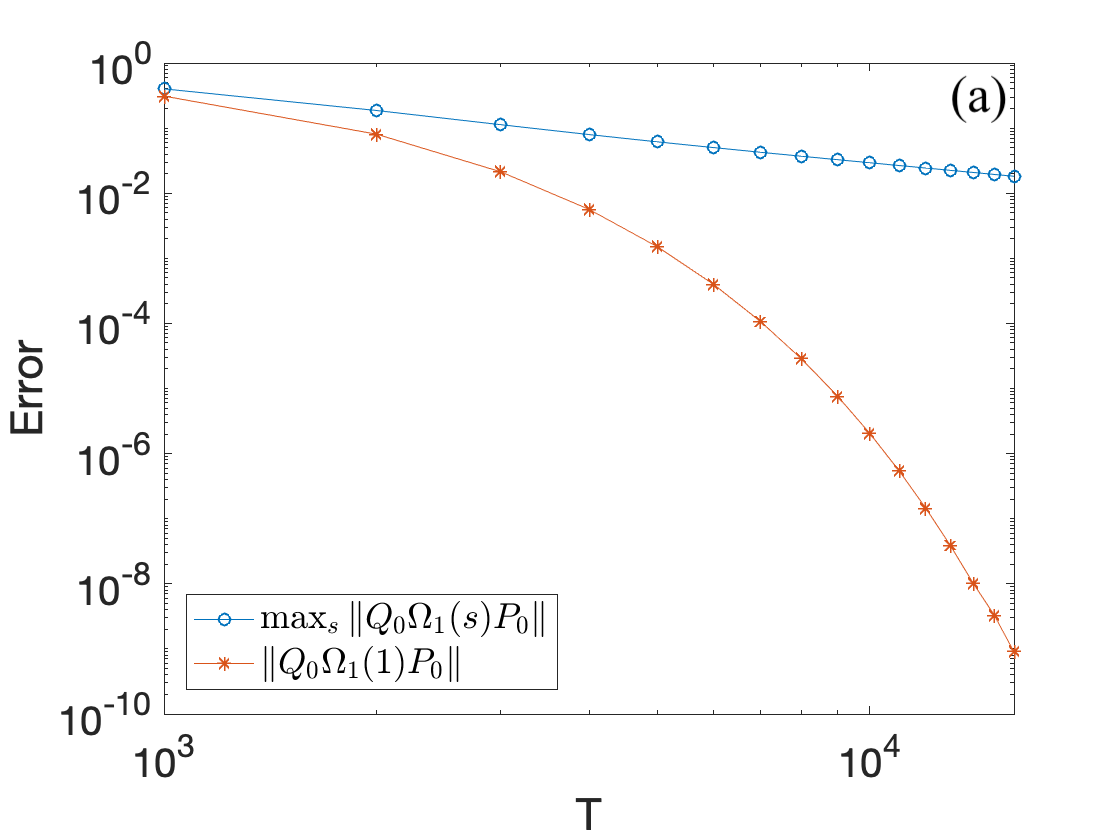}
    \includegraphics[width = 0.45\textwidth]{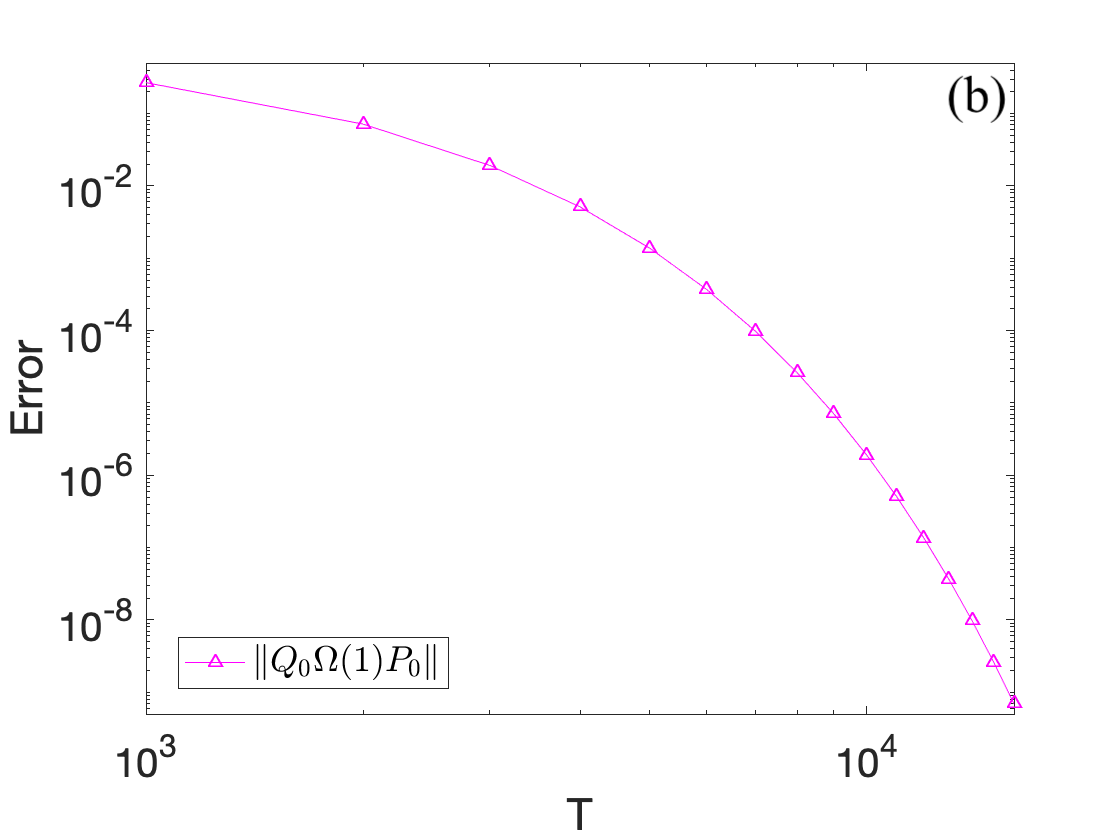}
    \caption{Numerical tests of the high-order discrete adiabatic theorem and its intermediate step. Left: numerical scaling of $\|Q_0 \Omega_1(s) P_0\|$ inside the interval and at the boundary. Right: numerical scaling of $\|Q_0 \Omega(s) P_0\|$ at the boundary $s = 1$. }
    \label{fig:DKS_bounds}
\end{figure}

Figure \ref{fig:DKS_bounds} (a) shows the numerical scaling of $\max_s \|Q_0 \Omega_1(s) P_0\|$ and $\|Q_0 \Omega_1(1) P_0\|$ under different choices of $T_d$. 
Here we test the scaling for $\Omega_1$ instead of $\Omega_0$ because $\Omega_0$ is always the identity matrix so the orthogonal projection is always $0$. 
Our numerical results suggest that $\max_s \|Q_0 \Omega_1(s) P_0\|$ is of order $\mathcal{O}(1/T_d)$ while $\|Q_0 \Omega_1(1) P_0\|$ decays super-polynomially, which verifies that the induction assumption used in~\cite{DKS98} is questionable.

Although the proof presented in~\cite{DKS98} seems flawed, we believe that the conclusion of the high-order discrete adiabatic theorem still holds true. 
We verify this using the same example, and our numerical results shown in~\cref{fig:DKS_bounds} (b) confirms that the adiabatic error $\|Q_0 \Omega(1) P_0\|$ at the boundary indeed converges super-polynomially in $1/T_d$.

\section{Spectral gap of the walk operator in adiabatic Grover search}\label{app:Grover_gap}

\subsection{Direct computations}
Here we determine the spectral gap of the walk operator $W(s)$ for a step of Trotter as applied to a search problem in \cref{sec:Grover}.
First we diagonalise $H_0$ using~\cref{eqn:Grover_H0_H1_2dim} as 
\begin{equation}
    H_0 = \left(\begin{array}{cc}
        \sqrt{M/N} & \sqrt{(N-M)/N} \\
        \sqrt{(N-M)/N} & -\sqrt{M/N}
    \end{array}\right)
    \left(\begin{array}{cc}
        0 & 0 \\
        0 & 1
    \end{array}\right)
    \left(\begin{array}{cc}
        \sqrt{M/N} & \sqrt{(N-M)/N} \\
        \sqrt{(N-M)/N} & -\sqrt{M/N}
    \end{array}\right). 
\end{equation}
That gives the walk operator $W(s)$ as
\begin{align}
    W(s) &= e^{-if(s)H_1} e^{-i(1-f(s))H_0} \nn
    &= \left(\begin{array}{cc}
        1 & 0 \\
        0 & e^{-i f(s)} 
    \end{array}\right)
    \left(\begin{array}{cc}
        \sqrt{M/N} & \sqrt{(N-M)/N} \\
        \sqrt{(N-M)/N} & -\sqrt{M/N}
    \end{array}\right)
    \left(\begin{array}{cc}
        1 & 0 \\
        0 & e^{-i (1-f(s))} 
    \end{array}\right)
    \left(\begin{array}{cc}
        \sqrt{M/N} & \sqrt{(N-M)/N} \\
        \sqrt{(N-M)/N} & -\sqrt{M/N}
    \end{array}\right) \nn
    &= \left(\begin{array}{cc}
        \frac{M}{N} + e^{-i(1-f(s))} \frac{N-M}{N} & \frac{\sqrt{M(N-M)}}{N} \left(1 - e^{-i(1-f(s))} \right)  \\
        \frac{\sqrt{M(N-M)}}{N} \left(e^{-if(s)} - e^{-i} \right) & e^{-if(s)} \frac{N-M}{N} + e^{-i} \frac{M}{N}
    \end{array}\right). 
\end{align}

The eigenvalues of $W(s)$ satisfy the following quadratic equation 
\begin{equation}
    \lambda^2 - \left[(1+e^{-i})\frac{M}{N} + (e^{-i(1-f(s))} + e^{-if(s)})\frac{N-M}{N}\right] \lambda + e^{-i} = 0. 
\end{equation}
So the eigenvalues can be represented as $\lambda = e^{-\frac{i}{2}} \mu$, where $\mu$ satisfies 
\begin{equation}
    \mu^2 - 2\left(\frac{M}{N}\cos\left(\frac{1}{2}\right) + \frac{N-M}{N} \cos\left(\frac{1}{2}-f(s)\right)\right)\mu + 1 = 0. 
\end{equation}
Solving this quadratic equation yields the eigenvalues 
\begin{align}
    \lambda_{\pm} &= e^{-\frac{i}{2}} \left(\xi \pm i\sqrt{1-\xi^2} \right), \\
    \xi &= \frac{M}{N}\cos\left(\frac{1}{2}\right) + \frac{N-M}{N} \cos\left(\frac{1}{2}-f(s)\right). 
\end{align}
The spectral gap of $W(s)$ is 
\begin{align}
    \Delta_{0,W}(s) &= 2\arccos \left( \frac{M}{N}\cos\left(\frac{1}{2}\right) + \frac{N-M}{N} \cos\left(\frac{1}{2}-f(s)\right) \right) \nn
    &= 2\arccos \left( \cos\left(\frac{1}{2}-f(s)\right) - \frac{M}{N}\left[\cos\left(\frac{1}{2}-f(s)\right) - \cos\left(\frac{1}{2}\right) \right] \right). 
\end{align}

\subsection{Proof of \texorpdfstring{\cref{lem:Grover_gap_walk}}{}}

This Lemma lower bounds the gap of the walk operator $W(s)$ for a step of Trotter in terms of the spectral gap of the Hamiltonian.
To bound this spectral gap, we may use~\cref{lem:gap_Trotter_1st_2nd} with $h=1$ to give 
    \begin{align}
        \Delta_W(s;N,M,f) &\geq \Delta_H(s;N,M,f) - \frac{\pi}{144\sqrt{3}}\left( 2\|[H_1,[H_1,H_0]]\| +  \|[H_0,[H_0,H_1]]\|\right) \nn
        & \geq \Delta_H(s;N,M,f) - \frac{\pi}{24\sqrt{3}}\|[H_0,H_1]\|. \label{eqn:Grover_gap_Trotter_Hamiltonian}
    \end{align}
    The commutator in the basis $\left\{\ket{e_0},\ket{e_1}\right\}$ can be computed as 
\begin{align}
    [H_0, H_1] = \left(\begin{array}{cc}
        0 & -\sqrt{M(N-M)}/N \\
        \sqrt{M(N-M)}/N & 0
    \end{array}\right), 
\end{align}
and 
\begin{equation}
    \|[H_0, H_1]\| = \sqrt{M(N-M)}/N. 
\end{equation}
Then~\cref{eqn:Grover_gap_Trotter_Hamiltonian} becomes 
\begin{align}
    \Delta_W(s;N,M,f) 
    & \geq \sqrt{(1-2f(s))^2 + \frac{4M}{N}f(s)(1-f(s))} - \frac{\pi}{24\sqrt{3}} \frac{\sqrt{M(N-M)}}{N} \nn
    & > \sqrt{(1-2f(s))^2 + \frac{4M}{N}f(s)(1-f(s))} - \frac{\pi}{24\sqrt{3}} \sqrt{\frac{M}{N}} \nn
    & > \frac{2}{3} \sqrt{(1-2f(s))^2 + \frac{4M}{N}f(s)(1-f(s))} \nn
    & = \frac{2}{3} \Delta_ H(s;N,M,f). 
\end{align}

\section{Proof of \texorpdfstring{\cref{thm:Grover_adiabatic}}{} }\label{app:Grover_proof_linear}

This Theorem bounds the error in using the Trotter formula to solve a search problem.
The majority of the proof uses the general form of the discrete adiabatic theorem in~\cref{lem:DAE_linear} to bound the diabatic errors. 
To use~\cref{lem:DAE_linear}, we first need to bound $c_1$ and $c_2$. 
This can be done by straightforward calculations similar to the proof of~\cref{lem:c1_c2_product_formula} and we can choose
\begin{equation}
    c_1(s) = 2|f'(s)|, \quad c_2(s) = 2 |f''(s)| + 4|f'(s)|^2. 
\end{equation}
For general $p$, by the definition of $f(s)$, we have 
\begin{equation}\label{eqn:proof_Grover_eq1}
    f'(s) = d_{N,p} \Delta_H(s;N,1,f(s))^p, 
\end{equation}
and 
\begin{align}
    f''(s) &= d_{N,p} \frac{\partial (\Delta^p)}{\partial f} f'(s) \nn
    &= d_{N,p} p \Delta_H(s;N,1,f)^{p-1} \frac{\partial \Delta}{\partial f} d_{N,p} \Delta_H(s;N,1,f(s))^p \nn
    &= p d_{N,p}^2 \Delta_H(s;N,1,f)^{2p-1} \frac{\partial \Delta}{\partial f}, 
\end{align}
where 
\begin{align}
    \left|\frac{\partial \Delta}{\partial f}\right| &= \left|\frac{\left(-2+ \frac{2}{N}\right)(1-2f(s))}{\sqrt{(1-2f(s))^2 + \frac{4}{N}f(s)(1-f(s))}}\right| \nn
    & \leq \frac{(2-2/N)|1-2(f(s))|}{\sqrt{(1-2f(s))^2}} \nn
    & \leq 2. 
\end{align}
Therefore we can choose 
\begin{equation}\label{eqn:app_Grover_def_c1_c2}
    c_1(s) = 2 d_{N,p} \Delta_H(s;N,1,f)^p, \quad c_2(s) = (4 p+4) d_{N,p}^2 \Delta_H(s;N,1,f)^{2p-1}. 
\end{equation}
Here 
\begin{equation}\label{eqn:proof_Grover_def_dNp}
    d_{N,p} = \int_0^1 \Delta_H(s;N,1,f)^{-p} df. 
\end{equation}

The error bound in~\cref{lem:DAE_linear} involves the multistep gap and the maximum/minimum over several subsequent steps. 
Here we show that, when $T$ is sufficiently large (which can be ensured by the choice of $T$ in~\cref{thm:Grover_adiabatic}), we can ignore the maximum and minimum and consider the single-step gap condition at sacrifice of extra multiplicative constant factors. 
We first focus on the minimum of the multistep gap of $W$. 
Since the two relevant eigenvalues determining the gap have different rotational directions on the circle for any time $s$, the multistep gap of $W$ is exactly the minimum single step gap over two subsequent steps, so
\begin{align}
    \check{\Delta}_{2,W}(s;M) &= \min_{s'=s+k/T, -1 \leq k \leq 3} \Delta_W(s';M)\nn
    &\geq \frac{2}{3} \min_{s'=s+\tau,|\tau|\leq 3/T} \Delta_H(s';M).
\end{align}
We bound the ratio of gaps at steps $s$ and $s'$ as 
\begin{align}
    \frac{\Delta_H(s;M)}{\Delta_H(s';M)} & = 1 + \frac{\Delta_H(s;M)-\Delta_H(s';M)}{\Delta_H(s';M)} \nonumber\\
    & \leq 1 + \frac{|s-s'|}{\Delta_H(s';M)} \max_{\xi} \left| \frac{\partial \Delta(\xi;M)}{\partial f}\right| |f'(\xi)| \nonumber\\
    & \leq 1 + \frac{6d_{N,p}}{T} \frac{\max_{s:|s-s'|\leq 3/T} \Delta_H(s;M)}{\Delta_H(s';M)}. 
\end{align}
Suppose that $T \geq 12 d_{N,p}$, then for any $|s-s'|\leq 3/T$, we have 
\begin{equation}\label{eqn:app_Grover_gap_unify_time}
    \frac{\Delta_H(s;M)}{\Delta_H(s';M)} \leq 2. 
\end{equation}
This implies that 
\begin{equation}
    \check{\Delta}_{2,W}(s;M) \geq \frac{2}{3} \min_{s'=s+\tau,|\tau|\leq 3/T} \Delta_H(s';M) \geq \frac{1}{3} \Delta_H(s;M). 
\end{equation}
Therefore we can directly replace the minimum of the multistep gap by the single-step gap at fixed time $s$ up to some constant factors. 
According to~\cref{eqn:app_Grover_def_c1_c2} and~\cref{eqn:app_Grover_gap_unify_time}, the functions $\hat{c}_1(s)$ and $\hat{c}_2(s)$ can also be replaced, up to constant factors, by $c_1(s)$ and $c_2(s)$.

We now compute each term in~\cref{lem:DAE_linear}. 
For simplicity of notation, we omit the dependence on $H$, $f$ and $N$ in the gap and use $\Delta(s;M)$ to denote $\Delta_H(s;N,M,f)$, and we ignore all constant factors in the computations. 
The boundary terms can be bounded as 
\begin{align}
    \frac{c_1(0)}{\Delta(0;M)^2} &= 2 d_{N,p}\, , \\
    \frac{c_1(1)}{\Delta(1;M)^2} &= 2 d_{N,p}\, . 
\end{align}
For the summations, we first bound them as 
\begin{align}
    \sum_{j=0}^{T-1} \frac{c_1(j/T)^2}{T\Delta(j/T;M)^3} 
    & =  \sum_{j=0}^{T-1} \frac{4d_{N,p}^2 \Delta(j/T;1)^{2p}}{T \Delta(j/T;M)^3} \leq 4 d_{N,p}^2 \sum_{j=0}^{T-1} \frac{1}{T} \frac{\Delta(j/T;1)^{2p-1}}{\Delta(j/T;M)^{2}}, 
\end{align}
and 
\begin{equation}
    \sum_{j=0}^{T-1} \frac{c_2(j/T)}{T\Delta(j/T;M)^2} = \sum_{j=0}^{T-1} \frac{(4p+4)d_{N,p}^2 \Delta(j/M;1)^{2p-1}}{T\Delta(j/T;M)^2} = 4(p+1)d_{N,p}^2 \sum_{j=0}^{T-1} \frac{1}{T} \frac{\Delta(j/T;1)^{2p-1}}{\Delta(j/T;M)^{2}}. 
\end{equation}
So the overall error is bounded by 
\begin{equation}\label{eqn:Grover_total_error}
    \mathcal{O}\left(\frac{d_{N,p}}{T} + \frac{d_{N,p}^2}{T} \sum_{j=0}^{T-1} \frac{1}{T}\frac{\Delta(j/T;1)^{2p-1}}{\Delta(j/T;M)^{2}} \right). 
\end{equation}
It remains to estimate $d_{N,p}$ and $\sum_{j=0}^{T-1} \frac{1}{T}\frac{\Delta(j/T;1)^{2p-1}}{\Delta(j/T;M)^{2}} $. 

Before we proceed, it is useful to preliminarily calculate $d_{N,p}$ for general $1 \leq p < 2$. 
By~\cref{eqn:proof_Grover_def_dNp} and the change of variables $f = \frac{1}{2} + \frac{1}{2\sqrt{N-1}} \tan \theta$, we have 
\begin{align}
    d_{N,p} &= \int_0^1 \frac{df}{\left( (1-2f)^2 + \frac{4}{N}f(1-f)\right)^{p/2}} \nn
    &= \int_0^1 \frac{df}{ \left( 4(1-1/N)(f-1/2)^2 + 1/N \right)^{p/2}}\nn
    & = \int_{-\arctan \sqrt{N-1}}^{\arctan \sqrt{N-1}} \frac{d\theta}{2\sqrt{N-1}\cos^2\theta \left(\frac{1}{N}\tan^2\theta + \frac{1}{N}\right)^{p/2}} \nn
    & = \frac{N^{p/2}}{\sqrt{N-1}} \int_0^{\arctan \sqrt{N-1}} \cos^{p-2}\theta \, d\theta\, . \label{eqn:Grover_dN_computation}
\end{align}
Here the integral in the first line comes from~\cref{eqn:proof_Grover_eq1} and the condition $f(1) = 1$.

\subsection{\texorpdfstring{$p=1$}{p=1}}

We first consider the limit case when $p = 1$. 
The error bound~\cref{eqn:Grover_total_error} becomes 
\begin{equation}\label{eqn:Grover_p1_total_error}
    \text{Error} \leq \mathcal{O}\left(\frac{d_{N,1}}{T} + \frac{d_{N,1}^2}{T} \sum_{j=0}^{T-1} \frac{1}{T}\frac{\Delta(j/T;1)}{\Delta(j/T;M)^2} \right). 
\end{equation}
For $d_{N,1}$, using~\cref{eqn:Grover_dN_computation} we have 
\begin{align}
    d_{N,1} &= \sqrt{\frac{N}{N-1}} \int_0^{\arctan \sqrt{N-1}} \frac{1}{\cos \theta} d\theta \nn
    &= \sqrt{\frac{N}{N-1}} \left. \log\left(\frac{1+\sin\theta}{\cos\theta}\right) \right|_0^{\arctan\sqrt{N-1}} \nn
    & = \sqrt{\frac{N}{N-1}} \log(\sqrt{N}+\sqrt{N-1}) \nn
    & \leq 2\log(N). \label{eqn:Grover_p1_dN_estimate}
\end{align}
For the summation, we first use the fact that $\Delta(s;M)$ is symmetric with respect to $1/2$ and monotonically decreases on $[0,1/2]$, and we can obtain 
\begin{align}
    \sum_{j=0}^{T-1} \frac{1}{T} \frac{\Delta(j/T;1)}{\Delta(j/T;M)^2} & \leq 2 \sum_{j=0}^{\lfloor T/2 \rfloor} \frac{1}{T} \frac{\Delta(j/T;1)}{\Delta(j/T;M)^2} \nn
    & = 2 \sum_{j=1}^{\lfloor T/2 \rfloor - 1} \frac{1}{T} \frac{\Delta(j/T;1)}{\Delta(j/T;M)^2} +  \frac{2}{T}\frac{\Delta(0;1)}{\Delta(0;M)^2} +  \frac{2}{T}\frac{\Delta(\lfloor T/2 \rfloor/T;1)}{\Delta(\lfloor T/2 \rfloor/T;M)^2} \nn
    & \leq 2 \sum_{j=1}^{\lfloor T/2 \rfloor - 1} \frac{1}{T} \frac{\Delta(j/T;1)}{\Delta(j/T;M)^2} +  \frac{2}{T}\frac{\Delta(0;1)}{\Delta(0;M)^2} +  \frac{2}{T}\frac{1}{\Delta(\lfloor T/2 \rfloor/T;M)} \nn
    & \leq 2 \sum_{j=1}^{\lfloor T/2 \rfloor - 1} \frac{1}{T} \frac{\Delta(j/T;1)}{\Delta(j/T;M)^2} + \frac{2}{T}\left(1 + \sqrt{\frac{N}{M}}\right). \label{eqn:Grover_eq_1}
\end{align}
Now we focus on the dominant summation. 
Using monotonicity of $\Delta(s;1)$, we have 
\begin{equation}
    2 \sum_{j=1}^{\lfloor T/2 \rfloor - 1} \frac{1}{T} \frac{\Delta(j/T;1)}{\Delta(j/T;M)^2} \leq 2 \sum_{j=1}^{\lfloor T/2 \rfloor - 1} \int_{(j-1)h}^{jh} \frac{\Delta(s;1)}{\Delta(j/T;M)^2} ds. 
\end{equation}
We change the denominator of the integrand from $\Delta(j/T;M)^2$ to $\Delta(s;M)^2$ at the cost of an extra multiplicative factor. 
Specifically, for any $s\in[(j-1)/T,j/T]$, using 
\begin{align}
    \left| \frac{d(\Delta(s;M)^2)}{ds} \right| & = \left| \left(-4+\frac{4M}{N}\right) (1-2f(s)) d_{N,1} \Delta(s;1)  \right| \nn
    & \leq \left(4-\frac{4M}{N}\right) (1-2f((j-1)/T)) d_{N,1} \Delta((j-1)/T;1) \nn
    & \leq 8\log(N) \Delta((j-1)/T;1)^2\, , 
\end{align}
we have 
\begin{align}
    \frac{\Delta((j-1)/T; M)^2}{\Delta(j/T; M)^2} & = 1 + \frac{\Delta((j-1)/T; M)^2 - \Delta(j/T; M)^2}{\Delta(j/T; M)^2} \nn
    & \leq 1 + \frac{1}{\Delta(j/T; M)^2} \frac{1}{T} \max_{s\in[(j-1)/T,j/T]} \left| \frac{d(\Delta(s;M)^2)}{ds} \right| \nn
    & \leq 1 + \frac{8\log(N)}{T} \frac{\Delta((j-1)/T;1)^2}{\Delta(j/T; M)^2}. 
\end{align}
Suppose that $T \geq 16 \log(N)$, we can conclude that 
\begin{equation}
    \frac{\Delta(s; M)^2}{\Delta(j/T; M)^2} \leq \frac{\Delta((j-1)/T; M)^2}{\Delta(j/T; M)^2} \leq 2. 
\end{equation}
Therefore~\cref{eqn:Grover_eq_1} becomes 
\begin{align}
    2 \sum_{j=1}^{\lfloor T/2 \rfloor - 1} \frac{1}{T} \frac{\Delta(j/T;1)}{\Delta(j/T;M)^2} &\leq 2 \sum_{j=1}^{\lfloor T/2 \rfloor - 1} \int_{(j-1)h}^{jh} \frac{\Delta(s;1)}{\Delta(j/T;M)^2} ds  \nn
    &\leq 4 \sum_{j=1}^{\lfloor T/2 \rfloor - 1} \int_{(j-1)h}^{jh} \frac{\Delta(s;1)}{\Delta(s;M)^2} ds \nn
    & \leq 4 \int_{0}^{1/2} \frac{\Delta(s;1)}{\Delta(s;M)^2} ds \nn
    & = 2 \int_{0}^{1} \frac{\Delta(s;1)}{\Delta(s;M)^2} ds. \label{eqn:Grover_summation2integral}
\end{align}
By changing the variable from $s$ to $f$ and using $df = d_{N,1} \Delta(s;1) ds$, we have 
\begin{align}
    2 \sum_{j=1}^{\lfloor T/2 \rfloor - 1} \frac{1}{T} \frac{\Delta(j/T;1)}{\Delta(j/T;M)^2} & \leq \frac{2}{d_{N,1}} \int_{0}^{1} \frac{df}{\Delta(s;M)^2} \nn
    & = \frac{2}{d_{N,1}} \frac{N/M}{\sqrt{N/M-1}} \arctan\sqrt{N/M-1} \nn
    & \leq \frac{\sqrt{2}\pi}{d_{N,1}} \sqrt{\frac{N}{M}}. 
\end{align}
Therefore~\cref{eqn:Grover_eq_1} becomes 
\begin{equation}\label{eqn:Grover_p1_summation_estimate}
    \sum_{j=0}^{T-1} \frac{1}{T} \frac{\Delta(j/T;1)}{\Delta(j/T;M)^2} \leq \frac{\sqrt{2}\pi}{d_{N,1}} \sqrt{\frac{N}{M}} + \frac{2}{T}\left(1 + \sqrt{\frac{N}{M}}\right). 
\end{equation}
Plugging~\cref{eqn:Grover_p1_dN_estimate} and~\cref{eqn:Grover_p1_summation_estimate} back to~\cref{eqn:Grover_p1_total_error}, we can bound the total error as
\begin{align}
    \text{Error} &\leq \mathcal{O}\left(\frac{d_{N,1}}{T} + \frac{d_{N,1}^2}{T} \sum_{j=0}^{T-1} \frac{1}{T}\frac{\Delta(j/T;1)}{\Delta(j/T;M)^2} \right) \nn
    & \leq \mathcal{O}\left( \frac{d_{N,1}}{T} + \frac{d_{N,1}}{T} \sqrt{2}\pi \sqrt{\frac{N}{M}} + \frac{2d_{N,1}^2}{T^2} \left(1 + \sqrt{\frac{N}{M}}\right) \right) \nn
    & \leq \mathcal{O}\left( \frac{\log(N)}{T} + \frac{\log(N)}{T} \sqrt{\frac{N}{M}} + \frac{\log^2(N)}{T^2}  \sqrt{\frac{N}{M}} \right) \nn
    & \leq \mathcal{O}\left( \frac{\log(N)}{T} \sqrt{\frac{N}{M}}  \right). 
\end{align}

\subsection{\texorpdfstring{$1<p<2$}{1<p<2} }

Now we consider $1<p<2$. 
Continuing with~\cref{eqn:Grover_dN_computation} and using integration by parts, we may obtain 
\begin{equation}\label{eqn:Grover_integration_by_parts}
    \int_0^{\arctan \sqrt{N-1}} \cos^{p-2} \theta\, d \theta = -\frac{1}{p-1} \cos^{p-1}\theta \sin \theta \Big|_{0}^{\arctan \sqrt{N-1}} + \frac{p}{p-1} \int_0^{\arctan \sqrt{N-1}} \cos^p \theta\, d \theta. 
\end{equation}
Both terms on the right hand side of~\cref{eqn:Grover_integration_by_parts} are bounded independent of $N$, so $\int_0^{\arctan\sqrt{N-1}} \cos^{p-2}\theta\, d \theta =\mathcal{O}(1)$ and thus, by~\cref{eqn:Grover_dN_computation}, 
\begin{equation}\label{eqn:Grover_gen_p_dN_estimate}
    d_{N,p} = \mathcal{O}\left(N^{\frac{p-1}{2}}\right). 
\end{equation}

For the summation, with the same reasoning of~\cref{eqn:Grover_eq_1} and~\cref{eqn:Grover_summation2integral}, we have 
\begin{align}
    \sum_{j=0}^{T-1} \frac{1}{T} \frac{\Delta(j/T;1)^{2p-1}}{\Delta(j/T;M)^2}  &\leq 2 \sum_{j=1}^{\lfloor T/2 \rfloor - 1} \frac{1}{T} \frac{\Delta(j/T;1)^{2p-1}}{\Delta(j/T;M)^2} + \frac{2}{T}\left(1 + \max\left\{ 1,\frac{N^{3/2-p}}{M^{3/2-p}} \right\} \right) \nn
    & \leq 2 \int_{0}^{1} \frac{\Delta(s;1)^{2p-1}}{\Delta(s;M)^2} ds + \frac{4}{T}\max\left\{ 1,\frac{N^{3/2-p}}{M^{3/2-p}} \right\} . \label{eqn:Grover_eq_3}
\end{align}
We focus on the integral. 
By changing the variable from $s$ to $f$ and using $df = d_{N,p} \Delta(s;1)^p ds$, we have 
\begin{align}
    \int_{0}^{1} \frac{\Delta(s;1)^{2p-1}}{\Delta(s;M)^2} ds &= \frac{1}{d_{N,p}} \int_0^1 \frac{\Delta(s;1)^{p-1}}{\Delta(s;M)^2} df \nn
    & = \frac{1}{d_{N,p}} \int_0^1 \frac{\left( (1-2f)^2 + \frac{4}{N}f(1-f) \right)^{(p-1)/2}}{(1-2f)^2 + \frac{4M}{N}f(1-f)} df \nn
    & = \frac{1}{d_{N,p}} \int_0^1 \frac{\left( 4\frac{N-1}{N}(f-1/2)^2 + \frac{1}{N} \right)^{(p-1)/2}}{4\frac{N-M}{N}(f-1/2)^2 + \frac{M}{N}} df. 
\end{align}
By the change of variable $f = \frac{1}{2} + \frac{\sqrt{M}}{2\sqrt{N-M}} \tan \theta$, we have 
\begin{align}
    \int_{0}^{1} \frac{\Delta(s;1)^{2p-1}}{\Delta(s;M)^2} ds &= \frac{1}{d_{N,p}} \int_{-\arctan \frac{\sqrt{N-M}}{\sqrt{M}}}^{\arctan \frac{\sqrt{N-M}}{\sqrt{M}}} \frac{\left( \frac{N-1}{N}\frac{M}{N-M}\tan^2\theta + \frac{1}{N} \right)^{(p-1)/2}}{\frac{M}{N} (\tan^2\theta + 1) } \frac{\sqrt{M}}{\sqrt{N-M}} \frac{1}{\cos^2\theta} d\theta \nn
    & = \frac{2}{d_{N,p}} \frac{N}{\sqrt{M(N-M)}}  \int_{0}^{\arctan \frac{\sqrt{N-M}}{\sqrt{M}}} \left( \frac{N-1}{N}\frac{M}{N-M}\tan^2\theta + \frac{1}{N} \right)^{(p-1)/2} d\theta. 
\end{align}
Notice that, for $\theta \in[0,\pi/2)$, the function $\tan^2\theta$ monotonically increases. 
So we can separate the integral into two parts as 
\begin{align}
    & \quad \int_{0}^{1} \frac{\Delta(s;1)^{2p-1}}{\Delta(s;M)^2} ds \nn
    &= \frac{2}{d_{N,p}} \frac{N}{\sqrt{M(N-M)}} \left( \int_{0}^{\arctan \frac{\sqrt{N-M}}{\sqrt{M(N-1)}}} + \int_{\arctan \frac{\sqrt{N-M}}{\sqrt{M(N-1)}}}^{\arctan \frac{\sqrt{N-M}}{\sqrt{M}}} \right)\left( \frac{N-1}{N}\frac{M}{N-M}\tan^2\theta + \frac{1}{N} \right)^{\frac{p-1}{2}} d\theta \nn
    & \leq \frac{2}{d_{N,p}} \frac{N}{\sqrt{M(N-M)}} \left( \int_{0}^{\arctan \frac{\sqrt{N-M}}{\sqrt{M(N-1)}}} \left(\frac{2}{N}\right)^{\frac{p-1}{2}}d\theta + \int_{\arctan \frac{\sqrt{N-M}}{\sqrt{M(N-1)}}}^{\arctan \frac{\sqrt{N-M}}{\sqrt{M}}} \left(2\frac{N-1}{N}\frac{M}{N-M}\tan^2\theta \right)^{\frac{p-1}{2}}  d\theta \right) \nn
    & = \frac{2}{d_{N,p}} \frac{N}{\sqrt{M(N-M)}} \left( \frac{2^{\frac{p-1}{2}}}{N^{\frac{p-1}{2}}} \arctan \frac{\sqrt{N-M}}{\sqrt{M(N-1)}} \right. \nn
    & \quad\quad\quad\quad\quad\quad\quad\quad\quad\quad\left. + 2^{\frac{p-1}{2}} \left(\frac{N-1}{N}\right)^{\frac{p-1}{2}} \frac{M^{\frac{p-1}{2}}}{(N-M)^{\frac{p-1}{2}}} \int_{\arctan \frac{\sqrt{N-M}}{\sqrt{M(N-1)}}}^{\arctan \frac{\sqrt{N-M}}{\sqrt{M}}} \tan^{p-1}\theta d\theta  \right)
\end{align}
Using $1<p<2$ and thus 
\begin{equation}
    \int_{\arctan \frac{\sqrt{N-M}}{\sqrt{M(N-1)}}}^{\arctan \frac{\sqrt{N-M}}{\sqrt{M}}} \tan^{p-1}\theta d\theta \leq \int_0^{\pi/2} \tan^{p-1}\theta d\theta = \frac{\pi}{2} \text{sec}\left(\frac{(p-1)\pi}{2}\right) = \mathcal{O}(1), 
\end{equation}
we have 
\begin{align}
    \int_{0}^{1} \frac{\Delta(s;1)^{2p-1}}{\Delta(s;M)^2} ds & \leq \mathcal{O}\left( \frac{1}{d_{N,p}}\frac{N}{\sqrt{M(N-M)}} \left(\frac{1}{N^{\frac{p-1}{2}}} \frac{\sqrt{N-M}}{\sqrt{M(N-1)}} + \frac{M^{\frac{p-1}{2}}}{(N-M)^{\frac{p-1}{2}}} \right)  \right) \nn
    & \leq \mathcal{O}\left( \frac{1}{d_{N,p}}\frac{\sqrt{N}}{\sqrt{M}} \left(\frac{1}{N^{\frac{p-1}{2}}\sqrt{M}} + \frac{M^{\frac{p-1}{2}}}{N^{\frac{p-1}{2}}} \right)  \right) \nn
    & \leq \mathcal{O}\left( \frac{1}{d_{N,p}} \frac{N^{1-\frac{p}{2}}}{M^{1-\frac{p}{2}}} \right). 
\end{align}
Plugging this back into~\cref{eqn:Grover_eq_3}, we have 
\begin{align}
     \sum_{j=0}^{T-1} \frac{1}{T} \frac{\Delta(j/T;1)^{2p-1}}{\Delta(j/T;M)^2} 
    & \leq \mathcal{O}\left(  \frac{1}{d_{N,p}} \frac{N^{1-\frac{p}{2}}}{M^{1-\frac{p}{2}}} + \frac{1}{T}\max\left\{ 1,\frac{N^{3/2-p}}{M^{3/2-p}} \right\} \right). \label{eqn:Grover_gen_p_summation_estimate}
\end{align}
Plugging~\cref{eqn:Grover_gen_p_dN_estimate} and~\cref{eqn:Grover_gen_p_summation_estimate} back to~\cref{eqn:Grover_total_error}, we can bound the total error by 
\begin{align}
    &\quad  \mathcal{O}\left(\frac{d_{N,p}}{T} + \frac{d_{N,p}^2}{T} \sum_{j=0}^{T-1} \frac{1}{T}\frac{\Delta(j/T;1)^{2p-1}}{\Delta(j/T;M)^{2}} \right) \nn
    & \leq \mathcal{O}\left(\frac{N^{\frac{p-1}{2}}}{T} + \frac{d_{N,p}^2}{T}\left(  \frac{1}{d_{N,p}} \frac{N^{1-\frac{p}{2}}}{M^{1-\frac{p}{2}}} + \frac{1}{T}\max\left\{ 1,\frac{N^{3/2-p}}{M^{3/2-p}} \right\} \right)  \right) \nn
    & \leq \mathcal{O}\left(\frac{N^{\frac{p-1}{2}}}{T} + \frac{1}{T} \frac{\sqrt{N}}{M^{1-\frac{p}{2}}}  + \frac{N^{p-1}}{T^2}\max\left\{ 1,\frac{N^{3/2-p}}{M^{3/2-p}} \right\}  \right) \nn
    & \leq \mathcal{O}\left(\frac{1}{T} \frac{\sqrt{N}}{M^{1-\frac{p}{2}}}  + \frac{N^{p-1}}{T^2}\max\left\{ 1,\frac{N^{3/2-p}}{M^{3/2-p}} \right\}  \right). 
\end{align}

\section{Proof of \texorpdfstring{\cref{thm:Grover_adiabatic_BC}}{} }\label{app:Grover_proof_exp}

In this section we give the complete proof of~\cref{thm:Grover_adiabatic_BC}. 
We start with summarizing the properties of the function $f(s)$ in the following lemma. 
\begin{lemma}
    Let $f(s)$ be the function defined in~\cref{eqn:Grover_scheduling_def_BC}. 
    Then 
    \begin{enumerate}
        \item $f(s) \in C^{\infty}[0,1]$, 
        \item $f(0) = 0$, $f(1) = 1$, and $f^{(k)}(0) = f^{(k)}(1) = 0$ for all $k \geq 1$, 
        \item for all $s\in[0,1/2]$ and $k \geq 1$, $f(s) + f(1-s) = 1$ and $f^{(k)}(s) = f^{(k)}(1-s)$. 
    \end{enumerate}
\end{lemma}

To prove~\cref{thm:Grover_adiabatic_BC}, we will frequently use the following lemma. 
\begin{lemma}\label{lem:Grover_BC_tech}
    Let $\Delta_H(s)$ be the gap defined in~\cref{eqn:Grover_gap_H}. 
    Then for any $s\in [0,1/2]$ and any integer $l \geq 0$, we have 
    \begin{equation}
        \frac{\exp\left(-\frac{1}{2s(1-2s)}\right)}{ \Delta_H(s) (2s(1-2s))^{l} } \leq 2^{l+3} c_e^{-l/2} (l+1)^{l+2} \left(\log(N/M)\right)^{l+2} \, 
    \end{equation}
    where $c_e$ is the normalization factor defined in~\cref{eqn:def_ce}. 
\end{lemma}
\begin{proof}
    We start with 
    \begin{align}
        \Delta_H(s) &= \sqrt{(1-2f(s))^2 + \frac{4M}{N}f(s)(1-f(s))} \nn
        & = \sqrt{(1-M/N) (1-2f(s))^2 + \frac{M}{N}} \nn
        & = \sqrt{(1-M/N) (1-g(2s))^2 + \frac{M}{N}} \nn
        & \geq \frac{1}{\sqrt{2}} \sqrt{(1-g(2s))^2 + \frac{M}{N}} \nn
        & \geq \frac{1}{2} \left( 1-g(2s) + \sqrt{\frac{M}{N}} \right). 
    \end{align}
    Let $t = 2s \in [0,1]$, then 
    \begin{equation}
        \frac{\exp\left(-\frac{1}{2s(1-2s)}\right)}{ \Delta_H(s) (2s(1-2s))^{l} } \leq \frac{2\exp\left(-\frac{1}{t(1-t)}\right)}{ \left(1-g(t) + \sqrt{\frac{M}{N}}\right) (t(1-t))^{l} } = 2F(t)\, ,
    \end{equation}
    where we define 
    \begin{equation}
        F(t) := \frac{\exp\left(-\frac{1}{t(1-t)}\right)}{ \left(1-g(t) + \sqrt{\frac{M}{N}}\right) (t(1-t))^{l} }. 
    \end{equation}
    It suffices to bound $F(t)$. 
    This can be done with nuanced modifications of the proof of~\cite[Lemma 14]{an2019quantum}, but for completeness we provide the details again. 
    Since $F(0) = F(1) = 0$ and $F(t) > 0$ for $t \in (0,1)$, there exists a $t_* \in (0,1)$ such that $F(t_*)$ reaches its maximum. 
    At this point, $F'(t_*) = 0$, where $F'(t)$ can be computed as 
    \begin{align}
        &\quad \left(1-g(t) + \sqrt{\frac{M}{N}}\right)^2 (t(1-t))^{2l} F'(t) \nn
        & = \left(1-g(t) + \sqrt{\frac{M}{N}}\right) (t(1-t))^{l} \exp\left(-\frac{1}{t(1-t)}\right) \frac{1-2t}{(t(1-t))^2} \nn
        & \quad - \exp\left(-\frac{1}{t(1-t)}\right) \left(-c_e^{-1} \exp\left(-\frac{1}{t(1-t)}\right) (t(1-t))^l + \left(1-g(t) + \sqrt{\frac{M}{N}}\right) l (t(1-t))^{l-1}(1-2t)  \right) \nn
        & = \exp\left(-\frac{1}{t(1-t)}\right) (t(1-t))^{l-2} G(t), 
    \end{align}
    where 
    \begin{equation}
        G(t) = \left(1-g(t) + \sqrt{\frac{M}{N}}\right)(1-2t) (1-lt(1-t)) + c_e^{-1} \exp\left(-\frac{1}{t(1-t)}\right) (t(1-t))^2. 
    \end{equation}
    For $t \in [1 - \frac{\sqrt{c_e}}{(l+1)\log(N/M)}, 1]$ where $\sqrt{c_e} \approx 0.08$, we have 
    \begin{equation}
        1-g(t) + \sqrt{\frac{M}{N}} \geq \sqrt{\frac{M}{N}}, 
    \end{equation}
    \begin{equation}
        1 - 2t \leq \frac{2\sqrt{c_e}}{(l+1)\log(N/M)} - 1 \leq \frac{2\sqrt{c_e}}{\log 2} - 1 , 
    \end{equation}
    and 
    \begin{equation}
        1-lt(1-t) \geq 1 - l(1-t) \geq 1 - \frac{\sqrt{c_e}l}{(l+1)\log(N/M)} \geq 1 - \frac{\sqrt{c_e}}{\log(2)}. 
    \end{equation}
    Then 
    \begin{equation}
        \left(1-g(t) + \sqrt{\frac{M}{N}}\right)(1-2t) (1-lt(1-t)) \leq - \left(1-\frac{2\sqrt{c_e}}{\log 2}\right)\left(1-\frac{\sqrt{c_e}}{\log 2}\right) \sqrt{\frac{M}{N}} \leq -\frac{3}{5} \sqrt{\frac{M}{N}}. 
    \end{equation}
    Meanwhile, for the same range of $t$, 
    \begin{align}
        c_e^{-1} \exp\left(-\frac{1}{t(1-t)}\right) (t(1-t))^2 & \leq c_e^{-1} \exp\left(-\frac{1}{t(1-t)}\right) (1-t)^2 \nn
        & \leq \frac{1}{c_e} \exp\left(-\left(1 - \frac{\sqrt{c_e}}{(l+1)\log(N/M)}\right)^{-1}\frac{(l+1)\log(N/M)}{\sqrt{c_e}}\right) \left(\frac{\sqrt{c_e}}{(l+1)\log(N/M)}\right)^2 \nn
        & \leq \frac{1}{(\log 2)^2} \exp\left(-\left(1 - \frac{\sqrt{c_e}}{(l+1)\log(N/M)}\right)^{-1}\frac{(l+1)\log(N/M)}{\sqrt{c_e}}\right) \nn
        & = \frac{1}{(\log 2)^2} \left(\frac{N}{M}\right)^{-\left(1 - \frac{\sqrt{c_e}}{(l+1)\log(N/M)}\right)^{-1} \frac{l+1}{\sqrt{c_e}}} \nn
        & \leq \frac{1}{(\log 2)^2} \left(\frac{N}{M}\right)^{-\frac{l+1}{\sqrt{c_e}}} \nn
        & \leq \frac{1}{(\log 2)^2} \left(\frac{N}{M}\right)^{-10} \nn
        & \leq \frac{1}{2^7} \sqrt{\frac{M}{N}}. 
    \end{align}
    Therefore, for $t \in [1 - \frac{\sqrt{c_e}}{(l+1)\log(N/M)}, 1]$, we always have 
    \begin{equation}\label{eqn:Grover_BC_eq1}
        G(t) \leq -\frac{3}{5} \sqrt{\frac{M}{N}} + \frac{1}{2^7} \sqrt{\frac{M}{N}} < 0. 
    \end{equation}
    Since $F'(t_*) = 0$, we must have $G(t_*) = 0$, so~\cref{eqn:Grover_BC_eq1} implies that $t_* \leq 1 - \frac{\sqrt{c_e}}{(l+1)\log(N/M)}$. 
    Meanwhile, notice that for $t\in[0,1/2)$, we have $F(t) < F(1-t)$, so $t_* \geq 1/2$. 
    Therefore we obtain the range of $t_*$ that 
    \begin{equation}
        t_* \in \left[\frac{1}{2}, 1 - \frac{\sqrt{c_e}}{(l+1)\log(N/M)}\right] . 
    \end{equation}
    We now bound $F(t)$. 
    Using $G(t_*) = 0$ again, we have 
    \begin{equation}
        \frac{\exp\left(-\frac{1}{t_*(1-t_*)}\right)}{ 1-g(t_*) + \sqrt{\frac{M}{N}}} = \frac{c_e(2t_*-1)(1-lt_*(1-t_*))}{(t_*(1-t_*))^2}. 
    \end{equation}
    Therefore 
    \begin{align}
        F(t) & \leq F(t_*) \nn
        & = \frac{\exp\left(-\frac{1}{t_*(1-t_*)}\right)}{ \left(1-g(t_*) + \sqrt{\frac{M}{N}}\right) (t_*(1-t_*))^{l} } \nn
        & = \frac{c_e(2t_*-1)(1-lt_*(1-t_*))}{(t_*(1-t_*))^{l+2}} \nn
        & \leq \frac{c_e}{(t_*(1-t_*))^{l+2}} \nn
        & \leq \frac{2^{l+2}c_e}{(1-t_*)^{l+2}} \nn
        & \leq 2^{l+2}c_e \left( \frac{(l+1)\log(N/M)}{\sqrt{c_e}} \right)^{l+2}  \nn
        & = 2^{l+2} c_e^{-l/2} (l+1)^{l+2} \left(\log(N/M)\right)^{l+2}. 
    \end{align}
\end{proof}

We are ready to present the main proof of~\cref{thm:Grover_adiabatic_BC}. 
The idea is separately using~\cref{lem:DAE_linear} and~\cref{lem:DAE_exp} to obtain two independent error bounds, so the overall error is bounded by the minimum of the two. 

\begin{proof}[Proof of~\cref{thm:Grover_adiabatic_BC}]
    By differentiating the product walk operator, we may choose
    \begin{equation}\label{eqn:app_Grover_exp_def_c1_c2}
        c_1(s) = 2f'(s), \quad c_2(s) = 2 |f''(s)| + 4f'(s)^2. 
    \end{equation}
    We use the similar trick as in~\cref{app:Grover_proof_linear} to get rid of the hat and check notations and the multistep gap in~\cref{lem:DAE_linear}. 
    Specifically, for $|s-s'|\leq 3/T$, we have  
    \begin{align}
    \frac{\Delta_H(s;M)}{\Delta_H(s';M)} & = 1 + \frac{\Delta_H(s;M)-\Delta_H(s';M)}{\Delta_H(s';M)} \nn
    & \leq 1 + \frac{|s-s'|}{\Delta_H(s';M)} \max_{\xi} \left| \frac{\partial \Delta(\xi;M)}{\partial f}\right| |f'(\xi)| \nn
    & \leq 1 + \frac{6}{T \Delta_H(s';M)} \max  |f' | \nn
    & \leq 1 + \frac{6 \max  |f' | }{T \sqrt{M/N}}. 
    \end{align}
    Here $\max |f'|$ is an absolute constant which is independent of $N$ and $M$, so we may choose $T \geq 6 \sqrt{N/M} \max|f'|$ and thus 
    \begin{equation}
        \frac{\Delta_H(s;M)}{\Delta_H(s';M)} \leq 2. 
    \end{equation}
    So we can bound $\check{\Delta}_{2,W}(s;M)$ from below by $\frac{1}{3}\Delta_H(s;M)$ and drop the hat notations in $c_1$ and $c_2$ in the same way as in~\cref{app:Grover_proof_linear}. 
    
    Now we bound each term in the upper bound of~\cref{lem:DAE_linear}. 
    The two boundary terms are $0$ due to the boundary cancellation. 
    For the summations, we use symmetry of $f(s)$ and~\cref{lem:Grover_BC_tech} and obtain 
    \begin{align}
        \sum_{j=0}^{T-1} \frac{c_1(j/T)^2}{T \Delta_W(j/T;M)^3} & \leq \frac{27}{2} \sum_{j=0}^{T-1} \frac{f'(j/T)^2}{T \Delta_H(j/T;M)^3} \nn
        & = 27 \sum_{j=0}^{\lfloor T/2 \rfloor} \frac{g'(2j/T)^2}{T \Delta_H(j/T;M)^3} \nn
        & = \frac{27}{c_e^2} \sum_{j=0}^{\lfloor T/2 \rfloor} \frac{1}{T \Delta_H(j/T;M)} \left( \frac{\exp\left(-\frac{1}{2j/T(1-2j/T)}\right)}{ \Delta_H(j/T;M) } \right)^2 \nn
        & \leq \frac{27}{c_e^2} \sum_{j=0}^{\lfloor T/2 \rfloor} \frac{1}{T \sqrt{M/N}} 2^{6}  \left(\log(N/M)\right)^{4} \nn
        & = \mathcal{O} \left( \sqrt{\frac{N}{M}} \left(\log\frac{N}{M}\right)^4  \right), 
    \end{align}
    and 
    \begin{align}
        & \quad \sum_{j=0}^{T-1} \frac{c_2(j/T)}{T \Delta_W(j/T;M)^2} \nn
        & \leq  \frac{9}{2} \sum_{j=0}^{T-1} \frac{|f''(j/T)|}{T \Delta_H(j/T;M)^2}  + 9 \sum_{j=0}^{T-1} \frac{f'(j/T)^2}{T \Delta_H(j/T;M)^2} \nn
        & = 18 \sum_{j=0}^{\lfloor T/2 \rfloor} \frac{|g''(2j/T)|}{T \Delta_H(j/T;M)^2}  + 18 \sum_{j=0}^{\lfloor T/2 \rfloor} \frac{g'(2j/T)^2}{T \Delta_H(j/T;M)^2} \nn
        & = \frac{18}{c_e} \sum_{j=0}^{\lfloor T/2 \rfloor} \frac{\exp\left(-\frac{1}{2j/T(1-2j/T)}\right) |1-4j/T|}{T \Delta_H(j/T;M)^2 (2j/T(1-2j/T))^2 }  + \frac{18}{c_e^2} \sum_{j=0}^{\lfloor T/2 \rfloor} \frac{\left(\exp\left(-\frac{1}{2j/T(1-2j/T)}\right)\right)^2}{T \Delta_H(j/T;M)^2} \nn
        & \leq \frac{18}{c_e} \sum_{j=0}^{\lfloor T/2 \rfloor} \frac{1 }{T \Delta_H(j/T;M) } \frac{\exp\left(-\frac{1}{2j/T(1-2j/T)}\right)}{\Delta_H(j/T;M) (2j/T(1-2j/T))^2} + \frac{18}{c_e^2} \sum_{j=0}^{\lfloor T/2 \rfloor} \frac{1}{T}\left(\frac{\left(\exp\left(-\frac{1}{2j/T(1-2j/T)}\right)\right)}{\Delta_H(j/T;M)}\right)^2 \nn
        & \leq \frac{18}{c_e} \sum_{j=0}^{\lfloor T/2 \rfloor} \frac{1 }{T \sqrt{M/N} } 2^{5} c_e^{-1} 3^{4} \left(\log(N/M)\right)^{4}  + \frac{18}{c_e^2} \sum_{j=0}^{\lfloor T/2 \rfloor} \frac{1}{T} 2^{6}  \left(\log(N/M)\right)^{4} \nn
        & = \mathcal{O} \left( \sqrt{\frac{N}{M}} \left(\log\frac{N}{M}\right)^4  \right). 
    \end{align}
    Therefore~\cref{lem:DAE_linear} implies an error bound 
    \begin{equation}\label{eqn:Grover_BC_bound_linear}
        \mathcal{O} \left(\frac{1}{T} \sqrt{\frac{N}{M}} \left(\log\frac{N}{M}\right)^4  \right). 
    \end{equation}
    
    Meanwhile, since the scheduling function $f(s)$ satisfies the boundary cancellation condition,~\cref{lem:DAE_exp} implies another error bound 
    \begin{equation}\label{eqn:Grover_BC_bound_exp}
        \frac{C_k(N,M)}{T^k}
    \end{equation}
    for any integer $k \geq 1$. 
    Here the constant $C_k(N,M)$ which is independent of $T$ but may depend on $k,N,M$. 
    Therefore the overall error can be bounded by the minimum of~\cref{eqn:Grover_BC_bound_linear} and~\cref{eqn:Grover_BC_bound_exp}. 
\end{proof}

\end{document}